\documentclass[10pt,twocolumn]{article}
\usepackage{latex8}
\usepackage[utf8]{inputenc}
\usepackage[T1]{fontenc}
\usepackage{lmodern}
\usepackage[english]{babel}
\usepackage[unicode]{hyperref}
\usepackage{graphicx}
\usepackage{tikz}
\usepackage{stmaryrd}
\usepackage{slashbox}
\usepackage{mathpartir}
\usepackage{amsmath}
\usepackage{amsthm}
\usepackage[all,2cell,ps]{xy}
\xyoption{curve}
\UseAllTwocells
\SilentMatrices
\CompileMatrices
\usepackage{times}

\newlength{\reduce}
\newenvironment{mymath}{\vspace\reduce\vspace{-1mm}\begin{displaymath}}{\vspace\reduce\vspace{-1mm}\end{displaymath}\hspace{-0.6ex}}
\newenvironment{myequation}{\vspace{-1mm}\vspace\reduce\begin{equation}}{\vspace\reduce\vspace{-1mm}\end{equation}\hspace{-0.6ex}}
\newcommand{\catquot}[1]{#1/\!\!\!\equiv}

\newcommand{\sstrid}[1]{\input{#1.tex}}
\newcommand{\svxym}[1]{\vcenter{\xymatrix@C=7ex@R=7ex{#1}}}

\newcommand{\axioms}{\text{\emph{Ax}}}
\newcommand{\sgramor}{\ \ \ |\ \ \ }

\allowdisplaybreaks[1]

\usepackage{shortcuts}
\renewcommand{\wrt}{wrt\xspace}

\title{The Structure of First-Order Causality}

\hypersetup{
  pdftitle={\csname @title\endcsname},
  pdfauthor={Samuel Mimram},
  unicode=true,
  colorlinks=true,
  linkcolor=black,
  citecolor=black,
  urlcolor=black
}

\author{Samuel Mimram\\
  CNRS -- Université Paris Diderot\\
  \href{mailto:samuel.mimram@ens-lyon.org}{\texttt{samuel.mimram@pps.jussieu.fr}}
}

\newcommand{\FMR}{\mathbf{MRel}}

\newcommand{\Games}{\mathbf{Games}}
\newcommand{\Alg}[2]{\mathbf{Alg}_{#1}^{#2}}
\newcommand{\intset}[1]{\underline{#1}}

\newcommand{\moves}[1]{M_{#1}}

\newcommand{\qforall}[1]{\forall{#1}.}
\newcommand{\qexists}[1]{\exists{#1}.}

\newtheorem{definition}{Definition}

\newtheorem{lemma}[definition]{Lemma}
\newtheorem{property}[definition]{Property}

\newtheorem{theorem}[definition]{Theorem}

\renewcommand{\paragraph}[1]{\bigskip\noindent\textbf{#1}\;}

\newcommand{\eqth}[1]{\mathfrak{#1}}
\newcommand{\intp}[1]{\llbracket{#1}\rrbracket}
\newcommand{\represents}[1]{\widetilde{#1}}
\newcommand{\card}[1]{\left|#1\right|}
\newcommand{\size}[1]{\left|#1\right|}
\newcommand{\lrule}[1]{\text{(#1)}}
\newcommand{\before}{\varolessthan}

\begin{document}
\maketitle

\begin{abstract}
  Game semantics describe the interactive behavior of proofs by interpreting
  formulas as games on which proofs induce strategies. Such a semantics is
  introduced here for capturing dependencies induced by quantifications in
  first-order propositional logic. One of the main difficulties that has to be
  faced during the elaboration of this kind of semantics is to characterize
  definable strategies, that is strategies which actually behave like a
  proof. This is usually done by restricting the model to strategies satisfying
  subtle combinatorial conditions, whose preservation under composition is often
  difficult to show. Here, we present an original methodology to achieve this
  task, which requires to combine advanced tools from game semantics, rewriting
  theory and categorical algebra. We introduce a diagrammatic presentation of
  the monoidal category of definable strategies of our model, by the means of
  generators and relations: those strategies can be generated from a finite set
  of atomic strategies and the equality between strategies admits a finite
  axiomatization, this equational structure corresponding to a polarized
  variation of the notion of bialgebra. This work thus bridges algebra and
  denotational semantics in order to reveal the structure of dependencies
  induced by first-order quantifiers, and lays the foundations for a mechanized
  analysis of causality in programming languages.
\end{abstract}

Denotational semantics were introduced to provide useful abstract invariants of
proofs and programs modulo cut-elimination or reduction. In particular, game
semantics, introduced in the nineties, have been very successful in capturing
precisely the interactive behavior of programs. In these semantics, every type
is interpreted as a \emph{game} (that is as a set of \emph{moves} that can be
played during the game) together with the rules of the game (formalized by a
partial order on the moves of the game indicating the dependencies between
them). Every move is to be played by one of the two players, called
\emph{Proponent} and \emph{Opponent}, who should be thought respectively as the
program and its environment.
A program is characterized by the sequences of moves that it can exchange with
its environment during an execution and thus defines a \emph{strategy}
reflecting the interactive behavior of the program inside the game specified by
the type of the program.

The notion of \emph{pointer game}, introduced by Hyland and
Ong~\cite{hyland-ong:full-abstraction-pcf}, gave one of the first fully abstract
models of PCF (a simply-typed \hbox{$\lambda$-calculus} extended with recursion,
conditional branching and arithmetical constants). It has revealed that PCF
programs generate strategies with partial memory, called \emph{innocent} because
they react to Opponent moves according to their own \emph{view} of the
play. Innocence
is in this setting the main ingredient to characterize \emph{definable}
strategies, that is strategies which are the interpretation of a PCF term,
because it describes the behavior of the purely functional core of the language
(\ie $\lambda$-terms), which also corresponds to proofs in propositional
logic. This seminal work has lead to an extremely successful series of
semantics: by relaxing in various ways the innocence constraint on strategies,
it became suddenly possible to generalize this characterization to PCF programs
extended with imperative features such as references, control, non-determinism,
etc.

Unfortunately, these constraints are quite specific to game semantics and remain
difficult to link with other areas of computer science or algebra. They are
moreover very subtle and combinatorial and thus sometimes difficult to work
with. This work is an attempt to find new ways to describe the behavior of
proofs.


\paragraph{Generating instead of restricting.}
In this paper, we introduce a game semantics capturing dependencies induced by
quantifiers in first-order propositional logic, forming a strict monoidal
category called $\Games$. Instead of characterizing definable strategies of the
model by restricting to strategies satisfying particular conditions, we show
here that we can equivalently use here a kind of converse approach. We show how
to \emph{generate} definable strategies by giving a \emph{presentation} of those
strategies: a finite set of definable strategies can be used to generate all
definable strategies by composition and tensoring, and the equality between
strategies obtained this way can be finitely axiomatized.

What we mean precisely by a presentation is a generalization of the usual notion
of presentation of a monoid to monoidal categories. For example, consider the
additive monoid $\mathbb{N}^2=\mathbb{N}\times\mathbb{N}$. It admits the
presentation~$\pangle{\;p,q\;|\;qp=pq\;}$, where $p$ and $q$ are two
\emph{generators} and $qp=pq$ is a relation between two elements of the free
monoid $M$ on $\{p,q\}$. This means that $\mathbb{N}^2$ is isomorphic to the
free monoid $M$ on the two generators, quotiented by the smallest congruence
$\equiv$ (\wrt{} multiplication) such that $qp\equiv pq$.
More generally, a (strict) monoidal category~$\mathcal{C}$ (such as~$\Games$)
can be presented by a \emph{polygraph}, consisting of typed generators in
dimension 1 and~2 and relations in dimension 3, such that the
category~$\mathcal{C}$ is monoidally equivalent to the free monoidal category on
the generators, quotiented by the congruence generated by the relations.


\paragraph{Reasoning locally.}
The usefulness of our construction is both theoretic and practical. It reveals
that the essential algebraic structure of dependencies induced by quantifiers is
a polarized variation of the well-known structure of bialgebra, thus bridging
game semantics and algebra. It also proves very useful from a technical point of
view: this presentation allows us to reason locally about strategies. In
particular, it enables us to deduce a posteriori that the strategies of the
category~$\Games$ are definable (one only needs to check that generators are
definable) and that these strategies actually compose, which is not
trivial. Finally, the presentation gives a finite description of the category,
that we can hope to manipulate with a computer, paving the way for a series of
new tools to automate the study of semantics of programming languages.

\paragraph{A game semantics capturing first-order causality.}
Game semantics has revealed that proofs in logic describe particular strategies
to explore formulas. Namely, a formula $A$ is a syntactic tree expressing in
which order its connectives must be introduced in cut-free proofs of $A$. In
this sense, it can be seen as the rules of a game whose moves correspond to
connectives. For instance, consider a formula of the form
\vspace{-1.6ex}
\begin{equation}
  \label{eq:ex-formula}
  \qforall x P\quad\Rightarrow\quad\qforall y\qexists z Q
  \vspace{-0.3ex}
\end{equation}
%
where $P$ and $Q$ are propositional formulas which may contain free
variables. When searching for a proof of~\eqref{eq:ex-formula}, the $\forall y$
quantification must be introduced before the $\exists z$ quantification, and the
$\forall x$ quantification can be introduced independently. Here, introducing an
existential quantification should be thought as playing a Proponent move (the
strategy gives a witness for which the formula holds) and introducing an
universal quantification as playing an Opponent move (the strategy receives a
term from its environment, for which it has to show that the formula holds). So,
the game associated to the formula~\eqref{eq:ex-formula} will be the partial
order on the first-order quantifications appearing in the formula, depicted
below (to be read from the top to the bottom): \vspace\reduce
\[
  \xymatrix@R=4ex@C=4ex{
    \forall x&\ar@{-}[d]\forall y\\
    &\exists z\\
  }
\]


To understand exactly which dependencies induced by proofs are interesting, we
shall examine proofs of the formula $\qexists x P\Rightarrow\qexists y Q$,
which induces the following game:
\vspace{-2ex}
\[
\xymatrix{
  \exists x&\exists y\\
}
\vspace\reduce
\]
By permuting the order of introduction rules, the proof of this formula on the
left-hand side of
\begin{mymath}
\inferrule{
\inferrule{
\inferrule{\pi}
{P\vdash Q[t/y]}
}
{P\vdash \qexists y Q}
}
{\qexists x P\vdash \qexists y Q}
\qquad
\rightsquigarrow
\qquad
\inferrule{
\inferrule{
\inferrule{\pi}
{P\vdash Q[t/y]}
}
{\qexists x P\vdash Q[t/y]}
}
{\qexists x P\vdash \qexists y Q}
\end{mymath}
might be reorganized as the proof on the right-hand side if and only if the term
$t$ used in the introduction rule of the $\exists y$ connective does not have
$x$ as free variable. If the variable~$x$ is free in $t$ then the rule
introducing $\exists y$ can only be used after the rule introducing the $\exists
x$ connective. In this case, it will be reflected by a causal dependency in the
strategy corresponding to the proof, depicted by an oriented wire:
\vspace\reduce
\[
\begin{tikzpicture}[xscale=0.40,yscale=0.40]
\useasboundingbox (-0.5,-0.5) rectangle (4.5,0.5);
\draw[,] (1.00,0.00) -- (3.00,0.00);
\draw (1.85,0.10) -- (2.15,0.00);
\draw (1.85,-0.10) -- (2.15,0.00);
\draw (0.00,0.00) node{$\exists x$};
\draw (4.00,0.00) node{$\exists y$};
\end{tikzpicture}

\vspace{-1mm}
\vspace\reduce
\]
and we sometimes say that the move~$\exists x$ \emph{justifies} the
move~$\exists y$. A simple further study of permutability of introduction rules
of first-order quantifiers shows that this is the only kind of relevant
dependencies.
These permutations of rules where the motivation for the introduction of
non-alternating asynchronous game semantics~\cite{mellies-mimram:ag5}. However,
we focus here on causality and define strategies by the dependencies they induce
on moves.

We thus build a strict monoidal category whose objects are games and whose
morphisms are strategies, in which we can interpret formulas and proofs in
first-order propositional logic, and write~$\Games$ for the subcategory of
definable strategies. This paper is devoted to the construction of a
presentation for this category. We introduce formally the notion of presentation
of a monoidal category in Section~\ref{sec:pres} and recall some useful
classical algebraic structures in Section~\ref{sec:alg-struct}. Then, we give a
presentation of the category of relations in
Section~\ref{section:presentation-rel} and extend this presentation to the
category~$\Games$, that we define formally in
Section~\ref{section:games-strategies}.


\section{Presentations of monoidal categories}
\label{sec:pres}
For the lack of space, we don't recall here the basic definitions in category
theory, such as the definition of monoidal categories. The interested reader can
find a presentation of these concepts in MacLane's reference
book~\cite{maclane:cwm}.

\paragraph{Monoidal theories.}
A \emph{monoidal theory} $\mathbb{T}$ is a strict monoidal category whose
objects are the natural integers, such that the tensor product on objects is the
addition of integers. By an integer $\intset{n}$, we mean here the finite
ordinal $\intset{n}=\{0,1,\ldots,n-1\}$ and the addition is given by
$\intset{m}+\intset{n}=\intset{m+n}$. An \emph{algebra} $F$ of a monoidal theory
$\mathbb{T}$ in a strict monoidal category $\mathcal{C}$ is a strict monoidal
functor from~$\mathbb{T}$ to $\mathcal{C}$; we
write~$\Alg{\mathbb{T}}{\mathcal{C}}$ for the category of algebras from
$\mathbb{T}$ to $\mathcal{C}$ and monoidal natural transformations between
them. Monoidal theories are sometimes called PRO, this terminology was
introduced by Mac Lane in~\cite{maclane:ca} as an abbreviation for ``category
with products''. They generalize equational theories (or Lawere
theories~\cite{lawvere:phd}) in the sense that operations are typed and can
moreover have multiple outputs as well as multiple inputs, and are not
necessarily cartesian but only monoidal.

\paragraph{Presentations of monoidal categories.}
\label{subsection:moncat-presentation}
We now recall the notion of \emph{presentation} of a monoidal category by the
means of typed 1- and 2-dimensional generators and relations.

Suppose that we are given a set $E_1$ whose elements are called \emph{atomic
  types}. We write $E_1^*$ for the free monoid on the set $E_1$ and $i_1:E_1\to
E_1^*$ for the corresponding injection; the product of this monoid is written
$\otimes$.
The elements of~$E_1^*$ are called \emph{types}. Suppose moreover that we are
given a set $E_2$, whose elements are called \emph{generators}, together with
two functions $s_1,t_1:E_2\to E_1^*$, which to every generator associate a type
called respectively its \emph{source} and \emph{target}. We call a
\emph{signature} such a 4-uple $(E_1,s_1,t_1,E_2)$:
%
\vspace{-0.5ex}
\begin{mymath}
\svxym{
  E_1\ar[r]^{i_1}&E_1^*&\ar@<-0.7ex>[l]_{s_1}\ar@<0.7ex>[l]^{t_1}E_2\\
}
\vspace{-0.5ex}
\end{mymath}
\noindent
Every such signature $(E_1,s_1,t_1,E_2)$ generates a free strict monoidal
category $\mathcal{E}$, whose objects are the elements of~$E_1^*$ and whose
morphisms are formal composite and tensor products of elements of $E_2$,
quotiented by suitable laws imposing associativity of composition and tensor and
compatibility of composition with tensor, see~\cite{burroni:higher-word}. If we
write $E_2^*$ for the morphisms of this category and \hbox{$i_2:E_2\to E_2^*$}
for the injection of the generators into this category, we get a diagram
\begin{mymath}
\svxym{
  E_1\ar[d]_{i_1}&\ar@<-0.7ex>[dl]_{s_1}\ar@<0.7ex>[dl]^{t_1}E_2\ar[d]_{i_2}\\
  E_1^*&\ar@<-0.7ex>[l]_{\overline{s_1}}\ar@<0.7ex>[l]^{\overline{t_1}}E_2^*\\
}
\end{mymath}
in $\Set$ together with a structure of monoidal category~$\mathcal{E}$ on the
graph
\vspace{-1ex}
\begin{mymath}
\xymatrix@C=10ex@R=10ex{
  E_1^*&\ar@<-0.7ex>[l]_{\overline{s_1}}\ar@<0.7ex>[l]^{\overline{t_1}}E_2^*\\
}
\end{mymath}
where the morphisms $\overline{s_1},\overline{t_1}:E_2^*\to E_1^*$ are the
morphisms (unique by universality of $E_2^*$) such that
\hbox{$s_1=\overline{s_1}\circ i_2$} and \hbox{$t_1=\overline{t_1}\circ i_2$}.
The \emph{size} $\size{f}$ of a morphism $f:A\to B$ in $E_2^*$ is defined
inductively by
\begin{mymath}
\begin{array}{r@{ = }l@{\qquad}r@{ = }l}
  \size{\id}&0
  &
  \size{f}&1
  \text{\quad if $f$ is a generator}
  \\
  \size{f_1\otimes f_2}
  &
  \size{f_1}+\size{f_2}
  &
  \size{f_2\circ f_1}
  &
  \size{f_1}+\size{f_2}
\end{array}
\end{mymath}
In particular, a morphism is of size $0$ if and only if it is an identity.

Our constructions are a particular case of Burroni's
polygraphs~\cite{burroni:higher-word} (and Street's
2-computads~\cite{street:limit-indexed-by-functors}) who made precise the sense
in which the generated monoidal category is free on the signature. In
particular, the following notion of equational theory is a specialization of the
definition of a 3-polygraph to the case where there is only one 0-cell.


\begin{definition}
  A \textbf{monoidal equational theory} is a 7-uple
  \begin{mymath}
  \eqth{E}=(E_1,s_1,t_1,E_2,s_2,t_2,E_3)
  \end{mymath}
  where $(E_1,s_1,t_1,E_2)$ is a signature together with a set $E_3$ of
  \emph{relations} and two morphisms $s_2,t_2:E_3\to E_2^*$, as pictured in the
  diagram
  \begin{mymath}
  \svxym{
    E_1\ar[d]_{i_1}&\ar@<-0.7ex>[dl]_{s_1}\ar@<0.7ex>[dl]^{t_1}E_2\ar[d]_{i_2}&\ar@<-0.7ex>[dl]_{s_2}\ar@<0.7ex>[dl]^{t_2}E_3\\
    E_1^*&\ar@<-0.7ex>[l]_{\overline{s_1}}\ar@<0.7ex>[l]^{\overline{t_1}}E_2^*\\
  }
  \end{mymath}
  such that $\overline{s_1}\circ s_2=\overline{s_1}\circ t_2$ and
  $\overline{t_1}\circ s_2=\overline{t_1}\circ t_2$.
\end{definition}
Every equational theory defines a monoidal category
$\mathbb{E}=\catquot{\mathcal{E}}$ obtained from the monoidal category
$\mathcal{E}$ generated by the signature $(E_1,s_1,t_1,E_2)$ by quotienting the
morphisms by the congruence $\equiv$ generated by the relations of the
equational theory $\eqth{E}$: it is the smallest congruence (\wrt{} both
composition and tensoring) such that $s_2(e)\equiv t_2(e)$ for every element $e$
of $E_3$. We say that a monoidal equational theory $\eqth{E}$ is a
\emph{presentation} of a strict monoidal category $\mathcal{M}$ when
$\mathcal{M}$ is monoidally equivalent to the category $\mathbb{E}$ generated by
$\eqth{E}$.

We sometimes informally say that an equational theory has a \emph{generator}
$f:A\to B$ to mean that $f$ is an element of $E_2$ such that $s_1(f)=A$ and
$t_1(f)=B$. We also say that the equational theory has a \emph{relation} $f=g$
to mean that there exists an element $e$ of $E_3$ such that $s_2(e)=f$
and~$t_2(e)=g$.



It is remarkable that every monoidal equational theory
$(E_1,s_1,t_1,E_2,s_2,t_2,E_3)$ where the set~$E_1$ is reduced to only one
object $\{1\}$ generates a monoidal category which is a monoidal theory ($\N$ is
the free monoid on one object), thus giving a notion of presentation of those
categories.

\paragraph{Presented categories as models.}
Suppose that a strict monoidal category $\mathcal{M}$ is presented by an
equational theory $\eqth{E}$, generating a category
$\mathbb{E}=\catquot{\mathcal{E}}$. The proof that $\eqth{E}$
presents~$\mathcal{M}$ can generally be decomposed in two parts:
\begin{enumerate}
\item \emph{$\mathcal{M}$ is a model of the equational theory $\eqth{E}$}: there
  exists a functor $\represents{-}$ from the category $\mathbb{E}$ to
  $\mathcal{M}$. This amounts to checking that there exists a functor
  $F:\mathcal{E}\to\mathcal{M}$ such that for all morphisms \hbox{$f,g:A\to
    B$} in $\mathcal{E}$, $f\equiv g$ implies $Ff=Fg$.
\item \emph{$\mathcal{M}$ is a fully-complete model of the equational theory
    $\eqth{E}$}: the functor $\represents{-}$ is full and faithful.
\end{enumerate}
We say that a morphism $f:A\to B$ of $\mathbb{E}$ \emph{represents} the morphism
$\represents{f}:\represents{A}\to\represents{B}$ of $\mathcal{M}$.

Usually, the first point is a straightforward verification. Proving that the
functor $\represents{-}$ is full and faithful often requires more work. In this
paper, we use the methodology introduced by Burroni~\cite{burroni:higher-word}
and refined by Lafont~\cite{lafont:boolean-circuits}. We first define
\emph{canonical forms} which are canonical representatives of the equivalence
classes of morphisms of $\mathcal{E}$ under the congruence $\equiv$ generated by
the relations of $\eqth{E}$. Proving that every morphism is equal to a canonical
form can be done by induction on the size of the morphisms. Then, we show that
the functor $\represents{-}$ is full and faithful by showing that the canonical
forms are in bijection with the morphisms of~$\mathcal{M}$.

It should be noted that this is not the only technique to prove that an
equational theory presents a monoidal category. In particular, Joyal and Street
have used topological methods~\cite{joyal-street:geometry-tensor-calculus} by
giving a geometrical construction of the category generated by a signature, in
which morphisms are equivalence classes under continuous deformation of
progressive plane diagrams (we give some more details about those diagrams, also
called string diagrams, later on). Their work is for example extended by Baez
and Langford in~\cite{baez-langford:two-tangles} to give a presentation of the
2-category of 2-tangles in 4~dimensions. The other general methodology the
author is aware of, is given by Lack in~\cite{lack:composing-props}, by
constructing elaborate monoidal theories from simpler monoidal theories. Namely,
a monoidal theory can be seen as a monad in a particular span bicategory, and
monoidal theories can therefore be ``composed'' given a distributive law between
their corresponding monads. We chose not to use those me\-thods because, even
though they can be very helpful to build intuitions, they are difficult to
formalize and even more to mechanize.

\paragraph{String diagrams.}
\label{subsection:string-diagrams}
\emph{String diagrams} provide a convenient way to represent and manipulate the
morphisms in the category generated by a presentation. Given an object $M$ in a
strict monoidal category $\mathcal{C}$, a morphism $\mu:M\otimes M\to M$ can be
drawn graphically as a device with two inputs and one output of type $M$ as
follows:
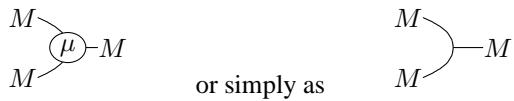
\begin{mymath}
\begin{tikzpicture}[xscale=0.40,yscale=0.40]
\useasboundingbox (-0.5,-0.5) rectangle (4.5,2.5);
\draw[,] (2.00,1.00) -- (3.00,1.00);
\draw[] (1.00,2.00) -- (1.06,1.97) -- (1.13,1.94) -- (1.19,1.91) -- (1.26,1.88) -- (1.32,1.85) -- (1.38,1.81) -- (1.44,1.78) -- (1.50,1.74) -- (1.56,1.71) -- (1.61,1.67) -- (1.66,1.63) -- (1.71,1.59) -- (1.76,1.55) -- (1.80,1.51) -- (1.84,1.46) -- (1.87,1.41) -- (1.90,1.36) -- (1.93,1.31) -- (1.96,1.26) -- (1.97,1.20);
\draw[] (1.97,1.20) -- (1.98,1.19) -- (1.98,1.18) -- (1.98,1.17) -- (1.98,1.16) -- (1.99,1.15) -- (1.99,1.14) -- (1.99,1.13) -- (1.99,1.12) -- (1.99,1.11) -- (1.99,1.10) -- (1.99,1.09) -- (2.00,1.08) -- (2.00,1.07) -- (2.00,1.06) -- (2.00,1.05) -- (2.00,1.04) -- (2.00,1.03) -- (2.00,1.02) -- (2.00,1.01) -- (2.00,1.00);
\draw[] (2.00,1.00) -- (2.00,0.99) -- (2.00,0.98) -- (2.00,0.97) -- (2.00,0.96) -- (2.00,0.95) -- (2.00,0.94) -- (2.00,0.93) -- (2.00,0.92) -- (1.99,0.91) -- (1.99,0.90) -- (1.99,0.89) -- (1.99,0.88) -- (1.99,0.87) -- (1.99,0.86) -- (1.99,0.85) -- (1.98,0.84) -- (1.98,0.83) -- (1.98,0.82) -- (1.98,0.81) -- (1.97,0.80);
\draw[] (1.97,0.80) -- (1.96,0.74) -- (1.93,0.69) -- (1.90,0.64) -- (1.87,0.59) -- (1.84,0.54) -- (1.80,0.49) -- (1.76,0.45) -- (1.71,0.41) -- (1.66,0.37) -- (1.61,0.33) -- (1.56,0.29) -- (1.50,0.26) -- (1.44,0.22) -- (1.38,0.19) -- (1.32,0.15) -- (1.26,0.12) -- (1.19,0.09) -- (1.13,0.06) -- (1.06,0.03) -- (1.00,0.00);
\filldraw[fill=white] (2.00,1.00) ellipse (0.60cm and 0.50cm);
\draw (0.50,2.00) node{$M$};
\draw (2.00,1.00) node{$\mu$};
\draw (3.50,1.00) node{$M$};
\draw (0.50,0.00) node{$M$};
\end{tikzpicture}
\qquad\text{or simply as}\qquad
\begin{tikzpicture}[xscale=0.40,yscale=0.40]
\useasboundingbox (-0.5,-0.5) rectangle (4.5,2.5);
\draw[,] (2.00,1.00) -- (3.00,1.00);
\draw[] (1.00,2.00) -- (1.06,1.97) -- (1.13,1.94) -- (1.19,1.91) -- (1.26,1.88) -- (1.32,1.85) -- (1.38,1.81) -- (1.44,1.78) -- (1.50,1.74) -- (1.56,1.71) -- (1.61,1.67) -- (1.66,1.63) -- (1.71,1.59) -- (1.76,1.55) -- (1.80,1.51) -- (1.84,1.46) -- (1.87,1.41) -- (1.90,1.36) -- (1.93,1.31) -- (1.96,1.26) -- (1.97,1.20);
\draw[] (1.97,1.20) -- (1.98,1.19) -- (1.98,1.18) -- (1.98,1.17) -- (1.98,1.16) -- (1.99,1.15) -- (1.99,1.14) -- (1.99,1.13) -- (1.99,1.12) -- (1.99,1.11) -- (1.99,1.10) -- (1.99,1.09) -- (2.00,1.08) -- (2.00,1.07) -- (2.00,1.06) -- (2.00,1.05) -- (2.00,1.04) -- (2.00,1.03) -- (2.00,1.02) -- (2.00,1.01) -- (2.00,1.00);
\draw[] (2.00,1.00) -- (2.00,0.99) -- (2.00,0.98) -- (2.00,0.97) -- (2.00,0.96) -- (2.00,0.95) -- (2.00,0.94) -- (2.00,0.93) -- (2.00,0.92) -- (1.99,0.91) -- (1.99,0.90) -- (1.99,0.89) -- (1.99,0.88) -- (1.99,0.87) -- (1.99,0.86) -- (1.99,0.85) -- (1.98,0.84) -- (1.98,0.83) -- (1.98,0.82) -- (1.98,0.81) -- (1.97,0.80);
\draw[] (1.97,0.80) -- (1.96,0.74) -- (1.93,0.69) -- (1.90,0.64) -- (1.87,0.59) -- (1.84,0.54) -- (1.80,0.49) -- (1.76,0.45) -- (1.71,0.41) -- (1.66,0.37) -- (1.61,0.33) -- (1.56,0.29) -- (1.50,0.26) -- (1.44,0.22) -- (1.38,0.19) -- (1.32,0.15) -- (1.26,0.12) -- (1.19,0.09) -- (1.13,0.06) -- (1.06,0.03) -- (1.00,0.00);
\draw (0.50,2.00) node{$M$};
\draw (3.50,1.00) node{$M$};
\draw (0.50,0.00) node{$M$};
\end{tikzpicture}
\end{mymath}
when it is clear from the context which morphism of type $M\otimes M\to M$ we
are picturing (we sometimes even omit the source and target of the
morphisms). Similarly, the identity $\id_M:M\to M$ (which we sometimes simply
write $M$) can be pictured as a wire
\vspace\reduce
\[
\begin{tikzpicture}[xscale=0.40,yscale=0.40]
\useasboundingbox (-0.5,-0.5) rectangle (4.5,0.5);
\draw[] (1.00,0.00) -- (1.05,0.00) -- (1.10,0.00) -- (1.15,0.00) -- (1.20,0.00) -- (1.25,0.00) -- (1.30,0.00) -- (1.35,0.00) -- (1.40,0.00) -- (1.45,0.00) -- (1.50,0.00) -- (1.55,0.00) -- (1.60,0.00) -- (1.65,0.00) -- (1.70,0.00) -- (1.75,0.00) -- (1.80,0.00) -- (1.85,0.00) -- (1.90,0.00) -- (1.95,0.00) -- (2.00,0.00);
\draw[] (2.00,0.00) -- (2.05,0.00) -- (2.10,0.00) -- (2.15,0.00) -- (2.20,0.00) -- (2.25,0.00) -- (2.30,0.00) -- (2.35,0.00) -- (2.40,0.00) -- (2.45,0.00) -- (2.50,0.00) -- (2.55,0.00) -- (2.60,0.00) -- (2.65,0.00) -- (2.70,0.00) -- (2.75,0.00) -- (2.80,0.00) -- (2.85,0.00) -- (2.90,0.00) -- (2.95,0.00) -- (3.00,0.00);
\draw (0.50,0.00) node{$M$};
\draw (3.50,0.00) node{$M$};
\end{tikzpicture}
\vspace\reduce
\]
The tensor $f\otimes g$ of two morphisms $f:A\to B$ and \hbox{$g:C\to D$} is
obtained by putting the diagram corresponding to $f$ above the diagram
corresponding to~$g$.
So, for instance, the morphism $\mu\otimes M$
can be drawn diagrammatically as
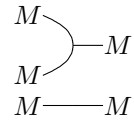
\[
\begin{tikzpicture}[xscale=0.40,yscale=0.40]
\useasboundingbox (-0.5,-0.5) rectangle (4.5,3.5);
\draw[,] (2.00,2.00) -- (3.00,2.00);
\draw[] (1.00,3.00) -- (1.06,2.97) -- (1.13,2.94) -- (1.19,2.91) -- (1.26,2.88) -- (1.32,2.85) -- (1.38,2.81) -- (1.44,2.78) -- (1.50,2.74) -- (1.56,2.71) -- (1.61,2.67) -- (1.66,2.63) -- (1.71,2.59) -- (1.76,2.55) -- (1.80,2.51) -- (1.84,2.46) -- (1.87,2.41) -- (1.90,2.36) -- (1.93,2.31) -- (1.96,2.26) -- (1.97,2.20);
\draw[] (1.97,2.20) -- (1.98,2.19) -- (1.98,2.18) -- (1.98,2.17) -- (1.98,2.16) -- (1.99,2.15) -- (1.99,2.14) -- (1.99,2.13) -- (1.99,2.12) -- (1.99,2.11) -- (1.99,2.10) -- (1.99,2.09) -- (2.00,2.08) -- (2.00,2.07) -- (2.00,2.06) -- (2.00,2.05) -- (2.00,2.04) -- (2.00,2.03) -- (2.00,2.02) -- (2.00,2.01) -- (2.00,2.00);
\draw[] (2.00,2.00) -- (2.00,1.99) -- (2.00,1.98) -- (2.00,1.97) -- (2.00,1.96) -- (2.00,1.95) -- (2.00,1.94) -- (2.00,1.93) -- (2.00,1.92) -- (1.99,1.91) -- (1.99,1.90) -- (1.99,1.89) -- (1.99,1.88) -- (1.99,1.87) -- (1.99,1.86) -- (1.99,1.85) -- (1.98,1.84) -- (1.98,1.83) -- (1.98,1.82) -- (1.98,1.81) -- (1.97,1.80);
\draw[] (1.97,1.80) -- (1.96,1.74) -- (1.93,1.69) -- (1.90,1.64) -- (1.87,1.59) -- (1.84,1.54) -- (1.80,1.49) -- (1.76,1.45) -- (1.71,1.41) -- (1.66,1.37) -- (1.61,1.33) -- (1.56,1.29) -- (1.50,1.26) -- (1.44,1.22) -- (1.38,1.19) -- (1.32,1.15) -- (1.26,1.12) -- (1.19,1.09) -- (1.13,1.06) -- (1.06,1.03) -- (1.00,1.00);
\draw[] (1.00,0.00) -- (1.05,0.00) -- (1.10,0.00) -- (1.15,0.00) -- (1.20,0.00) -- (1.25,0.00) -- (1.30,0.00) -- (1.35,0.00) -- (1.40,0.00) -- (1.45,0.00) -- (1.50,0.00) -- (1.55,0.00) -- (1.60,0.00) -- (1.65,0.00) -- (1.70,0.00) -- (1.75,0.00) -- (1.80,0.00) -- (1.85,0.00) -- (1.90,0.00) -- (1.95,0.00) -- (2.00,0.00);
\draw[] (2.00,0.00) -- (2.05,0.00) -- (2.10,0.00) -- (2.15,0.00) -- (2.20,0.00) -- (2.25,0.00) -- (2.30,0.00) -- (2.35,0.00) -- (2.40,0.00) -- (2.45,0.00) -- (2.50,0.00) -- (2.55,0.00) -- (2.60,0.00) -- (2.65,0.00) -- (2.70,0.00) -- (2.75,0.00) -- (2.80,0.00) -- (2.85,0.00) -- (2.90,0.00) -- (2.95,0.00) -- (3.00,0.00);
\draw (0.50,3.00) node{$M$};
\draw (3.50,2.00) node{$M$};
\draw (0.50,1.00) node{$M$};
\draw (0.50,0.00) node{$M$};
\draw (3.50,0.00) node{$M$};
\end{tikzpicture}
\vspace\reduce
\]
Finally, the composite $g\circ f:A\to C$ of two morphisms $f:A\to B$ and $g:B\to
C$ can be drawn diagrammatically by putting the diagram corresponding to $g$ at
the right of the diagram corresponding to $f$ and ``linking the wires''.
The diagram corresponding to the morphism $\mu\circ(\mu\otimes M)$
is thus
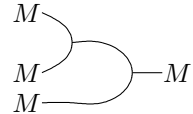
\begin{mymath}
\begin{tikzpicture}[xscale=0.40,yscale=0.40]
\useasboundingbox (-0.5,-0.5) rectangle (6.5,3.5);
\draw[] (2.00,2.00) -- (2.05,2.01) -- (2.10,2.01) -- (2.15,2.02) -- (2.20,2.02) -- (2.25,2.03) -- (2.30,2.03) -- (2.35,2.04) -- (2.40,2.04) -- (2.45,2.04) -- (2.50,2.05) -- (2.55,2.05) -- (2.60,2.05) -- (2.65,2.05) -- (2.70,2.04) -- (2.75,2.04) -- (2.80,2.03) -- (2.85,2.03) -- (2.90,2.02) -- (2.95,2.01) -- (3.00,2.00);
\draw[] (3.00,2.00) -- (3.06,1.98) -- (3.13,1.96) -- (3.19,1.94) -- (3.25,1.92) -- (3.32,1.89) -- (3.38,1.86) -- (3.44,1.83) -- (3.50,1.79) -- (3.55,1.75) -- (3.61,1.71) -- (3.66,1.67) -- (3.71,1.63) -- (3.75,1.58) -- (3.80,1.53) -- (3.84,1.48) -- (3.87,1.43) -- (3.90,1.37) -- (3.93,1.32) -- (3.96,1.26) -- (3.97,1.20);
\draw[] (3.97,1.20) -- (3.98,1.19) -- (3.98,1.18) -- (3.98,1.17) -- (3.98,1.16) -- (3.99,1.15) -- (3.99,1.14) -- (3.99,1.13) -- (3.99,1.12) -- (3.99,1.11) -- (3.99,1.10) -- (3.99,1.09) -- (4.00,1.08) -- (4.00,1.07) -- (4.00,1.06) -- (4.00,1.05) -- (4.00,1.04) -- (4.00,1.03) -- (4.00,1.02) -- (4.00,1.01) -- (4.00,1.00);
\draw[] (4.00,1.00) -- (4.00,0.99) -- (4.00,0.98) -- (4.00,0.97) -- (4.00,0.96) -- (4.00,0.95) -- (4.00,0.94) -- (4.00,0.93) -- (4.00,0.92) -- (3.99,0.91) -- (3.99,0.90) -- (3.99,0.89) -- (3.99,0.88) -- (3.99,0.87) -- (3.99,0.86) -- (3.99,0.85) -- (3.98,0.84) -- (3.98,0.83) -- (3.98,0.82) -- (3.98,0.81) -- (3.97,0.80);
\draw[] (3.97,0.80) -- (3.96,0.74) -- (3.93,0.68) -- (3.90,0.63) -- (3.87,0.57) -- (3.84,0.52) -- (3.80,0.47) -- (3.75,0.42) -- (3.71,0.37) -- (3.66,0.33) -- (3.61,0.28) -- (3.55,0.24) -- (3.50,0.21) -- (3.44,0.17) -- (3.38,0.14) -- (3.32,0.11) -- (3.25,0.08) -- (3.19,0.06) -- (3.13,0.03) -- (3.06,0.02) -- (3.00,0.00);
\draw[] (3.00,0.00) -- (2.95,-0.01) -- (2.90,-0.02) -- (2.85,-0.03) -- (2.80,-0.03) -- (2.75,-0.03) -- (2.70,-0.04) -- (2.65,-0.04) -- (2.60,-0.04) -- (2.55,-0.04) -- (2.50,-0.04) -- (2.45,-0.03) -- (2.40,-0.03) -- (2.35,-0.03) -- (2.30,-0.02) -- (2.25,-0.02) -- (2.20,-0.02) -- (2.15,-0.01) -- (2.10,-0.01) -- (2.05,-0.00) -- (2.00,0.00);
\draw[] (2.00,0.00) -- (1.95,0.00) -- (1.90,0.01) -- (1.85,0.01) -- (1.80,0.01) -- (1.75,0.01) -- (1.70,0.01) -- (1.65,0.01) -- (1.60,0.01) -- (1.55,0.01) -- (1.50,0.01) -- (1.45,0.01) -- (1.40,0.01) -- (1.35,0.01) -- (1.30,0.01) -- (1.25,0.01) -- (1.20,0.01) -- (1.15,0.00) -- (1.10,0.00) -- (1.05,0.00) -- (1.00,0.00);
\draw[] (1.00,3.00) -- (1.06,2.97) -- (1.13,2.94) -- (1.19,2.91) -- (1.26,2.88) -- (1.32,2.85) -- (1.38,2.81) -- (1.44,2.78) -- (1.50,2.74) -- (1.56,2.71) -- (1.61,2.67) -- (1.66,2.63) -- (1.71,2.59) -- (1.76,2.55) -- (1.80,2.51) -- (1.84,2.46) -- (1.87,2.41) -- (1.90,2.36) -- (1.93,2.31) -- (1.96,2.26) -- (1.97,2.20);
\draw[] (1.97,2.20) -- (1.98,2.19) -- (1.98,2.18) -- (1.98,2.17) -- (1.98,2.16) -- (1.99,2.15) -- (1.99,2.14) -- (1.99,2.13) -- (1.99,2.12) -- (1.99,2.11) -- (1.99,2.10) -- (1.99,2.09) -- (2.00,2.08) -- (2.00,2.07) -- (2.00,2.06) -- (2.00,2.05) -- (2.00,2.04) -- (2.00,2.03) -- (2.00,2.02) -- (2.00,2.01) -- (2.00,2.00);
\draw[] (2.00,2.00) -- (2.00,1.99) -- (2.00,1.98) -- (2.00,1.97) -- (2.00,1.96) -- (2.00,1.95) -- (2.00,1.94) -- (2.00,1.93) -- (2.00,1.92) -- (1.99,1.91) -- (1.99,1.90) -- (1.99,1.89) -- (1.99,1.88) -- (1.99,1.87) -- (1.99,1.86) -- (1.99,1.85) -- (1.98,1.84) -- (1.98,1.83) -- (1.98,1.82) -- (1.98,1.81) -- (1.97,1.80);
\draw[] (1.97,1.80) -- (1.96,1.74) -- (1.93,1.69) -- (1.90,1.64) -- (1.87,1.59) -- (1.84,1.54) -- (1.80,1.49) -- (1.76,1.45) -- (1.71,1.41) -- (1.66,1.37) -- (1.61,1.33) -- (1.56,1.29) -- (1.50,1.26) -- (1.44,1.22) -- (1.38,1.19) -- (1.32,1.15) -- (1.26,1.12) -- (1.19,1.09) -- (1.13,1.06) -- (1.06,1.03) -- (1.00,1.00);
\draw[,] (4.00,1.00) -- (5.00,1.00);
\draw (0.50,3.00) node{$M$};
\draw (0.50,1.00) node{$M$};
\draw (5.50,1.00) node{$M$};
\draw (0.50,0.00) node{$M$};
\end{tikzpicture}
\end{mymath}

Suppose that $(E_1,s_1,t_1,E_2)$ is a signature. Every element $f$ of $E_2$ such
that
\begin{mymath}
s_1(f)=A_1\otimes\cdots\otimes A_m
\qtand
t_1(f)=B_1\otimes\cdots\otimes B_n
\end{mymath}
where the $A_i$ and $B_i$ are elements of $E_1$, can be similarly represented by
a diagram
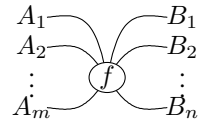
\begin{mymath}
\begin{tikzpicture}[xscale=0.40,yscale=0.40]
\useasboundingbox (-0.5,-0.5) rectangle (6.5,3.5);
\draw[] (1.00,3.00) -- (1.05,3.01) -- (1.11,3.01) -- (1.16,3.02) -- (1.22,3.03) -- (1.27,3.03) -- (1.32,3.04) -- (1.38,3.04) -- (1.43,3.05) -- (1.48,3.05) -- (1.53,3.05) -- (1.58,3.05) -- (1.63,3.05) -- (1.68,3.05) -- (1.73,3.05) -- (1.78,3.05) -- (1.82,3.04) -- (1.87,3.03) -- (1.91,3.02) -- (1.96,3.01) -- (2.00,3.00);
\draw[] (2.00,3.00) -- (2.09,2.96) -- (2.17,2.92) -- (2.25,2.86) -- (2.32,2.80) -- (2.39,2.73) -- (2.45,2.65) -- (2.51,2.56) -- (2.56,2.47) -- (2.61,2.37) -- (2.66,2.26) -- (2.70,2.15) -- (2.74,2.03) -- (2.78,1.91) -- (2.82,1.79) -- (2.85,1.66) -- (2.88,1.53) -- (2.91,1.40) -- (2.94,1.27) -- (2.97,1.13) -- (3.00,1.00);
\draw[] (5.00,3.00) -- (4.95,3.01) -- (4.89,3.01) -- (4.84,3.02) -- (4.78,3.03) -- (4.73,3.03) -- (4.68,3.04) -- (4.62,3.04) -- (4.57,3.05) -- (4.52,3.05) -- (4.47,3.05) -- (4.42,3.05) -- (4.37,3.05) -- (4.32,3.05) -- (4.27,3.05) -- (4.22,3.05) -- (4.18,3.04) -- (4.13,3.03) -- (4.09,3.02) -- (4.04,3.01) -- (4.00,3.00);
\draw[] (4.00,3.00) -- (3.91,2.96) -- (3.83,2.92) -- (3.75,2.86) -- (3.68,2.80) -- (3.61,2.73) -- (3.55,2.65) -- (3.49,2.56) -- (3.44,2.47) -- (3.39,2.37) -- (3.34,2.26) -- (3.30,2.15) -- (3.26,2.03) -- (3.22,1.91) -- (3.18,1.79) -- (3.15,1.66) -- (3.12,1.53) -- (3.09,1.40) -- (3.06,1.27) -- (3.03,1.13) -- (3.00,1.00);
\draw[] (1.00,2.00) -- (1.05,2.01) -- (1.11,2.01) -- (1.16,2.02) -- (1.21,2.03) -- (1.26,2.03) -- (1.32,2.04) -- (1.37,2.04) -- (1.42,2.05) -- (1.47,2.05) -- (1.52,2.05) -- (1.57,2.06) -- (1.62,2.06) -- (1.67,2.05) -- (1.72,2.05) -- (1.77,2.05) -- (1.82,2.04) -- (1.86,2.03) -- (1.91,2.03) -- (1.96,2.01) -- (2.00,2.00);
\draw[] (2.00,2.00) -- (2.06,1.98) -- (2.12,1.95) -- (2.18,1.92) -- (2.23,1.88) -- (2.29,1.85) -- (2.34,1.80) -- (2.40,1.76) -- (2.45,1.71) -- (2.50,1.66) -- (2.55,1.61) -- (2.59,1.56) -- (2.64,1.50) -- (2.69,1.44) -- (2.73,1.38) -- (2.78,1.32) -- (2.82,1.26) -- (2.87,1.19) -- (2.91,1.13) -- (2.96,1.06) -- (3.00,1.00);
\draw[] (5.00,2.00) -- (4.95,2.01) -- (4.89,2.01) -- (4.84,2.02) -- (4.79,2.03) -- (4.74,2.03) -- (4.68,2.04) -- (4.63,2.04) -- (4.58,2.05) -- (4.53,2.05) -- (4.48,2.05) -- (4.43,2.06) -- (4.38,2.06) -- (4.33,2.05) -- (4.28,2.05) -- (4.23,2.05) -- (4.18,2.04) -- (4.14,2.03) -- (4.09,2.03) -- (4.04,2.01) -- (4.00,2.00);
\draw[] (4.00,2.00) -- (3.94,1.98) -- (3.88,1.95) -- (3.82,1.92) -- (3.77,1.88) -- (3.71,1.85) -- (3.66,1.80) -- (3.60,1.76) -- (3.55,1.71) -- (3.50,1.66) -- (3.45,1.61) -- (3.41,1.56) -- (3.36,1.50) -- (3.31,1.44) -- (3.27,1.38) -- (3.22,1.32) -- (3.18,1.26) -- (3.13,1.19) -- (3.09,1.13) -- (3.04,1.06) -- (3.00,1.00);
\draw[] (3.00,1.00) -- (3.04,0.94) -- (3.09,0.87) -- (3.13,0.81) -- (3.18,0.74) -- (3.22,0.68) -- (3.27,0.62) -- (3.31,0.56) -- (3.36,0.50) -- (3.41,0.44) -- (3.45,0.39) -- (3.50,0.34) -- (3.55,0.29) -- (3.60,0.24) -- (3.66,0.20) -- (3.71,0.15) -- (3.77,0.12) -- (3.82,0.08) -- (3.88,0.05) -- (3.94,0.02) -- (4.00,0.00);
\draw[] (4.00,0.00) -- (4.04,-0.01) -- (4.09,-0.03) -- (4.14,-0.03) -- (4.18,-0.04) -- (4.23,-0.05) -- (4.28,-0.05) -- (4.33,-0.05) -- (4.38,-0.06) -- (4.43,-0.06) -- (4.48,-0.05) -- (4.53,-0.05) -- (4.58,-0.05) -- (4.63,-0.04) -- (4.68,-0.04) -- (4.74,-0.03) -- (4.79,-0.03) -- (4.84,-0.02) -- (4.89,-0.01) -- (4.95,-0.01) -- (5.00,0.00);
\draw[] (1.00,0.00) -- (1.05,-0.01) -- (1.11,-0.01) -- (1.16,-0.02) -- (1.21,-0.03) -- (1.26,-0.03) -- (1.32,-0.04) -- (1.37,-0.04) -- (1.42,-0.05) -- (1.47,-0.05) -- (1.52,-0.05) -- (1.57,-0.06) -- (1.62,-0.06) -- (1.67,-0.05) -- (1.72,-0.05) -- (1.77,-0.05) -- (1.82,-0.04) -- (1.86,-0.03) -- (1.91,-0.03) -- (1.96,-0.01) -- (2.00,0.00);
\draw[] (2.00,0.00) -- (2.06,0.02) -- (2.12,0.05) -- (2.18,0.08) -- (2.23,0.12) -- (2.29,0.15) -- (2.34,0.20) -- (2.40,0.24) -- (2.45,0.29) -- (2.50,0.34) -- (2.55,0.39) -- (2.59,0.44) -- (2.64,0.50) -- (2.69,0.56) -- (2.73,0.62) -- (2.78,0.68) -- (2.82,0.74) -- (2.87,0.81) -- (2.91,0.87) -- (2.96,0.94) -- (3.00,1.00);
\filldraw[fill=white] (3.00,1.00) ellipse (0.60cm and 0.50cm);
\draw (0.50,3.00) node{$A_1$};
\draw (5.50,3.00) node{$B_1$};
\draw (0.50,2.00) node{$A_2$};
\draw (5.50,2.00) node{$B_2$};
\draw (0.50,1.00) node{$\vdots$};
\draw (3.00,1.00) node{$f$};
\draw (5.50,1.00) node{$\vdots$};
\draw (0.50,0.00) node{$A_m$};
\draw (5.50,0.00) node{$B_n$};
\end{tikzpicture}
\end{mymath}
Bigger diagrams can be constructed from these diagrams by composing and
tensoring them, as explained above. Joyal and Street have shown in details
in~\cite{joyal-street:geometry-tensor-calculus} that the category of those
diagrams, modulo continuous deformations, is precisely the free category
generated by a signature (which they call a ``tensor scheme''). For example, the
equality
\begin{mymath}
(M\otimes\mu)\circ(\mu\otimes M\otimes M)
\qeq
(\mu\otimes M)\circ(M\otimes M\otimes\mu)
\end{mymath}
in the category $\mathcal{C}$ of the above example, which holds because of the
axioms satisfied in any monoidal category, can be shown by continuously
deforming the diagram on the left-hand side below into the diagram on the
right-hand side:
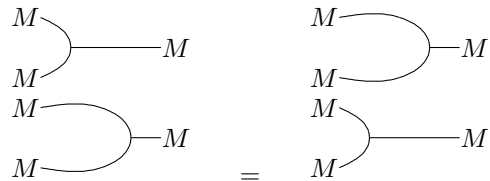
\begin{mymath}
\sstrid{mu_x_mu_r}
\qeq
\sstrid{mu_x_mu_l}
\end{mymath}
All the equalities satisfied in any monoidal category generated by a signature
have a similar geometrical interpretation. And conversely, any deformation of
diagrams corresponds to an equality of morphisms in monoidal categories.

\section{Algebraic structures}
\label{sec:alg-struct}
In this section, we recall the categorical formulation of some well-known
algebraic structures.
We give those definitions in the setting of a \emph{strict} monoidal category
which is \emph{not} required to be symmetric.  We suppose that
$(\mathcal{C},\otimes,I)$ is a strict monoidal category, fixed throughout the
section. For the lack of space, we will only give graphical representations of
axioms, but they can always be reformulated as commutative diagrams.

\paragraph{Symmetric objects.}
A \emph{symmetric object} of $\mathcal{C}$ is an object~$S$ together with a
morphism
\begin{mymath}
\gamma:S\otimes S\to S\otimes S
\end{mymath}
called \emph{symmetry} and pictured as
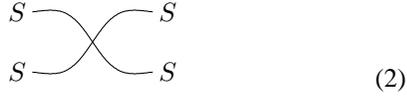
\begin{myequation}
  \label{eq:sym-string}
  \begin{tikzpicture}[xscale=0.40,yscale=0.40]
\useasboundingbox (-0.5,-0.5) rectangle (6.5,2.5);
\draw[] (1.00,2.00) -- (1.05,2.01) -- (1.11,2.01) -- (1.16,2.02) -- (1.21,2.03) -- (1.26,2.03) -- (1.32,2.04) -- (1.37,2.04) -- (1.42,2.05) -- (1.47,2.05) -- (1.52,2.05) -- (1.57,2.06) -- (1.62,2.06) -- (1.67,2.05) -- (1.72,2.05) -- (1.77,2.05) -- (1.82,2.04) -- (1.86,2.03) -- (1.91,2.03) -- (1.96,2.01) -- (2.00,2.00);
\draw[] (2.00,2.00) -- (2.06,1.98) -- (2.12,1.95) -- (2.18,1.92) -- (2.23,1.88) -- (2.29,1.85) -- (2.34,1.80) -- (2.40,1.76) -- (2.45,1.71) -- (2.50,1.66) -- (2.55,1.61) -- (2.59,1.56) -- (2.64,1.50) -- (2.69,1.44) -- (2.73,1.38) -- (2.78,1.32) -- (2.82,1.26) -- (2.87,1.19) -- (2.91,1.13) -- (2.96,1.06) -- (3.00,1.00);
\draw[] (3.00,1.00) -- (3.04,0.94) -- (3.09,0.87) -- (3.13,0.81) -- (3.18,0.74) -- (3.22,0.68) -- (3.27,0.62) -- (3.31,0.56) -- (3.36,0.50) -- (3.41,0.44) -- (3.45,0.39) -- (3.50,0.34) -- (3.55,0.29) -- (3.60,0.24) -- (3.66,0.20) -- (3.71,0.15) -- (3.77,0.12) -- (3.82,0.08) -- (3.88,0.05) -- (3.94,0.02) -- (4.00,0.00);
\draw[] (4.00,0.00) -- (4.04,-0.01) -- (4.09,-0.03) -- (4.14,-0.03) -- (4.18,-0.04) -- (4.23,-0.05) -- (4.28,-0.05) -- (4.33,-0.05) -- (4.38,-0.06) -- (4.43,-0.06) -- (4.48,-0.05) -- (4.53,-0.05) -- (4.58,-0.05) -- (4.63,-0.04) -- (4.68,-0.04) -- (4.74,-0.03) -- (4.79,-0.03) -- (4.84,-0.02) -- (4.89,-0.01) -- (4.95,-0.01) -- (5.00,0.00);
\draw[] (5.00,2.00) -- (4.95,2.01) -- (4.89,2.01) -- (4.84,2.02) -- (4.79,2.03) -- (4.74,2.03) -- (4.68,2.04) -- (4.63,2.04) -- (4.58,2.05) -- (4.53,2.05) -- (4.48,2.05) -- (4.43,2.06) -- (4.38,2.06) -- (4.33,2.05) -- (4.28,2.05) -- (4.23,2.05) -- (4.18,2.04) -- (4.14,2.03) -- (4.09,2.03) -- (4.04,2.01) -- (4.00,2.00);
\draw[] (4.00,2.00) -- (3.94,1.98) -- (3.88,1.95) -- (3.82,1.92) -- (3.77,1.88) -- (3.71,1.85) -- (3.66,1.80) -- (3.60,1.76) -- (3.55,1.71) -- (3.50,1.66) -- (3.45,1.61) -- (3.41,1.56) -- (3.36,1.50) -- (3.31,1.44) -- (3.27,1.38) -- (3.22,1.32) -- (3.18,1.26) -- (3.13,1.19) -- (3.09,1.13) -- (3.04,1.06) -- (3.00,1.00);
\draw[] (3.00,1.00) -- (2.96,0.94) -- (2.91,0.87) -- (2.87,0.81) -- (2.82,0.74) -- (2.78,0.68) -- (2.73,0.62) -- (2.69,0.56) -- (2.64,0.50) -- (2.59,0.44) -- (2.55,0.39) -- (2.50,0.34) -- (2.45,0.29) -- (2.40,0.24) -- (2.34,0.20) -- (2.29,0.15) -- (2.23,0.12) -- (2.18,0.08) -- (2.12,0.05) -- (2.06,0.02) -- (2.00,0.00);
\draw[] (2.00,0.00) -- (1.96,-0.01) -- (1.91,-0.03) -- (1.86,-0.03) -- (1.82,-0.04) -- (1.77,-0.05) -- (1.72,-0.05) -- (1.67,-0.05) -- (1.62,-0.06) -- (1.57,-0.06) -- (1.52,-0.05) -- (1.47,-0.05) -- (1.42,-0.05) -- (1.37,-0.04) -- (1.32,-0.04) -- (1.26,-0.03) -- (1.21,-0.03) -- (1.16,-0.02) -- (1.11,-0.01) -- (1.05,-0.01) -- (1.00,0.00);
\draw (0.50,2.00) node{$S$};
\draw (5.50,2.00) node{$S$};
\draw (0.50,0.00) node{$S$};
\draw (5.50,0.00) node{$S$};
\end{tikzpicture}
\end{myequation}
such that the two equalities
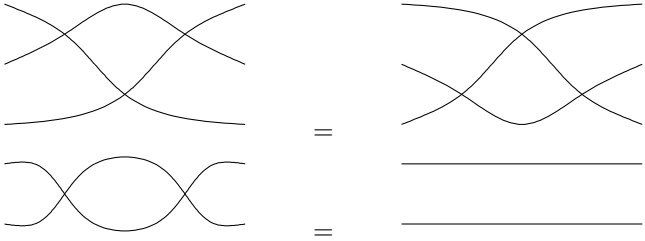
\begin{mymath}
\hspace{-3ex}
\begin{array}{rcl}
  \sstrid{yang_baxter_r}
  &\qeq&
  \sstrid{yang_baxter_l}
  \\
  \sstrid{sym_sym}
  &\qeq&
  \sstrid{id_x_id}
\end{array}
\end{mymath}
hold (the first equation is sometimes called the Yang-Baxter equation for
braids). In particular, in a symmetric monoidal category, every object is
canonically equipped with a structure of symmetric object.


\paragraph{Monoids.}
\label{subsection:monoids}
A \emph{monoid} $(M,\mu,\eta)$ in $\mathcal{C}$ is an object $M$ together with
two morphisms
\begin{mymath}
\mu : M\otimes M\to M
\qtand
\eta : I\to M
\end{mymath}
called respectively \emph{multiplication} and \emph{unit} and pictured respectively as
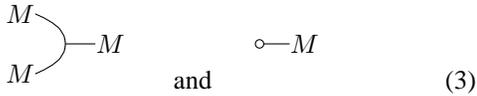
\begin{myequation}
  \label{eq:monoid-string}
  \begin{tikzpicture}[xscale=0.40,yscale=0.40]
\useasboundingbox (-0.5,-0.5) rectangle (4.5,2.5);
\draw[,] (2.00,1.00) -- (3.00,1.00);
\draw[] (1.00,2.00) -- (1.06,1.97) -- (1.13,1.94) -- (1.19,1.91) -- (1.26,1.88) -- (1.32,1.85) -- (1.38,1.81) -- (1.44,1.78) -- (1.50,1.74) -- (1.56,1.71) -- (1.61,1.67) -- (1.66,1.63) -- (1.71,1.59) -- (1.76,1.55) -- (1.80,1.51) -- (1.84,1.46) -- (1.87,1.41) -- (1.90,1.36) -- (1.93,1.31) -- (1.96,1.26) -- (1.97,1.20);
\draw[] (1.97,1.20) -- (1.98,1.19) -- (1.98,1.18) -- (1.98,1.17) -- (1.98,1.16) -- (1.99,1.15) -- (1.99,1.14) -- (1.99,1.13) -- (1.99,1.12) -- (1.99,1.11) -- (1.99,1.10) -- (1.99,1.09) -- (2.00,1.08) -- (2.00,1.07) -- (2.00,1.06) -- (2.00,1.05) -- (2.00,1.04) -- (2.00,1.03) -- (2.00,1.02) -- (2.00,1.01) -- (2.00,1.00);
\draw[] (2.00,1.00) -- (2.00,0.99) -- (2.00,0.98) -- (2.00,0.97) -- (2.00,0.96) -- (2.00,0.95) -- (2.00,0.94) -- (2.00,0.93) -- (2.00,0.92) -- (1.99,0.91) -- (1.99,0.90) -- (1.99,0.89) -- (1.99,0.88) -- (1.99,0.87) -- (1.99,0.86) -- (1.99,0.85) -- (1.98,0.84) -- (1.98,0.83) -- (1.98,0.82) -- (1.98,0.81) -- (1.97,0.80);
\draw[] (1.97,0.80) -- (1.96,0.74) -- (1.93,0.69) -- (1.90,0.64) -- (1.87,0.59) -- (1.84,0.54) -- (1.80,0.49) -- (1.76,0.45) -- (1.71,0.41) -- (1.66,0.37) -- (1.61,0.33) -- (1.56,0.29) -- (1.50,0.26) -- (1.44,0.22) -- (1.38,0.19) -- (1.32,0.15) -- (1.26,0.12) -- (1.19,0.09) -- (1.13,0.06) -- (1.06,0.03) -- (1.00,0.00);
\draw (0.50,2.00) node{$M$};
\draw (3.50,1.00) node{$M$};
\draw (0.50,0.00) node{$M$};
\end{tikzpicture}
  \qtand
  \begin{tikzpicture}[xscale=0.40,yscale=0.40]
\useasboundingbox (-0.5,-0.5) rectangle (2.5,2.5);
\draw[,] (0.00,1.00) -- (1.00,1.00);
\filldraw[fill=white] (0.00,1.00) ellipse (0.14cm and 0.14cm);
\draw (1.50,1.00) node{$M$};
\end{tikzpicture}

\end{myequation}
satisfying the three equations
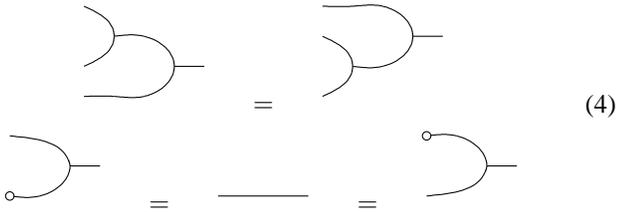
\begin{myequation}
  \label{eq:monoid}
  \hspace{-0.1cm}
  \begin{array}{c}
    \begin{tikzpicture}[xscale=0.40,yscale=0.40]
\useasboundingbox (-0.5,-0.5) rectangle (4.5,3.5);
\draw[] (1.00,2.00) -- (1.05,2.01) -- (1.10,2.01) -- (1.15,2.02) -- (1.20,2.02) -- (1.25,2.03) -- (1.30,2.03) -- (1.35,2.04) -- (1.40,2.04) -- (1.45,2.04) -- (1.50,2.05) -- (1.55,2.05) -- (1.60,2.05) -- (1.65,2.05) -- (1.70,2.04) -- (1.75,2.04) -- (1.80,2.03) -- (1.85,2.03) -- (1.90,2.02) -- (1.95,2.01) -- (2.00,2.00);
\draw[] (2.00,2.00) -- (2.06,1.98) -- (2.13,1.96) -- (2.19,1.94) -- (2.25,1.92) -- (2.32,1.89) -- (2.38,1.86) -- (2.44,1.83) -- (2.50,1.79) -- (2.55,1.75) -- (2.61,1.71) -- (2.66,1.67) -- (2.71,1.63) -- (2.75,1.58) -- (2.80,1.53) -- (2.84,1.48) -- (2.87,1.43) -- (2.90,1.37) -- (2.93,1.32) -- (2.96,1.26) -- (2.97,1.20);
\draw[] (2.97,1.20) -- (2.98,1.19) -- (2.98,1.18) -- (2.98,1.17) -- (2.98,1.16) -- (2.99,1.15) -- (2.99,1.14) -- (2.99,1.13) -- (2.99,1.12) -- (2.99,1.11) -- (2.99,1.10) -- (2.99,1.09) -- (3.00,1.08) -- (3.00,1.07) -- (3.00,1.06) -- (3.00,1.05) -- (3.00,1.04) -- (3.00,1.03) -- (3.00,1.02) -- (3.00,1.01) -- (3.00,1.00);
\draw[] (3.00,1.00) -- (3.00,0.99) -- (3.00,0.98) -- (3.00,0.97) -- (3.00,0.96) -- (3.00,0.95) -- (3.00,0.94) -- (3.00,0.93) -- (3.00,0.92) -- (2.99,0.91) -- (2.99,0.90) -- (2.99,0.89) -- (2.99,0.88) -- (2.99,0.87) -- (2.99,0.86) -- (2.99,0.85) -- (2.98,0.84) -- (2.98,0.83) -- (2.98,0.82) -- (2.98,0.81) -- (2.97,0.80);
\draw[] (2.97,0.80) -- (2.96,0.74) -- (2.93,0.68) -- (2.90,0.63) -- (2.87,0.57) -- (2.84,0.52) -- (2.80,0.47) -- (2.75,0.42) -- (2.71,0.37) -- (2.66,0.33) -- (2.61,0.28) -- (2.55,0.24) -- (2.50,0.21) -- (2.44,0.17) -- (2.38,0.14) -- (2.32,0.11) -- (2.25,0.08) -- (2.19,0.06) -- (2.13,0.03) -- (2.06,0.02) -- (2.00,0.00);
\draw[] (2.00,0.00) -- (1.95,-0.01) -- (1.90,-0.02) -- (1.85,-0.03) -- (1.80,-0.03) -- (1.75,-0.03) -- (1.70,-0.04) -- (1.65,-0.04) -- (1.60,-0.04) -- (1.55,-0.04) -- (1.50,-0.04) -- (1.45,-0.03) -- (1.40,-0.03) -- (1.35,-0.03) -- (1.30,-0.02) -- (1.25,-0.02) -- (1.20,-0.02) -- (1.15,-0.01) -- (1.10,-0.01) -- (1.05,-0.00) -- (1.00,0.00);
\draw[] (1.00,0.00) -- (0.95,0.00) -- (0.90,0.01) -- (0.85,0.01) -- (0.80,0.01) -- (0.75,0.01) -- (0.70,0.01) -- (0.65,0.01) -- (0.60,0.01) -- (0.55,0.01) -- (0.50,0.01) -- (0.45,0.01) -- (0.40,0.01) -- (0.35,0.01) -- (0.30,0.01) -- (0.25,0.01) -- (0.20,0.01) -- (0.15,0.00) -- (0.10,0.00) -- (0.05,0.00) -- (0.00,0.00);
\draw[] (0.00,3.00) -- (0.06,2.97) -- (0.13,2.94) -- (0.19,2.91) -- (0.26,2.88) -- (0.32,2.85) -- (0.38,2.81) -- (0.44,2.78) -- (0.50,2.74) -- (0.56,2.71) -- (0.61,2.67) -- (0.66,2.63) -- (0.71,2.59) -- (0.76,2.55) -- (0.80,2.51) -- (0.84,2.46) -- (0.87,2.41) -- (0.90,2.36) -- (0.93,2.31) -- (0.96,2.26) -- (0.97,2.20);
\draw[] (0.97,2.20) -- (0.98,2.19) -- (0.98,2.18) -- (0.98,2.17) -- (0.98,2.16) -- (0.99,2.15) -- (0.99,2.14) -- (0.99,2.13) -- (0.99,2.12) -- (0.99,2.11) -- (0.99,2.10) -- (0.99,2.09) -- (1.00,2.08) -- (1.00,2.07) -- (1.00,2.06) -- (1.00,2.05) -- (1.00,2.04) -- (1.00,2.03) -- (1.00,2.02) -- (1.00,2.01) -- (1.00,2.00);
\draw[] (1.00,2.00) -- (1.00,1.99) -- (1.00,1.98) -- (1.00,1.97) -- (1.00,1.96) -- (1.00,1.95) -- (1.00,1.94) -- (1.00,1.93) -- (1.00,1.92) -- (0.99,1.91) -- (0.99,1.90) -- (0.99,1.89) -- (0.99,1.88) -- (0.99,1.87) -- (0.99,1.86) -- (0.99,1.85) -- (0.98,1.84) -- (0.98,1.83) -- (0.98,1.82) -- (0.98,1.81) -- (0.97,1.80);
\draw[] (0.97,1.80) -- (0.96,1.74) -- (0.93,1.69) -- (0.90,1.64) -- (0.87,1.59) -- (0.84,1.54) -- (0.80,1.49) -- (0.76,1.45) -- (0.71,1.41) -- (0.66,1.37) -- (0.61,1.33) -- (0.56,1.29) -- (0.50,1.26) -- (0.44,1.22) -- (0.38,1.19) -- (0.32,1.15) -- (0.26,1.12) -- (0.19,1.09) -- (0.13,1.06) -- (0.06,1.03) -- (0.00,1.00);
\draw[,] (3.00,1.00) -- (4.00,1.00);
\end{tikzpicture}
    \qeq
    \begin{tikzpicture}[xscale=0.40,yscale=0.40]
\useasboundingbox (-0.5,-0.5) rectangle (4.5,3.5);
\draw[] (0.00,3.00) -- (0.05,3.00) -- (0.10,3.00) -- (0.15,3.00) -- (0.20,2.99) -- (0.25,2.99) -- (0.30,2.99) -- (0.35,2.99) -- (0.40,2.99) -- (0.45,2.99) -- (0.50,2.99) -- (0.55,2.99) -- (0.60,2.99) -- (0.65,2.99) -- (0.70,2.99) -- (0.75,2.99) -- (0.80,2.99) -- (0.85,2.99) -- (0.90,2.99) -- (0.95,3.00) -- (1.00,3.00);
\draw[] (1.00,3.00) -- (1.05,3.00) -- (1.10,3.01) -- (1.15,3.01) -- (1.20,3.02) -- (1.25,3.02) -- (1.30,3.02) -- (1.35,3.03) -- (1.40,3.03) -- (1.45,3.03) -- (1.50,3.04) -- (1.55,3.04) -- (1.60,3.04) -- (1.65,3.04) -- (1.70,3.04) -- (1.75,3.03) -- (1.80,3.03) -- (1.85,3.03) -- (1.90,3.02) -- (1.95,3.01) -- (2.00,3.00);
\draw[] (2.00,3.00) -- (2.06,2.98) -- (2.13,2.97) -- (2.19,2.94) -- (2.25,2.92) -- (2.32,2.89) -- (2.38,2.86) -- (2.44,2.83) -- (2.50,2.79) -- (2.55,2.76) -- (2.61,2.72) -- (2.66,2.67) -- (2.71,2.63) -- (2.75,2.58) -- (2.80,2.53) -- (2.84,2.48) -- (2.87,2.43) -- (2.90,2.37) -- (2.93,2.32) -- (2.96,2.26) -- (2.97,2.20);
\draw[] (2.97,2.20) -- (2.98,2.19) -- (2.98,2.18) -- (2.98,2.17) -- (2.98,2.16) -- (2.99,2.15) -- (2.99,2.14) -- (2.99,2.13) -- (2.99,2.12) -- (2.99,2.11) -- (2.99,2.10) -- (2.99,2.09) -- (3.00,2.08) -- (3.00,2.07) -- (3.00,2.06) -- (3.00,2.05) -- (3.00,2.04) -- (3.00,2.03) -- (3.00,2.02) -- (3.00,2.01) -- (3.00,2.00);
\draw[] (3.00,2.00) -- (3.00,1.99) -- (3.00,1.98) -- (3.00,1.97) -- (3.00,1.96) -- (3.00,1.95) -- (3.00,1.94) -- (3.00,1.93) -- (3.00,1.92) -- (2.99,1.91) -- (2.99,1.90) -- (2.99,1.89) -- (2.99,1.88) -- (2.99,1.87) -- (2.99,1.86) -- (2.99,1.85) -- (2.98,1.84) -- (2.98,1.83) -- (2.98,1.82) -- (2.98,1.81) -- (2.97,1.80);
\draw[] (2.97,1.80) -- (2.96,1.74) -- (2.93,1.68) -- (2.90,1.63) -- (2.87,1.57) -- (2.84,1.52) -- (2.80,1.47) -- (2.75,1.42) -- (2.71,1.37) -- (2.66,1.33) -- (2.61,1.29) -- (2.55,1.25) -- (2.50,1.21) -- (2.44,1.17) -- (2.38,1.14) -- (2.32,1.11) -- (2.25,1.08) -- (2.19,1.06) -- (2.13,1.04) -- (2.06,1.02) -- (2.00,1.00);
\draw[] (2.00,1.00) -- (1.95,0.99) -- (1.90,0.98) -- (1.85,0.97) -- (1.80,0.97) -- (1.75,0.96) -- (1.70,0.96) -- (1.65,0.95) -- (1.60,0.95) -- (1.55,0.95) -- (1.50,0.95) -- (1.45,0.96) -- (1.40,0.96) -- (1.35,0.96) -- (1.30,0.97) -- (1.25,0.97) -- (1.20,0.98) -- (1.15,0.98) -- (1.10,0.99) -- (1.05,0.99) -- (1.00,1.00);
\draw[,] (3.00,2.00) -- (4.00,2.00);
\draw[] (0.00,2.00) -- (0.06,1.97) -- (0.13,1.94) -- (0.19,1.91) -- (0.26,1.88) -- (0.32,1.85) -- (0.38,1.81) -- (0.44,1.78) -- (0.50,1.74) -- (0.56,1.71) -- (0.61,1.67) -- (0.66,1.63) -- (0.71,1.59) -- (0.76,1.55) -- (0.80,1.51) -- (0.84,1.46) -- (0.87,1.41) -- (0.90,1.36) -- (0.93,1.31) -- (0.96,1.26) -- (0.97,1.20);
\draw[] (0.97,1.20) -- (0.98,1.19) -- (0.98,1.18) -- (0.98,1.17) -- (0.98,1.16) -- (0.99,1.15) -- (0.99,1.14) -- (0.99,1.13) -- (0.99,1.12) -- (0.99,1.11) -- (0.99,1.10) -- (0.99,1.09) -- (1.00,1.08) -- (1.00,1.07) -- (1.00,1.06) -- (1.00,1.05) -- (1.00,1.04) -- (1.00,1.03) -- (1.00,1.02) -- (1.00,1.01) -- (1.00,1.00);
\draw[] (1.00,1.00) -- (1.00,0.99) -- (1.00,0.98) -- (1.00,0.97) -- (1.00,0.96) -- (1.00,0.95) -- (1.00,0.94) -- (1.00,0.93) -- (1.00,0.92) -- (0.99,0.91) -- (0.99,0.90) -- (0.99,0.89) -- (0.99,0.88) -- (0.99,0.87) -- (0.99,0.86) -- (0.99,0.85) -- (0.98,0.84) -- (0.98,0.83) -- (0.98,0.82) -- (0.98,0.81) -- (0.97,0.80);
\draw[] (0.97,0.80) -- (0.96,0.74) -- (0.93,0.69) -- (0.90,0.64) -- (0.87,0.59) -- (0.84,0.54) -- (0.80,0.49) -- (0.76,0.45) -- (0.71,0.41) -- (0.66,0.37) -- (0.61,0.33) -- (0.56,0.29) -- (0.50,0.26) -- (0.44,0.22) -- (0.38,0.19) -- (0.32,0.15) -- (0.26,0.12) -- (0.19,0.09) -- (0.13,0.06) -- (0.06,0.03) -- (0.00,0.00);
\end{tikzpicture}
    \\
    \sstrid{mult_unit_l}
    \qeq
    \sstrid{mult_unit_c}
    \qeq
    \sstrid{mult_unit_r}
  \end{array}
\end{myequation}

A \emph{symmetric monoid} is a monoid which admits a symmetry $\gamma:M\otimes
M\to M\otimes M$ which is compatible with the operations of the monoid in the
sense
that the equalities
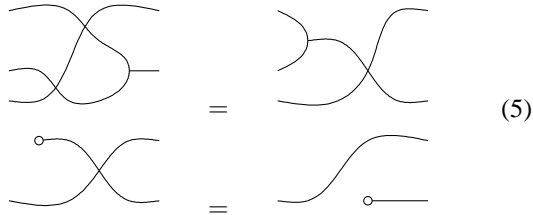
\begin{myequation}
  \label{eq:monoid-nat}
  \begin{array}{cc}
    \begin{tikzpicture}[xscale=0.40,yscale=0.40]
\useasboundingbox (-0.5,-0.5) rectangle (5.5,3.5);
\draw[] (0.00,3.00) -- (0.11,3.02) -- (0.22,3.05) -- (0.34,3.07) -- (0.45,3.09) -- (0.56,3.11) -- (0.67,3.13) -- (0.77,3.15) -- (0.88,3.16) -- (0.99,3.17) -- (1.09,3.18) -- (1.19,3.19) -- (1.29,3.19) -- (1.39,3.18) -- (1.49,3.17) -- (1.58,3.16) -- (1.67,3.14) -- (1.76,3.11) -- (1.84,3.08) -- (1.92,3.04) -- (2.00,3.00);
\draw[] (2.00,3.00) -- (2.05,2.96) -- (2.10,2.92) -- (2.15,2.88) -- (2.20,2.84) -- (2.25,2.79) -- (2.30,2.74) -- (2.34,2.69) -- (2.39,2.64) -- (2.44,2.58) -- (2.48,2.53) -- (2.53,2.47) -- (2.58,2.42) -- (2.62,2.36) -- (2.67,2.31) -- (2.72,2.25) -- (2.77,2.20) -- (2.83,2.15) -- (2.88,2.10) -- (2.94,2.05) -- (3.00,2.00);
\draw[] (3.00,2.00) -- (3.06,1.96) -- (3.11,1.92) -- (3.17,1.89) -- (3.23,1.85) -- (3.29,1.82) -- (3.35,1.78) -- (3.41,1.75) -- (3.47,1.71) -- (3.52,1.68) -- (3.58,1.64) -- (3.63,1.61) -- (3.69,1.57) -- (3.74,1.53) -- (3.78,1.49) -- (3.83,1.45) -- (3.86,1.41) -- (3.90,1.36) -- (3.93,1.31) -- (3.95,1.26) -- (3.97,1.20);
\draw[] (3.97,1.20) -- (3.98,1.19) -- (3.98,1.18) -- (3.98,1.17) -- (3.98,1.16) -- (3.99,1.15) -- (3.99,1.14) -- (3.99,1.13) -- (3.99,1.12) -- (3.99,1.11) -- (3.99,1.10) -- (4.00,1.09) -- (4.00,1.08) -- (4.00,1.07) -- (4.00,1.06) -- (4.00,1.05) -- (4.00,1.04) -- (4.00,1.03) -- (4.00,1.02) -- (4.00,1.01) -- (4.00,1.00);
\draw[] (4.00,1.00) -- (4.00,0.99) -- (4.00,0.98) -- (4.00,0.97) -- (4.00,0.96) -- (4.00,0.95) -- (4.00,0.94) -- (4.00,0.93) -- (4.00,0.92) -- (3.99,0.91) -- (3.99,0.90) -- (3.99,0.89) -- (3.99,0.88) -- (3.99,0.87) -- (3.99,0.86) -- (3.99,0.85) -- (3.98,0.84) -- (3.98,0.83) -- (3.98,0.82) -- (3.98,0.81) -- (3.97,0.80);
\draw[] (3.97,0.80) -- (3.96,0.74) -- (3.93,0.69) -- (3.91,0.63) -- (3.87,0.58) -- (3.84,0.53) -- (3.80,0.48) -- (3.76,0.43) -- (3.71,0.39) -- (3.67,0.34) -- (3.61,0.30) -- (3.56,0.27) -- (3.50,0.23) -- (3.45,0.19) -- (3.39,0.16) -- (3.33,0.13) -- (3.26,0.10) -- (3.20,0.07) -- (3.13,0.05) -- (3.07,0.02) -- (3.00,0.00);
\draw[] (3.00,0.00) -- (2.95,-0.02) -- (2.89,-0.03) -- (2.84,-0.05) -- (2.79,-0.06) -- (2.73,-0.07) -- (2.68,-0.08) -- (2.63,-0.09) -- (2.58,-0.09) -- (2.52,-0.10) -- (2.47,-0.10) -- (2.42,-0.10) -- (2.37,-0.10) -- (2.32,-0.09) -- (2.27,-0.09) -- (2.22,-0.08) -- (2.18,-0.07) -- (2.13,-0.06) -- (2.09,-0.04) -- (2.04,-0.02) -- (2.00,0.00);
\draw[] (2.00,0.00) -- (1.94,0.03) -- (1.89,0.07) -- (1.83,0.12) -- (1.78,0.17) -- (1.73,0.22) -- (1.68,0.27) -- (1.64,0.33) -- (1.59,0.39) -- (1.54,0.45) -- (1.50,0.51) -- (1.45,0.57) -- (1.41,0.63) -- (1.36,0.68) -- (1.31,0.74) -- (1.27,0.79) -- (1.22,0.84) -- (1.16,0.89) -- (1.11,0.93) -- (1.06,0.97) -- (1.00,1.00);
\draw[] (1.00,1.00) -- (0.96,1.02) -- (0.91,1.03) -- (0.87,1.05) -- (0.82,1.06) -- (0.78,1.07) -- (0.73,1.07) -- (0.68,1.08) -- (0.63,1.08) -- (0.58,1.08) -- (0.53,1.08) -- (0.48,1.07) -- (0.43,1.07) -- (0.38,1.06) -- (0.32,1.06) -- (0.27,1.05) -- (0.22,1.04) -- (0.16,1.03) -- (0.11,1.02) -- (0.05,1.01) -- (0.00,1.00);
\draw[] (0.00,0.00) -- (0.05,-0.01) -- (0.11,-0.01) -- (0.16,-0.02) -- (0.21,-0.03) -- (0.27,-0.03) -- (0.32,-0.04) -- (0.37,-0.04) -- (0.42,-0.05) -- (0.48,-0.05) -- (0.53,-0.05) -- (0.58,-0.05) -- (0.63,-0.05) -- (0.68,-0.05) -- (0.73,-0.05) -- (0.77,-0.04) -- (0.82,-0.04) -- (0.87,-0.03) -- (0.91,-0.02) -- (0.96,-0.01) -- (1.00,0.00);
\draw[] (1.00,0.00) -- (1.15,0.06) -- (1.28,0.15) -- (1.40,0.26) -- (1.52,0.38) -- (1.62,0.53) -- (1.71,0.69) -- (1.80,0.86) -- (1.89,1.03) -- (1.97,1.22) -- (2.05,1.41) -- (2.13,1.60) -- (2.20,1.79) -- (2.28,1.98) -- (2.37,2.16) -- (2.45,2.34) -- (2.55,2.50) -- (2.65,2.65) -- (2.75,2.79) -- (2.87,2.90) -- (3.00,3.00);
\draw[] (3.00,3.00) -- (3.08,3.04) -- (3.16,3.08) -- (3.24,3.11) -- (3.33,3.14) -- (3.42,3.16) -- (3.51,3.17) -- (3.61,3.18) -- (3.70,3.18) -- (3.80,3.18) -- (3.91,3.18) -- (4.01,3.17) -- (4.12,3.16) -- (4.22,3.15) -- (4.33,3.13) -- (4.44,3.11) -- (4.55,3.09) -- (4.66,3.07) -- (4.77,3.05) -- (4.89,3.02) -- (5.00,3.00);
\draw[,] (4.00,1.00) -- (5.00,1.00);
\end{tikzpicture}
    \qeq
    \begin{tikzpicture}[xscale=0.40,yscale=0.40]
\useasboundingbox (-0.5,-0.5) rectangle (5.5,3.5);
\draw[] (0.00,0.00) -- (0.11,-0.02) -- (0.21,-0.04) -- (0.32,-0.05) -- (0.42,-0.07) -- (0.53,-0.09) -- (0.63,-0.10) -- (0.74,-0.11) -- (0.84,-0.12) -- (0.94,-0.13) -- (1.04,-0.14) -- (1.15,-0.14) -- (1.25,-0.14) -- (1.34,-0.14) -- (1.44,-0.13) -- (1.54,-0.12) -- (1.63,-0.11) -- (1.73,-0.09) -- (1.82,-0.06) -- (1.91,-0.03) -- (2.00,0.00);
\draw[] (2.00,0.00) -- (2.06,0.03) -- (2.12,0.06) -- (2.18,0.09) -- (2.24,0.13) -- (2.30,0.16) -- (2.36,0.20) -- (2.41,0.25) -- (2.47,0.29) -- (2.52,0.34) -- (2.57,0.39) -- (2.62,0.44) -- (2.67,0.50) -- (2.72,0.55) -- (2.76,0.61) -- (2.81,0.67) -- (2.85,0.73) -- (2.89,0.80) -- (2.93,0.86) -- (2.96,0.93) -- (3.00,1.00);
\draw[] (3.00,1.00) -- (3.05,1.11) -- (3.10,1.23) -- (3.15,1.35) -- (3.19,1.47) -- (3.23,1.59) -- (3.27,1.71) -- (3.30,1.84) -- (3.34,1.96) -- (3.38,2.08) -- (3.42,2.19) -- (3.46,2.30) -- (3.50,2.41) -- (3.55,2.51) -- (3.59,2.60) -- (3.65,2.69) -- (3.71,2.77) -- (3.77,2.84) -- (3.84,2.91) -- (3.92,2.96) -- (4.00,3.00);
\draw[] (4.00,3.00) -- (4.04,3.01) -- (4.08,3.03) -- (4.12,3.04) -- (4.17,3.05) -- (4.21,3.05) -- (4.26,3.06) -- (4.31,3.06) -- (4.36,3.06) -- (4.41,3.06) -- (4.46,3.06) -- (4.51,3.06) -- (4.56,3.05) -- (4.62,3.05) -- (4.67,3.04) -- (4.72,3.04) -- (4.78,3.03) -- (4.83,3.02) -- (4.89,3.02) -- (4.94,3.01) -- (5.00,3.00);
\draw[] (1.00,2.00) -- (1.05,2.01) -- (1.11,2.01) -- (1.16,2.02) -- (1.21,2.03) -- (1.26,2.03) -- (1.32,2.04) -- (1.37,2.04) -- (1.42,2.05) -- (1.47,2.05) -- (1.52,2.05) -- (1.57,2.06) -- (1.62,2.06) -- (1.67,2.05) -- (1.72,2.05) -- (1.77,2.05) -- (1.82,2.04) -- (1.86,2.03) -- (1.91,2.03) -- (1.96,2.01) -- (2.00,2.00);
\draw[] (2.00,2.00) -- (2.06,1.98) -- (2.12,1.95) -- (2.18,1.92) -- (2.23,1.88) -- (2.29,1.85) -- (2.34,1.80) -- (2.40,1.76) -- (2.45,1.71) -- (2.50,1.66) -- (2.55,1.61) -- (2.59,1.56) -- (2.64,1.50) -- (2.69,1.44) -- (2.73,1.38) -- (2.78,1.32) -- (2.82,1.26) -- (2.87,1.19) -- (2.91,1.13) -- (2.96,1.06) -- (3.00,1.00);
\draw[] (3.00,1.00) -- (3.04,0.94) -- (3.09,0.87) -- (3.13,0.81) -- (3.18,0.74) -- (3.22,0.68) -- (3.27,0.62) -- (3.31,0.56) -- (3.36,0.50) -- (3.41,0.44) -- (3.45,0.39) -- (3.50,0.34) -- (3.55,0.29) -- (3.60,0.24) -- (3.66,0.20) -- (3.71,0.15) -- (3.77,0.12) -- (3.82,0.08) -- (3.88,0.05) -- (3.94,0.02) -- (4.00,0.00);
\draw[] (4.00,0.00) -- (4.04,-0.01) -- (4.09,-0.03) -- (4.14,-0.03) -- (4.18,-0.04) -- (4.23,-0.05) -- (4.28,-0.05) -- (4.33,-0.05) -- (4.38,-0.06) -- (4.43,-0.06) -- (4.48,-0.05) -- (4.53,-0.05) -- (4.58,-0.05) -- (4.63,-0.04) -- (4.68,-0.04) -- (4.74,-0.03) -- (4.79,-0.03) -- (4.84,-0.02) -- (4.89,-0.01) -- (4.95,-0.01) -- (5.00,0.00);
\draw[] (0.00,1.00) -- (0.06,1.03) -- (0.13,1.06) -- (0.19,1.09) -- (0.26,1.12) -- (0.32,1.15) -- (0.38,1.19) -- (0.44,1.22) -- (0.50,1.26) -- (0.56,1.29) -- (0.61,1.33) -- (0.66,1.37) -- (0.71,1.41) -- (0.76,1.45) -- (0.80,1.49) -- (0.84,1.54) -- (0.87,1.59) -- (0.90,1.64) -- (0.93,1.69) -- (0.96,1.74) -- (0.97,1.80);
\draw[] (0.97,1.80) -- (0.98,1.81) -- (0.98,1.82) -- (0.98,1.83) -- (0.98,1.84) -- (0.99,1.85) -- (0.99,1.86) -- (0.99,1.87) -- (0.99,1.88) -- (0.99,1.89) -- (0.99,1.90) -- (0.99,1.91) -- (1.00,1.92) -- (1.00,1.93) -- (1.00,1.94) -- (1.00,1.95) -- (1.00,1.96) -- (1.00,1.97) -- (1.00,1.98) -- (1.00,1.99) -- (1.00,2.00);
\draw[] (1.00,2.00) -- (1.00,2.01) -- (1.00,2.02) -- (1.00,2.03) -- (1.00,2.04) -- (1.00,2.05) -- (1.00,2.06) -- (1.00,2.07) -- (1.00,2.08) -- (0.99,2.09) -- (0.99,2.10) -- (0.99,2.11) -- (0.99,2.12) -- (0.99,2.13) -- (0.99,2.14) -- (0.99,2.15) -- (0.98,2.16) -- (0.98,2.17) -- (0.98,2.18) -- (0.98,2.19) -- (0.97,2.20);
\draw[] (0.97,2.20) -- (0.96,2.26) -- (0.93,2.31) -- (0.90,2.36) -- (0.87,2.41) -- (0.84,2.46) -- (0.80,2.51) -- (0.76,2.55) -- (0.71,2.59) -- (0.66,2.63) -- (0.61,2.67) -- (0.56,2.71) -- (0.50,2.74) -- (0.44,2.78) -- (0.38,2.81) -- (0.32,2.85) -- (0.26,2.88) -- (0.19,2.91) -- (0.13,2.94) -- (0.06,2.97) -- (0.00,3.00);
\end{tikzpicture}
    \\
    \begin{tikzpicture}[xscale=0.40,yscale=0.40]
\useasboundingbox (-0.5,-0.5) rectangle (5.5,2.5);
\draw[] (1.00,2.00) -- (1.05,2.01) -- (1.11,2.01) -- (1.16,2.02) -- (1.21,2.03) -- (1.26,2.03) -- (1.32,2.04) -- (1.37,2.04) -- (1.42,2.05) -- (1.47,2.05) -- (1.52,2.05) -- (1.57,2.06) -- (1.62,2.06) -- (1.67,2.05) -- (1.72,2.05) -- (1.77,2.05) -- (1.82,2.04) -- (1.86,2.03) -- (1.91,2.03) -- (1.96,2.01) -- (2.00,2.00);
\draw[] (2.00,2.00) -- (2.06,1.98) -- (2.12,1.95) -- (2.18,1.92) -- (2.23,1.88) -- (2.29,1.85) -- (2.34,1.80) -- (2.40,1.76) -- (2.45,1.71) -- (2.50,1.66) -- (2.55,1.61) -- (2.59,1.56) -- (2.64,1.50) -- (2.69,1.44) -- (2.73,1.38) -- (2.78,1.32) -- (2.82,1.26) -- (2.87,1.19) -- (2.91,1.13) -- (2.96,1.06) -- (3.00,1.00);
\draw[] (3.00,1.00) -- (3.04,0.94) -- (3.09,0.87) -- (3.13,0.81) -- (3.18,0.74) -- (3.22,0.68) -- (3.27,0.62) -- (3.31,0.56) -- (3.36,0.50) -- (3.41,0.44) -- (3.45,0.39) -- (3.50,0.34) -- (3.55,0.29) -- (3.60,0.24) -- (3.66,0.20) -- (3.71,0.15) -- (3.77,0.12) -- (3.82,0.08) -- (3.88,0.05) -- (3.94,0.02) -- (4.00,0.00);
\draw[] (4.00,0.00) -- (4.04,-0.01) -- (4.09,-0.03) -- (4.14,-0.03) -- (4.18,-0.04) -- (4.23,-0.05) -- (4.28,-0.05) -- (4.33,-0.05) -- (4.38,-0.06) -- (4.43,-0.06) -- (4.48,-0.05) -- (4.53,-0.05) -- (4.58,-0.05) -- (4.63,-0.04) -- (4.68,-0.04) -- (4.74,-0.03) -- (4.79,-0.03) -- (4.84,-0.02) -- (4.89,-0.01) -- (4.95,-0.01) -- (5.00,0.00);
\draw[] (0.00,0.00) -- (0.11,-0.02) -- (0.22,-0.04) -- (0.32,-0.06) -- (0.43,-0.08) -- (0.54,-0.09) -- (0.65,-0.11) -- (0.75,-0.12) -- (0.86,-0.14) -- (0.96,-0.15) -- (1.06,-0.15) -- (1.16,-0.16) -- (1.26,-0.16) -- (1.36,-0.15) -- (1.46,-0.14) -- (1.55,-0.13) -- (1.65,-0.12) -- (1.74,-0.10) -- (1.83,-0.07) -- (1.92,-0.04) -- (2.00,0.00);
\draw[] (2.00,0.00) -- (2.06,0.03) -- (2.12,0.06) -- (2.17,0.10) -- (2.23,0.14) -- (2.28,0.18) -- (2.33,0.22) -- (2.38,0.27) -- (2.44,0.31) -- (2.49,0.36) -- (2.54,0.42) -- (2.58,0.47) -- (2.63,0.52) -- (2.68,0.58) -- (2.73,0.64) -- (2.77,0.69) -- (2.82,0.75) -- (2.86,0.81) -- (2.91,0.88) -- (2.96,0.94) -- (3.00,1.00);
\draw[] (3.00,1.00) -- (3.04,1.06) -- (3.09,1.13) -- (3.13,1.19) -- (3.18,1.25) -- (3.22,1.31) -- (3.27,1.37) -- (3.32,1.43) -- (3.36,1.49) -- (3.41,1.55) -- (3.46,1.60) -- (3.51,1.66) -- (3.56,1.71) -- (3.61,1.76) -- (3.66,1.80) -- (3.71,1.84) -- (3.77,1.88) -- (3.82,1.92) -- (3.88,1.95) -- (3.94,1.98) -- (4.00,2.00);
\draw[] (4.00,2.00) -- (4.04,2.01) -- (4.09,2.03) -- (4.14,2.04) -- (4.18,2.04) -- (4.23,2.05) -- (4.28,2.05) -- (4.33,2.06) -- (4.38,2.06) -- (4.43,2.06) -- (4.48,2.06) -- (4.53,2.05) -- (4.58,2.05) -- (4.63,2.05) -- (4.68,2.04) -- (4.74,2.04) -- (4.79,2.03) -- (4.84,2.02) -- (4.89,2.01) -- (4.95,2.01) -- (5.00,2.00);
\filldraw[fill=white] (1.00,2.00) ellipse (0.14cm and 0.14cm);
\end{tikzpicture}
    \qeq
    \begin{tikzpicture}[xscale=0.40,yscale=0.40]
\useasboundingbox (-0.5,-0.5) rectangle (5.5,2.5);
\draw[] (0.00,0.00) -- (0.05,-0.01) -- (0.11,-0.01) -- (0.16,-0.02) -- (0.21,-0.03) -- (0.26,-0.03) -- (0.32,-0.04) -- (0.37,-0.04) -- (0.42,-0.04) -- (0.47,-0.05) -- (0.52,-0.05) -- (0.57,-0.05) -- (0.62,-0.05) -- (0.67,-0.05) -- (0.72,-0.05) -- (0.77,-0.04) -- (0.82,-0.04) -- (0.86,-0.03) -- (0.91,-0.02) -- (0.96,-0.01) -- (1.00,0.00);
\draw[] (1.00,0.00) -- (1.12,0.05) -- (1.24,0.11) -- (1.35,0.18) -- (1.46,0.26) -- (1.56,0.36) -- (1.66,0.46) -- (1.75,0.57) -- (1.85,0.69) -- (1.94,0.81) -- (2.03,0.94) -- (2.12,1.06) -- (2.21,1.19) -- (2.30,1.31) -- (2.39,1.43) -- (2.48,1.54) -- (2.58,1.65) -- (2.68,1.76) -- (2.78,1.85) -- (2.89,1.93) -- (3.00,2.00);
\draw[] (3.00,2.00) -- (3.08,2.04) -- (3.17,2.08) -- (3.26,2.11) -- (3.35,2.13) -- (3.44,2.15) -- (3.53,2.16) -- (3.63,2.17) -- (3.73,2.17) -- (3.83,2.17) -- (3.93,2.17) -- (4.03,2.16) -- (4.14,2.15) -- (4.24,2.14) -- (4.35,2.12) -- (4.46,2.11) -- (4.56,2.09) -- (4.67,2.07) -- (4.78,2.04) -- (4.89,2.02) -- (5.00,2.00);
\draw[,] (3.00,0.00) -- (5.00,0.00);
\filldraw[fill=white] (3.00,0.00) ellipse (0.14cm and 0.14cm);
\end{tikzpicture}
  \end{array}
\end{myequation}
are satisfied, as well as the equations obtained by turning the diagrams
upside-down. A \emph{commutative monoid} is a symmetric monoid such that
the equality
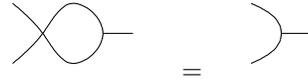
\begin{mymath}
  \begin{tikzpicture}[xscale=0.40,yscale=0.40]
\useasboundingbox (-0.5,-0.5) rectangle (4.5,2.5);
\draw[] (0.00,2.00) -- (0.05,1.96) -- (0.10,1.91) -- (0.16,1.87) -- (0.21,1.83) -- (0.26,1.78) -- (0.31,1.74) -- (0.36,1.69) -- (0.41,1.64) -- (0.46,1.60) -- (0.51,1.55) -- (0.56,1.50) -- (0.61,1.45) -- (0.66,1.40) -- (0.71,1.35) -- (0.76,1.29) -- (0.81,1.24) -- (0.86,1.18) -- (0.91,1.12) -- (0.95,1.06) -- (1.00,1.00);
\draw[] (1.00,1.00) -- (1.05,0.94) -- (1.09,0.87) -- (1.14,0.80) -- (1.18,0.74) -- (1.23,0.67) -- (1.27,0.60) -- (1.32,0.54) -- (1.36,0.47) -- (1.41,0.41) -- (1.46,0.35) -- (1.51,0.29) -- (1.56,0.24) -- (1.61,0.19) -- (1.66,0.15) -- (1.71,0.11) -- (1.76,0.07) -- (1.82,0.04) -- (1.88,0.02) -- (1.94,0.01) -- (2.00,0.00);
\draw[] (2.00,0.00) -- (2.06,0.00) -- (2.11,0.01) -- (2.17,0.02) -- (2.23,0.04) -- (2.29,0.06) -- (2.35,0.09) -- (2.41,0.12) -- (2.47,0.15) -- (2.53,0.19) -- (2.59,0.24) -- (2.64,0.28) -- (2.69,0.33) -- (2.74,0.39) -- (2.79,0.44) -- (2.83,0.50) -- (2.87,0.56) -- (2.90,0.62) -- (2.93,0.68) -- (2.95,0.74) -- (2.97,0.80);
\draw[] (2.97,0.80) -- (2.98,0.81) -- (2.98,0.82) -- (2.98,0.83) -- (2.98,0.84) -- (2.99,0.85) -- (2.99,0.86) -- (2.99,0.87) -- (2.99,0.88) -- (2.99,0.89) -- (2.99,0.90) -- (2.99,0.91) -- (3.00,0.92) -- (3.00,0.93) -- (3.00,0.94) -- (3.00,0.95) -- (3.00,0.96) -- (3.00,0.97) -- (3.00,0.98) -- (3.00,0.99) -- (3.00,1.00);
\draw[] (3.00,1.00) -- (3.00,1.01) -- (3.00,1.02) -- (3.00,1.03) -- (3.00,1.04) -- (3.00,1.05) -- (3.00,1.06) -- (3.00,1.07) -- (3.00,1.08) -- (2.99,1.09) -- (2.99,1.10) -- (2.99,1.11) -- (2.99,1.12) -- (2.99,1.13) -- (2.99,1.14) -- (2.99,1.15) -- (2.98,1.16) -- (2.98,1.17) -- (2.98,1.18) -- (2.98,1.19) -- (2.97,1.20);
\draw[] (2.97,1.20) -- (2.95,1.26) -- (2.93,1.32) -- (2.90,1.38) -- (2.87,1.44) -- (2.83,1.50) -- (2.79,1.56) -- (2.74,1.61) -- (2.69,1.67) -- (2.64,1.72) -- (2.59,1.76) -- (2.53,1.81) -- (2.47,1.85) -- (2.41,1.88) -- (2.35,1.91) -- (2.29,1.94) -- (2.23,1.96) -- (2.17,1.98) -- (2.11,1.99) -- (2.06,2.00) -- (2.00,2.00);
\draw[] (2.00,2.00) -- (1.94,1.99) -- (1.88,1.98) -- (1.82,1.96) -- (1.76,1.93) -- (1.71,1.89) -- (1.66,1.85) -- (1.61,1.81) -- (1.56,1.76) -- (1.51,1.71) -- (1.46,1.65) -- (1.41,1.59) -- (1.36,1.53) -- (1.32,1.46) -- (1.27,1.40) -- (1.23,1.33) -- (1.18,1.26) -- (1.14,1.20) -- (1.09,1.13) -- (1.05,1.06) -- (1.00,1.00);
\draw[] (1.00,1.00) -- (0.95,0.94) -- (0.91,0.88) -- (0.86,0.82) -- (0.81,0.76) -- (0.76,0.71) -- (0.71,0.65) -- (0.66,0.60) -- (0.61,0.55) -- (0.56,0.50) -- (0.51,0.45) -- (0.46,0.40) -- (0.41,0.36) -- (0.36,0.31) -- (0.31,0.26) -- (0.26,0.22) -- (0.21,0.17) -- (0.16,0.13) -- (0.10,0.09) -- (0.05,0.04) -- (0.00,0.00);
\draw[,] (3.00,1.00) -- (4.00,1.00);
\end{tikzpicture}
  \qeq
  \begin{tikzpicture}[xscale=0.40,yscale=0.40]
\useasboundingbox (-0.5,-0.5) rectangle (2.5,2.5);
\draw[,] (1.00,1.00) -- (2.00,1.00);
\draw[] (0.00,2.00) -- (0.06,1.97) -- (0.13,1.94) -- (0.19,1.91) -- (0.26,1.88) -- (0.32,1.85) -- (0.38,1.81) -- (0.44,1.78) -- (0.50,1.74) -- (0.56,1.71) -- (0.61,1.67) -- (0.66,1.63) -- (0.71,1.59) -- (0.76,1.55) -- (0.80,1.51) -- (0.84,1.46) -- (0.87,1.41) -- (0.90,1.36) -- (0.93,1.31) -- (0.96,1.26) -- (0.97,1.20);
\draw[] (0.97,1.20) -- (0.98,1.19) -- (0.98,1.18) -- (0.98,1.17) -- (0.98,1.16) -- (0.99,1.15) -- (0.99,1.14) -- (0.99,1.13) -- (0.99,1.12) -- (0.99,1.11) -- (0.99,1.10) -- (0.99,1.09) -- (1.00,1.08) -- (1.00,1.07) -- (1.00,1.06) -- (1.00,1.05) -- (1.00,1.04) -- (1.00,1.03) -- (1.00,1.02) -- (1.00,1.01) -- (1.00,1.00);
\draw[] (1.00,1.00) -- (1.00,0.99) -- (1.00,0.98) -- (1.00,0.97) -- (1.00,0.96) -- (1.00,0.95) -- (1.00,0.94) -- (1.00,0.93) -- (1.00,0.92) -- (0.99,0.91) -- (0.99,0.90) -- (0.99,0.89) -- (0.99,0.88) -- (0.99,0.87) -- (0.99,0.86) -- (0.99,0.85) -- (0.98,0.84) -- (0.98,0.83) -- (0.98,0.82) -- (0.98,0.81) -- (0.97,0.80);
\draw[] (0.97,0.80) -- (0.96,0.74) -- (0.93,0.69) -- (0.90,0.64) -- (0.87,0.59) -- (0.84,0.54) -- (0.80,0.49) -- (0.76,0.45) -- (0.71,0.41) -- (0.66,0.37) -- (0.61,0.33) -- (0.56,0.29) -- (0.50,0.26) -- (0.44,0.22) -- (0.38,0.19) -- (0.32,0.15) -- (0.26,0.12) -- (0.19,0.09) -- (0.13,0.06) -- (0.06,0.03) -- (0.00,0.00);
\end{tikzpicture}
\end{mymath}
is satisfied. In particular, a commutative monoid in a symmetric monoidal
category is a commutative monoid whose symmetry corresponds to the symmetry of
the category: $\gamma=\gamma_{M,M}$. In this case, the equations
\eqref{eq:monoid-nat} can always be deduced from the naturality of the symmetry
of the monoidal category.

A \emph{comonoid} $(M,\delta,\varepsilon)$ in $\mathcal{C}$ is an object $M$
together with two morphisms
\begin{mymath}
\delta:M\to M\otimes M
\qtand
\varepsilon:M\to I
\end{mymath}
respectively drawn as
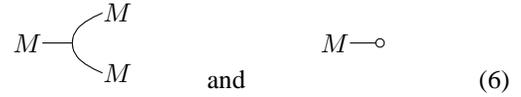
\begin{myequation}
  \label{eq:comonoid-string}
  \begin{tikzpicture}[xscale=0.40,yscale=0.40]
\useasboundingbox (-0.5,-0.5) rectangle (4.5,2.5);
\draw[,] (2.00,1.00) -- (1.00,1.00);
\draw[] (3.00,2.00) -- (2.94,1.97) -- (2.87,1.94) -- (2.81,1.91) -- (2.74,1.88) -- (2.68,1.85) -- (2.62,1.81) -- (2.56,1.78) -- (2.50,1.74) -- (2.44,1.71) -- (2.39,1.67) -- (2.34,1.63) -- (2.29,1.59) -- (2.24,1.55) -- (2.20,1.51) -- (2.16,1.46) -- (2.13,1.41) -- (2.10,1.36) -- (2.07,1.31) -- (2.04,1.26) -- (2.03,1.20);
\draw[] (2.03,1.20) -- (2.02,1.19) -- (2.02,1.18) -- (2.02,1.17) -- (2.02,1.16) -- (2.01,1.15) -- (2.01,1.14) -- (2.01,1.13) -- (2.01,1.12) -- (2.01,1.11) -- (2.01,1.10) -- (2.01,1.09) -- (2.00,1.08) -- (2.00,1.07) -- (2.00,1.06) -- (2.00,1.05) -- (2.00,1.04) -- (2.00,1.03) -- (2.00,1.02) -- (2.00,1.01) -- (2.00,1.00);
\draw[] (2.00,1.00) -- (2.00,0.99) -- (2.00,0.98) -- (2.00,0.97) -- (2.00,0.96) -- (2.00,0.95) -- (2.00,0.94) -- (2.00,0.93) -- (2.00,0.92) -- (2.01,0.91) -- (2.01,0.90) -- (2.01,0.89) -- (2.01,0.88) -- (2.01,0.87) -- (2.01,0.86) -- (2.01,0.85) -- (2.02,0.84) -- (2.02,0.83) -- (2.02,0.82) -- (2.02,0.81) -- (2.03,0.80);
\draw[] (2.03,0.80) -- (2.04,0.74) -- (2.07,0.69) -- (2.10,0.64) -- (2.13,0.59) -- (2.16,0.54) -- (2.20,0.49) -- (2.24,0.45) -- (2.29,0.41) -- (2.34,0.37) -- (2.39,0.33) -- (2.44,0.29) -- (2.50,0.26) -- (2.56,0.22) -- (2.62,0.19) -- (2.68,0.15) -- (2.74,0.12) -- (2.81,0.09) -- (2.87,0.06) -- (2.94,0.03) -- (3.00,0.00);
\draw (3.50,2.00) node{$M$};
\draw (0.50,1.00) node{$M$};
\draw (3.50,0.00) node{$M$};
\end{tikzpicture}
  \qqtand
  \begin{tikzpicture}[xscale=0.40,yscale=0.40]
\useasboundingbox (-0.5,-0.5) rectangle (2.5,2.5);
\draw[,] (2.00,1.00) -- (1.00,1.00);
\filldraw[fill=white] (2.00,1.00) ellipse (0.14cm and 0.14cm);
\draw (0.50,1.00) node{$M$};
\end{tikzpicture}

\end{myequation}
satisfying dual coherence diagrams. Similarly, the notions symmetric comonoid
and cocommutative comonoid can be defined by duality.

The definition of a monoid can be reformulated internally, in the language of
equational theories:
\begin{definition}
  \label{definition:e-t-monoid}
  The \emph{equational theory of monoids} $\eqth{M}$ has one object $1$ and two
  generators $\mu:2\to 1$ and $\eta:0\to 1$ subject to the three relations
  \begin{myequation}
    \label{eq:monoid-theory}
    \begin{array}{c}
      \mu\circ(\mu\otimes\id_1)
      \qeq
      \mu\circ(\id_1\otimes\mu)
      \\
      \mu\circ(\eta\otimes\id_1)
      \qeq
      \id_1
      \qeq
      \mu\circ(\id_1\otimes\eta)
    \end{array}
  \end{myequation}
\end{definition}
\vspace\reduce
\noindent
The equations~\eqref{eq:monoid-theory} are the algebraic formulation of the
equations~\eqref{eq:monoid}. If we write $\mathbb{M}$ for the monoidal category
generated by the equational theory $\eqth{M}$, the algebras of $\mathbb{M}$ in a
strict monoidal category~$\mathcal{C}$ are precisely its monoids: the category
$\Alg{\mathbb{M}}{\mathcal{C}}$ of algebras of the monoidal theory $\mathbb{M}$
in $\mathcal{C}$ is monoidally equivalent to the category of monoids in
$\mathcal{C}$. Similarly, all the algebraic structures introduced in this
section can be defined using algebraic theories.


\paragraph{Bialgebras.}
A \emph{bialgebra} $(B,\mu,\eta,\delta,\varepsilon,\gamma)$ in $\mathcal{C}$ is
an object $B$ together with four morphisms
\vspace\reduce
\[
\begin{array}{r@{\ :\ }l@{\qquad}r@{\ :\ }l}
  \hspace{-1ex}
  \mu&B\otimes B\to B
  &
  \eta&I\to B\\
  \delta&B\to B\otimes B
  &
  \varepsilon&B\to I\\
\end{array}
\ \tand\ 
\gamma : B \otimes B\to B\otimes B
\vspace\reduce
\]
such that $\gamma:B\otimes B\to B\otimes B$ is a symmetry for $B$,
$(B,\mu,\eta,\gamma)$ is a symmetric monoid and $(B,\delta,\varepsilon,\gamma)$
is a symmetric comonoid. The morphism $\gamma$ is thus pictured as
in~\eqref{eq:sym-string}, $\mu$ and $\eta$ as in \eqref{eq:monoid-string},
and~$\delta$ and~$\varepsilon$ as in~\eqref{eq:comonoid-string}. Those two
structures should be coherent, in the sense
that the four equalities
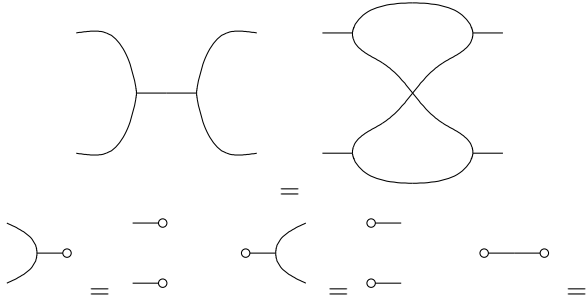
\begin{mymath}
\hspace{-3ex}
\begin{array}{c}
  \sstrid{hopf_l}
  =
  \sstrid{hopf_r}
  \\
  \sstrid{counit_mult}
  =
  \sstrid{counit_x_counit}
  \qquad
  \sstrid{comult_unit}
  =
  \sstrid{unit_x_unit}
  \qquad
  \sstrid{unit_counit}=
\end{array}
\end{mymath}
should be satisfied.

A bialgebra is \emph{commutative} (\resp \emph{cocommutative}) when the induced
symmetric monoid $(B,\mu,\eta,\gamma)$ (\resp symmetric comonoid
$(B,\delta,\varepsilon,\gamma)$) is commutative (\resp cocommutative), and
\emph{bicommutative} when it is both commutative and cocommutative. A bialgebra
is \emph{qualitative}
when the following equality holds:
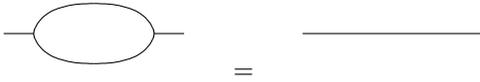
\begin{mymath}
\begin{tikzpicture}[xscale=0.40,yscale=0.40]
\useasboundingbox (-0.5,-0.5) rectangle (6.5,2.5);
\draw[,] (1.00,1.00) -- (0.00,1.00);
\draw[] (3.00,2.00) -- (2.88,2.00) -- (2.75,1.99) -- (2.63,1.99) -- (2.51,1.98) -- (2.39,1.97) -- (2.27,1.95) -- (2.15,1.93) -- (2.04,1.90) -- (1.93,1.87) -- (1.82,1.84) -- (1.72,1.80) -- (1.62,1.76) -- (1.53,1.71) -- (1.44,1.65) -- (1.36,1.59) -- (1.29,1.53) -- (1.22,1.45) -- (1.15,1.37) -- (1.10,1.29) -- (1.05,1.19);
\draw[] (1.05,1.19) -- (1.05,1.18) -- (1.04,1.17) -- (1.04,1.17) -- (1.03,1.16) -- (1.03,1.15) -- (1.03,1.14) -- (1.02,1.13) -- (1.02,1.12) -- (1.02,1.11) -- (1.01,1.10) -- (1.01,1.09) -- (1.01,1.08) -- (1.01,1.07) -- (1.01,1.06) -- (1.00,1.05) -- (1.00,1.04) -- (1.00,1.03) -- (1.00,1.02) -- (1.00,1.01) -- (1.00,1.00);
\draw[] (1.00,1.00) -- (1.00,0.99) -- (1.00,0.98) -- (1.00,0.97) -- (1.00,0.96) -- (1.00,0.95) -- (1.01,0.94) -- (1.01,0.93) -- (1.01,0.92) -- (1.01,0.91) -- (1.01,0.90) -- (1.02,0.89) -- (1.02,0.88) -- (1.02,0.87) -- (1.03,0.86) -- (1.03,0.85) -- (1.03,0.84) -- (1.04,0.83) -- (1.04,0.83) -- (1.05,0.82) -- (1.05,0.81);
\draw[] (1.05,0.81) -- (1.10,0.71) -- (1.15,0.63) -- (1.22,0.55) -- (1.29,0.47) -- (1.36,0.41) -- (1.44,0.35) -- (1.53,0.29) -- (1.62,0.24) -- (1.72,0.20) -- (1.82,0.16) -- (1.93,0.13) -- (2.04,0.10) -- (2.15,0.07) -- (2.27,0.05) -- (2.39,0.03) -- (2.51,0.02) -- (2.63,0.01) -- (2.75,0.01) -- (2.88,0.00) -- (3.00,0.00);
\draw[] (3.00,0.00) -- (3.12,0.00) -- (3.25,0.01) -- (3.37,0.01) -- (3.49,0.02) -- (3.61,0.03) -- (3.73,0.05) -- (3.85,0.07) -- (3.96,0.10) -- (4.07,0.13) -- (4.18,0.16) -- (4.28,0.20) -- (4.38,0.24) -- (4.47,0.29) -- (4.56,0.35) -- (4.64,0.41) -- (4.71,0.47) -- (4.78,0.55) -- (4.85,0.63) -- (4.90,0.71) -- (4.95,0.81);
\draw[] (4.95,0.81) -- (4.95,0.82) -- (4.96,0.83) -- (4.96,0.83) -- (4.97,0.84) -- (4.97,0.85) -- (4.97,0.86) -- (4.98,0.87) -- (4.98,0.88) -- (4.98,0.89) -- (4.99,0.90) -- (4.99,0.91) -- (4.99,0.92) -- (4.99,0.93) -- (4.99,0.94) -- (5.00,0.95) -- (5.00,0.96) -- (5.00,0.97) -- (5.00,0.98) -- (5.00,0.99) -- (5.00,1.00);
\draw[] (5.00,1.00) -- (5.00,1.01) -- (5.00,1.02) -- (5.00,1.03) -- (5.00,1.04) -- (5.00,1.05) -- (4.99,1.06) -- (4.99,1.07) -- (4.99,1.08) -- (4.99,1.09) -- (4.99,1.10) -- (4.98,1.11) -- (4.98,1.12) -- (4.98,1.13) -- (4.97,1.14) -- (4.97,1.15) -- (4.97,1.16) -- (4.96,1.17) -- (4.96,1.17) -- (4.95,1.18) -- (4.95,1.19);
\draw[] (4.95,1.19) -- (4.90,1.29) -- (4.85,1.37) -- (4.78,1.45) -- (4.71,1.53) -- (4.64,1.59) -- (4.56,1.65) -- (4.47,1.71) -- (4.38,1.76) -- (4.28,1.80) -- (4.18,1.84) -- (4.07,1.87) -- (3.96,1.90) -- (3.85,1.93) -- (3.73,1.95) -- (3.61,1.97) -- (3.49,1.98) -- (3.37,1.99) -- (3.25,1.99) -- (3.12,2.00) -- (3.00,2.00);
\draw[,] (5.00,1.00) -- (6.00,1.00);
\end{tikzpicture}
\qeq
\begin{tikzpicture}[xscale=0.40,yscale=0.40]
\useasboundingbox (-0.5,-0.5) rectangle (6.5,2.5);
\draw[,] (0.00,1.00) -- (6.00,1.00);
\end{tikzpicture}

\end{mymath}



We write $\eqth{B}$ for the equational theory of bicommutative
bialgebras and $\eqth{R}$ for the equational theory of qualitative bicommutative
bialgebras.

\paragraph{Dual objects.}
An object $L$ of $\mathcal{C}$ is said to be \emph{left dual} to an object $R$
when there exists two morphisms
\begin{mymath}
\eta:I\to R\otimes L
\qtand
\varepsilon:L\otimes R\to I
\end{mymath}
called respectively the \emph{unit} and the \emph{counit} of the duality and
respectively pictured as
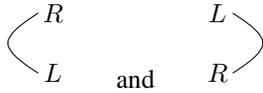
\begin{mymath}
\begin{tikzpicture}[xscale=0.40,yscale=0.40]
\useasboundingbox (-0.5,-0.5) rectangle (2.5,2.5);
\draw[] (1.00,2.00) -- (0.93,1.95) -- (0.85,1.90) -- (0.78,1.85) -- (0.70,1.80) -- (0.63,1.75) -- (0.56,1.70) -- (0.50,1.65) -- (0.43,1.60) -- (0.37,1.55) -- (0.31,1.50) -- (0.26,1.45) -- (0.21,1.40) -- (0.16,1.35) -- (0.12,1.30) -- (0.09,1.25) -- (0.06,1.20) -- (0.03,1.15) -- (0.01,1.10) -- (0.00,1.05) -- (0.00,1.00);
\draw[] (0.00,1.00) -- (0.00,0.95) -- (0.01,0.90) -- (0.03,0.85) -- (0.06,0.80) -- (0.09,0.75) -- (0.12,0.70) -- (0.16,0.65) -- (0.21,0.60) -- (0.26,0.55) -- (0.31,0.50) -- (0.37,0.45) -- (0.43,0.40) -- (0.50,0.35) -- (0.56,0.30) -- (0.63,0.25) -- (0.70,0.20) -- (0.78,0.15) -- (0.85,0.10) -- (0.93,0.05) -- (1.00,0.00);
\draw (1.50,2.00) node{$R$};
\draw (1.50,0.00) node{$L$};
\end{tikzpicture}
\qtand
\begin{tikzpicture}[xscale=0.40,yscale=0.40]
\useasboundingbox (-0.5,-0.5) rectangle (2.5,2.5);
\draw[] (1.00,2.00) -- (1.07,1.95) -- (1.15,1.90) -- (1.22,1.85) -- (1.30,1.80) -- (1.37,1.75) -- (1.44,1.70) -- (1.50,1.65) -- (1.57,1.60) -- (1.63,1.55) -- (1.69,1.50) -- (1.74,1.45) -- (1.79,1.40) -- (1.84,1.35) -- (1.88,1.30) -- (1.91,1.25) -- (1.94,1.20) -- (1.97,1.15) -- (1.99,1.10) -- (2.00,1.05) -- (2.00,1.00);
\draw[] (2.00,1.00) -- (2.00,0.95) -- (1.99,0.90) -- (1.97,0.85) -- (1.94,0.80) -- (1.91,0.75) -- (1.88,0.70) -- (1.84,0.65) -- (1.79,0.60) -- (1.74,0.55) -- (1.69,0.50) -- (1.63,0.45) -- (1.57,0.40) -- (1.50,0.35) -- (1.44,0.30) -- (1.37,0.25) -- (1.30,0.20) -- (1.22,0.15) -- (1.15,0.10) -- (1.07,0.05) -- (1.00,0.00);
\draw (0.50,2.00) node{$L$};
\draw (0.50,0.00) node{$R$};
\end{tikzpicture}
\end{mymath}
such that the two morphisms
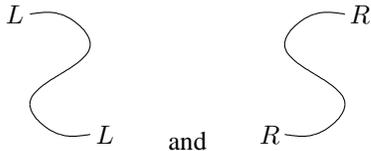
\begin{mymath}
\begin{tikzpicture}[xscale=0.40,yscale=0.40]
\useasboundingbox (-0.5,-0.5) rectangle (4.5,4.5);
\draw[] (1.00,4.00) -- (1.05,4.01) -- (1.10,4.02) -- (1.15,4.02) -- (1.20,4.03) -- (1.25,4.04) -- (1.30,4.04) -- (1.35,4.05) -- (1.40,4.05) -- (1.45,4.06) -- (1.49,4.06) -- (1.54,4.06) -- (1.59,4.06) -- (1.64,4.06) -- (1.70,4.06) -- (1.75,4.05) -- (1.80,4.05) -- (1.85,4.04) -- (1.90,4.03) -- (1.95,4.01) -- (2.00,4.00);
\draw[] (2.00,4.00) -- (2.07,3.98) -- (2.15,3.95) -- (2.22,3.91) -- (2.29,3.88) -- (2.36,3.84) -- (2.43,3.79) -- (2.50,3.74) -- (2.56,3.69) -- (2.62,3.64) -- (2.68,3.59) -- (2.73,3.53) -- (2.79,3.48) -- (2.83,3.42) -- (2.87,3.36) -- (2.91,3.30) -- (2.94,3.24) -- (2.97,3.18) -- (2.98,3.12) -- (3.00,3.06) -- (3.00,3.00);
\draw[] (3.00,3.00) -- (3.00,2.94) -- (2.99,2.89) -- (2.97,2.83) -- (2.95,2.78) -- (2.92,2.72) -- (2.88,2.67) -- (2.84,2.62) -- (2.79,2.57) -- (2.74,2.52) -- (2.69,2.47) -- (2.63,2.42) -- (2.57,2.37) -- (2.51,2.33) -- (2.44,2.28) -- (2.37,2.23) -- (2.30,2.18) -- (2.22,2.14) -- (2.15,2.09) -- (2.08,2.05) -- (2.00,2.00);
\draw[] (2.00,2.00) -- (1.92,1.95) -- (1.85,1.91) -- (1.78,1.86) -- (1.70,1.82) -- (1.63,1.77) -- (1.56,1.72) -- (1.49,1.67) -- (1.43,1.63) -- (1.37,1.58) -- (1.31,1.53) -- (1.26,1.48) -- (1.21,1.43) -- (1.16,1.38) -- (1.12,1.33) -- (1.08,1.28) -- (1.05,1.22) -- (1.03,1.17) -- (1.01,1.11) -- (1.00,1.06) -- (1.00,1.00);
\draw[] (1.00,1.00) -- (1.00,0.94) -- (1.02,0.88) -- (1.03,0.82) -- (1.06,0.76) -- (1.09,0.70) -- (1.13,0.64) -- (1.17,0.58) -- (1.21,0.52) -- (1.27,0.47) -- (1.32,0.41) -- (1.38,0.36) -- (1.44,0.31) -- (1.50,0.26) -- (1.57,0.21) -- (1.64,0.16) -- (1.71,0.12) -- (1.78,0.09) -- (1.85,0.05) -- (1.93,0.02) -- (2.00,0.00);
\draw[] (2.00,0.00) -- (2.05,-0.01) -- (2.10,-0.03) -- (2.15,-0.04) -- (2.20,-0.05) -- (2.25,-0.05) -- (2.30,-0.06) -- (2.36,-0.06) -- (2.41,-0.06) -- (2.46,-0.06) -- (2.51,-0.06) -- (2.55,-0.06) -- (2.60,-0.05) -- (2.65,-0.05) -- (2.70,-0.04) -- (2.75,-0.04) -- (2.80,-0.03) -- (2.85,-0.02) -- (2.90,-0.02) -- (2.95,-0.01) -- (3.00,0.00);
\draw (0.50,4.00) node{$L$};
\draw (3.50,0.00) node{$L$};
\end{tikzpicture}
\qtand
\begin{tikzpicture}[xscale=0.40,yscale=0.40]
\useasboundingbox (-0.5,-0.5) rectangle (4.5,4.5);
\draw[] (1.00,0.00) -- (1.05,-0.01) -- (1.10,-0.02) -- (1.15,-0.02) -- (1.20,-0.03) -- (1.25,-0.04) -- (1.30,-0.04) -- (1.35,-0.05) -- (1.40,-0.05) -- (1.45,-0.06) -- (1.49,-0.06) -- (1.54,-0.06) -- (1.59,-0.06) -- (1.64,-0.06) -- (1.70,-0.06) -- (1.75,-0.05) -- (1.80,-0.05) -- (1.85,-0.04) -- (1.90,-0.03) -- (1.95,-0.01) -- (2.00,0.00);
\draw[] (2.00,0.00) -- (2.07,0.02) -- (2.15,0.05) -- (2.22,0.09) -- (2.29,0.12) -- (2.36,0.16) -- (2.43,0.21) -- (2.50,0.26) -- (2.56,0.31) -- (2.62,0.36) -- (2.68,0.41) -- (2.73,0.47) -- (2.79,0.52) -- (2.83,0.58) -- (2.87,0.64) -- (2.91,0.70) -- (2.94,0.76) -- (2.97,0.82) -- (2.98,0.88) -- (3.00,0.94) -- (3.00,1.00);
\draw[] (3.00,1.00) -- (3.00,1.06) -- (2.99,1.11) -- (2.97,1.17) -- (2.95,1.22) -- (2.92,1.28) -- (2.88,1.33) -- (2.84,1.38) -- (2.79,1.43) -- (2.74,1.48) -- (2.69,1.53) -- (2.63,1.58) -- (2.57,1.63) -- (2.51,1.67) -- (2.44,1.72) -- (2.37,1.77) -- (2.30,1.82) -- (2.22,1.86) -- (2.15,1.91) -- (2.08,1.95) -- (2.00,2.00);
\draw[] (2.00,2.00) -- (1.92,2.05) -- (1.85,2.09) -- (1.78,2.14) -- (1.70,2.18) -- (1.63,2.23) -- (1.56,2.28) -- (1.49,2.33) -- (1.43,2.37) -- (1.37,2.42) -- (1.31,2.47) -- (1.26,2.52) -- (1.21,2.57) -- (1.16,2.62) -- (1.12,2.67) -- (1.08,2.72) -- (1.05,2.78) -- (1.03,2.83) -- (1.01,2.89) -- (1.00,2.94) -- (1.00,3.00);
\draw[] (1.00,3.00) -- (1.00,3.06) -- (1.02,3.12) -- (1.03,3.18) -- (1.06,3.24) -- (1.09,3.30) -- (1.13,3.36) -- (1.17,3.42) -- (1.21,3.48) -- (1.27,3.53) -- (1.32,3.59) -- (1.38,3.64) -- (1.44,3.69) -- (1.50,3.74) -- (1.57,3.79) -- (1.64,3.84) -- (1.71,3.88) -- (1.78,3.91) -- (1.85,3.95) -- (1.93,3.98) -- (2.00,4.00);
\draw[] (2.00,4.00) -- (2.05,4.01) -- (2.10,4.03) -- (2.15,4.04) -- (2.20,4.05) -- (2.25,4.05) -- (2.30,4.06) -- (2.36,4.06) -- (2.41,4.06) -- (2.46,4.06) -- (2.51,4.06) -- (2.55,4.06) -- (2.60,4.05) -- (2.65,4.05) -- (2.70,4.04) -- (2.75,4.04) -- (2.80,4.03) -- (2.85,4.02) -- (2.90,4.02) -- (2.95,4.01) -- (3.00,4.00);
\draw (3.50,4.00) node{$R$};
\draw (0.50,0.00) node{$R$};
\end{tikzpicture}
\end{mymath}
are equal to the identities on $L$ and $R$ respectively. We write~$\eqth{D}$ for
the equational theory associated to dual objects.


\section{Presenting the category of relations}
\label{section:presentation-rel}
We now introduce a presentation of the category $\Rel$ of finite ordinals and
relations, by refining presentations of simpler categories. This result is
mentioned in Examples~6 and~7 of~\cite{hyland-power:symmetric-monoidal-sketches}
and is proved in three different ways
in~\cite{lafont:equational-reasoning-diagrams}, \cite{pirashvili:bialg-prop}
and~\cite{lack:composing-props}. The methodology adopted here to build this
presentation has the advantage of being simple to check (although very
repetitive) and can be extended to give the presentation of the category of
games and strategies described in Section~\ref{subsection:walking-inno}. For the
lack of space, most of the proofs have been omitted or only sketched; detailed
proofs can be found in the author's PhD thesis~\cite{mimram:phd}.

\paragraph{The simplicial category.}
The simplicial category $\Delta$ is the monoidal theory whose morphisms
$f:\intset{m}\to\intset{n}$ are the monotone functions from $\intset{m}$ to
$\intset{n}$. It has been known for a long time that this category is closely
related to the notion of monoid, see~\cite{maclane:cwm}
or~\cite{lafont:boolean-circuits} for example. This result can be formulated as
follows:
\begin{property}
  \label{property:delta-presentation}
  The monoidal category $\Delta$ is presented by the equational theory of
  monoids $\eqth{M}$.
\end{property}
\noindent
In this sense, the simplicial category $\Delta$ impersonates the notion of
monoid. We extend here this result to more complex categories.




\paragraph{Multirelations.}
A \emph{multirelation} $R$ between two finite sets~$A$ and~$B$ is a function
$R:A\times B\to\N$. It can be equivalently be seen as a multiset whose elements
are in $A\times B$ or as a matrix over $\N$ (or as a span in the category of
finite sets).
If \hbox{$R_1:A\to B$} and \hbox{$R_2:B\to C$} are two multirelations, their
composition is defined by
\begin{mymath}
R_2\circ R_1(a,c)
\qeq
\sum_{b\in B}R_1(a,b)\times R_2(b,c)
\tdot
\end{mymath}
(this corresponds to the usual composition of matrices if we see $R_1$ and $R_2$
as matrices over $\N$).
The cardinal $\card{R}$ of a multirelation $R:A\to B$ is the sum of its
coefficients.
We write $\FMR$ for the monoidal theory of multirelations: its objects are
finite ordinals and morphisms are multirelations between them. It is a strict
symmetric monoidal category with the tensor product $\otimes$ defined on objects
and morphisms by disjoint union, and thus a monoidal theory.
In this category, the object~$\intset{1}$ can be equipped with the obvious
bicommutative bialgebra structure $(1,R^\mu,R^\eta,R^\delta,R^\varepsilon)$. For
example, $R^\mu:\intset{2}\to\intset{1}$ is the multirelation defined by
$R^\mu(i,0)=1$ for $i=0$ or $i=1$. We now show that the category of
multirelations is presented by the equational theory~$\eqth{B}$ of bicommutative
bialgebras. We write $\catquot{\mathcal{B}}$ for the monoidal category generated
by $\eqth{B}$.


For every morphism $\phi:m\to n$ in $\mathcal{B}$, where $m>0$, we define a
morphism $S\phi:m+1\to n$ by
\begin{myequation}
  \label{eq:ctx-S}
  S\phi\qeq \phi\circ(\gamma\otimes \id_{m-1})
\end{myequation}
The \emph{stairs} morphisms are defined inductively as either $\id_1$ or
$S\phi'$ where $\phi'$ is a stair, and are represented graphically as
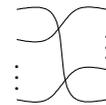
\begin{mymath}
\sstrid{gsym}
\end{mymath}
The \emph{length} of a stairs is defined as $0$ if it is an identity, or as the
length of the stairs $\phi'$ plus one if it is of the form $S\phi'$.

\pagebreak
Morphisms~$\phi$ which are \emph{precanonical forms} are defined inductively:
$\phi$ is either empty or
\vspace\reduce
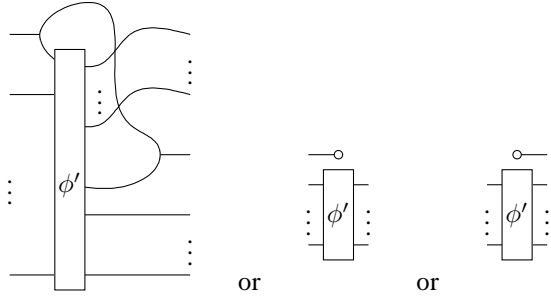
\[
\begin{array}{ccccc}
  \sstrid{bialg_nf_mu}
  &
  \tor
  &
  \sstrid{bialg_nf_eps}
  &
  \tor
  &
  \sstrid{bialg_nf_eta}
\end{array}
\vspace\reduce
\]
where $\phi'$ is a precanonical form. In this case, we write respectively $\phi$
as~$Z$ (the identity morphism~$\id_{\intset{0}}$), as $W_i\phi'$ (where $i$ is
the length of the stairs in the morphism), as~$E\phi'$ or
as~$H\phi'$. Precanonical forms~$\phi$ are thus the well formed morphisms (where
compositions respect types) generated by the following grammar:
\begin{myequation}
  \label{eq:precan-mrel-gram}
  \phi\qgramdef Z\gramor W_i\phi\gramor E\phi\gramor H\phi
\end{myequation}
\vspace{-2ex}




It is easy to see that every non-identity morphism~$\phi$ of a category
generated by a monoidal equational theory (such as~$\eqth{B}$) can be written as
$\phi=\phi'\circ(\intset{m}\otimes\pi\otimes\intset{n})$, where~$\pi$ is a
generator, thus allowing us to reason inductively about morphisms, by case
analysis on the integer $\intset{m}$ and on the generator~$\pi$. Using this
technique, we can prove that

\begin{lemma}
  \label{lemma:mrel-precan}
  Every morphism~$\phi$ of~$\mathcal{B}$ is equivalent (\wrt{} the relation
  $\equiv$) to a precanonical form.
\end{lemma}

The \emph{canonical forms} are precanonical forms which are normal \wrt{} the
following rewriting system:
\begin{myequation}
  \label{eq:mrel-cf-rs}
  \begin{array}{r@{\quad\Longrightarrow\quad}l}
    HW_i&W_{i+1}H\\
    HE&EH\\
    W_iW_j&W_jW_i\qquad\text{when $i<j$}
  \end{array}
\end{myequation}
when considered as words generated by the
grammar~\eqref{eq:precan-mrel-gram}. This rewriting system can easily be shown
to be terminating and confluent, and moreover two morphisms~$\phi$ and~$\psi$
such that $\phi\Longrightarrow\psi$ can be shown to be equivalent. By
Lemma~\ref{lemma:mrel-precan}, every morphism of~$\mathcal{B}$ is therefore
equivalent to an unique canonical form.

\begin{lemma}
  \label{lemma:mrel-bij}
  Every multirelation $R:m\to n$ is represented by an unique canonical form.
\end{lemma}
\begin{proof}
  We prove by induction on~$m$ and on the cardinal $\card{R}$ of $R$ that~$R$ is
  represented by a precanonical form.
  \begin{enumerate}
  \item If $m=0$ then $R$ is represented by the canonical form $H\ldots HZ$
    (with~$n$ occurrences of $H$).
  \item If $m>0$ and for every $j<n$, $R(0,j)=0$ then $R$ is of the form
    $R=R^\varepsilon\otimes R'$ and $R$ is necessarily represented by a
    precanonical form $E\phi'$ where $\phi'$ is a precanonical form representing
    $R':(m-1)\to n$, obtained by induction hypothesis.
  \item Otherwise, $R$ is necessarily represented by a precanonical form of the
    form $W_k\phi'$, where $k$ is the greatest index such that $R(0,k)>0$ and
    $\phi'$ is a precanonical form, obtained by induction, representing the
    relation $R':m\to n$ defined by
    \begin{mymath}
    R'(i,j)=
    \begin{cases}
      R(i,j)-1&\text{if $i=0$ and $j=k$,}\\
      R(i,j)&\text{otherwise.}
    \end{cases}
    \end{mymath}
  \end{enumerate}
  \vspace{-3ex} It can be moreover shown that every precanonical form
  representing~$R$ corresponds to such an enumeration of the coefficients
  of~$R$, that the precanonical form constructed by the proof above is
  canonical, and that it is the only way to obtain a precanonical form
  representing~$R$.
\end{proof}


Finally, we can deduce that
\begin{theorem}
  The category $\FMR$ of multirelations is presented by the equational theory
  $\eqth{B}$ of bicommutative bialgebras.
\end{theorem}
\vspace{-2ex}

\paragraph{Relations.}
The monoidal category $\Rel$ has finite ordinals as objects and relations as
morphisms. This category can be obtained from $\FMR$ by quotienting the
morphisms by the equivalence relation $\sim$ on multirelations such that two
multirelations $R_1,R_2:m\to n$ are equivalent when they have the same null
coefficients. We can therefore easily adapt the previous presentation to show
that

\begin{theorem}
  The category $\Rel$ of relations is presented by the equational theory
  $\eqth{R}$ of \emph{qualitative} bicommutative bialgebras.
\end{theorem}

\noindent
In particular, precanonical forms are the same and canonical forms are defined
by adding the rule $W_iW_i\Longrightarrow W_i$ to the rewriting
system~\eqref{eq:mrel-cf-rs}.

\section{A game semantics for first-order causality}
\label{section:games-strategies}

Suppose that we are given a fixed first-order language~$\mathcal{L}$, that is a
set of proposition symbols~$P,Q,\ldots$ with given arities, a set of function
symbols~$f,g,\ldots$ with given arities and a set of first-order
variables~$x,y,\ldots$. \emph{Terms}~$t$ and \emph{formulas}~$A$ are
respectively generated by the following grammars:
\begin{mymath}
\begin{array}{rcl}
  t&\qgramdef&x\gramor f(t,\ldots,t)
  \\
  A&\qgramdef&P(t,\ldots,t)\gramor\qforall{x}{A}\gramor\qexists{x}{A}
\end{array}
\end{mymath}
(we only consider formulas without connectives here). We suppose that
application of propositions and functions always respect arities.
Formulas are considered modulo renaming of bound variables and substitution
$A[t/x]$ of a free variable $x$ by a term $t$ in a formula $A$ is defined as
usual, avoiding capture of variables. In the following, we sometimes omit the
arguments of propositions when they are clear from the context. We also suppose
given a set~$\axioms$ of \emph{axioms}, that is pairs of propositions, which is
reflexive, transitive and closed under substitution. The logic associated to
these formulas has the following inference rules:
\[
\begin{array}{c@{\qquad}c}
  \inferrule{A[t/x]\vdash B}{\qforall x A\vdash B}{\lrule{$\forall$-L}}
  &
  \inferrule{A\vdash B}{A\vdash \qforall x B}{\lrule{$\forall$-R}}
  \displaybreak[0]
  \\
  &
  \text{(with $x$ not free in $A$)}
  \displaybreak[0]
  \\[2ex]
  \inferrule{A\vdash B}{\qexists x A\vdash B}{\lrule{$\exists$-L}}
  &
  \inferrule{A\vdash B[t/x]}{A\vdash \qexists x B}{\lrule{$\exists$-R}}
  \displaybreak[0]
  \\
  \text{(with $x$ not free in $B$)}
  &
  \displaybreak[0]
  \\[2ex]
  \inferrule{(P,Q)\in\axioms}{P\vdash Q}{\lrule{Ax}}
  &
  \inferrule{A\vdash B\\B\vdash C}{A\vdash C}{\lrule{Cut}}
\end{array}
\]

\paragraph{Games and strategies.}
\label{subsection:games-strategies}
Games are defined as follows.

\begin{definition}
  A \emph{game} $A=(\moves{A},\lambda_A,\leq_A)$ consists of a set of
  moves~$\moves{A}$, a polarization function $\lambda_A$ from~$\moves{A}$
  to~$\{-1,+1\}$ which to every move $m$ associates its \emph{polarity}, and a
  well-founded partial order $\leq_A$ on moves, called \emph{causality} or
  \emph{justification}.
  A move $m$ is said to be a Proponent move when $\lambda_A(m)=+1$ and an
  Opponent move otherwise.
\end{definition}

If $A$ and $B$ are two games, their tensor product $A\otimes B$ is defined by
disjoint union on moves, polarities and causality,
the opposite game $A^*$ of the game $A$ is obtained from~$A$ by inverting
polarities
and the arrow game $A\llimp B$ is defined by $A\llimp B=A^*\otimes B$.
A game $A$ is \emph{filiform} when the associated partial order is total (we are
mostly interested in such games in the following).


\begin{definition}
  \label{def:strategy}
  A \emph{strategy} $\sigma$ on a game $A$ is a partial order $\leq_\sigma$ on
  the moves of $A$ which satisfies the two following properties:
  \begin{enumerate}
  \item \emph{polarity}: for every pair of moves $m,n\in\moves{A}$ such that
    \hbox{$m<_\sigma n$}, we have $\lambda_A(m)=-1\tand\lambda_A(n)=+1$.
  \item \emph{acyclicity}: the partial order $\leq_\sigma$ is compatible with
    the partial order of the game, in the sense that the transitive closure of
    their union is still a partial order (\ie is acyclic).
  \end{enumerate}
\end{definition}
The \emph{size} $\size{A}$ of a game $A$ is the cardinal of $\moves{A}$ and the
\emph{size} $\size{\sigma}$ of a strategy $\sigma:A$ is the cardinal of the
relation~$\leq_\sigma$. If $\sigma:A\llimp B$ and $\tau:B\llimp C$ are two
strategies, their composite $\tau\circ\sigma:A\llimp C$ is the partial order
$\leq_{\tau\circ\sigma}$ on the moves of $A\llimp C$, defined as the restriction
of the transitive closure of the union~$\leq_\sigma\cup\leq_\tau$ of the partial
orders~$\leq_\sigma$ and~$\leq_\tau$ (considered as relations). The identity
strategy $\id_A:A\llimp A$ on a game $A$ is the strategy such that for every
move $m$ of $A$ we have $m_L\leq_{\id_A}m_R$ if $\lambda_A(m)=-1$ and
$m_R\leq_{\id_A}m_L$ if \hbox{$\lambda_A(m)=+1$}, where $m_L$ (\resp $m_R$) is
the instance of a move $m$ in the left-hand side (\resp right-hand side) copy
of~$A$.

Since composition of strategies is defined in the category of relations, we
still have to check that the composite of two strategies $\sigma$ and $\tau$ is
actually a strategy. The preservation by composition of the polarity condition
is immediate. However, proving that the relation $\leq_{\tau\circ\sigma}$
corresponding to the composite strategy is acyclic is more difficult: a direct
proof of this property is combinatorial, lengthy and requires global reasoning
about strategies. For now, we define the category~$\Games$ as the category whose
objects are finite filiform games, and whose sets of morphisms are the smallest
sets containing the strategies on the game $A\llimp B$ as morphisms between two
objects $A$ and $B$ and are moreover closed under composition. We will deduce at
the end of the section, from its presentation, that strategies are in fact the
only morphisms of this category.

If $A$ and $B$ are two games, the game $A\before{}B$ (to be read $A$
\emph{before} $B$) is the game defined as $A\lltens B$ on moves and polarities
and $\leq_{A\before{}B}$ is the transitive closure of the relation
\begin{mymath}
  \leq_{A\lltens B}\cup\;\setof{(a,b)\tq a\in M_A\tand b\in M_B}
\end{mymath}
This operation is extended as a bifunctor on strategies as follows. If
$\sigma:A\to B$ and $\tau:C\to D$ are two strategies, the strategy
$\sigma\before{}\tau:A\before{}C\to B\before{}D$ is defined as the relation
$\leq_{\sigma\before{}\tau}=\leq_\sigma\uplus\leq_\tau$.
This bifunctor induces a monoidal structure $(\Games,\before{},I)$ on the
category $\Games$, where $I$ denotes the empty game.

We write $O$ for a game with only one Opponent move and $P$ for a game with only
one Proponent move. It can be easily remarked that finite filiform games $A$ are
generated by the following grammar
\begin{mymath}
A\qqgramdef I\gramor O\before{}A\gramor P\before{}A
\end{mymath}
A strategy $\sigma:A\to B$ is represented graphically by drawing a line from a
move $m$ to a move $n$ whenever $m\leq_\sigma n$. For example, the strategy
$\mu^P:P\before{}P\to P$
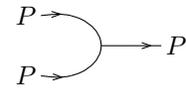
\begin{mymath}
\begin{tikzpicture}[xscale=0.40,yscale=0.40]
\useasboundingbox (-0.5,-0.5) rectangle (6.5,2.5);
\draw[] (1.00,2.00) -- (1.05,2.01) -- (1.10,2.01) -- (1.15,2.02) -- (1.20,2.02) -- (1.25,2.03) -- (1.30,2.03) -- (1.35,2.04) -- (1.40,2.04) -- (1.45,2.04) -- (1.50,2.05) -- (1.55,2.05) -- (1.60,2.05) -- (1.65,2.05) -- (1.70,2.04) -- (1.75,2.04) -- (1.80,2.03) -- (1.85,2.03) -- (1.90,2.02) -- (1.95,2.01) -- (2.00,2.00);
\draw[] (2.00,2.00) -- (2.06,1.98) -- (2.13,1.96) -- (2.19,1.94) -- (2.25,1.92) -- (2.32,1.89) -- (2.38,1.86) -- (2.44,1.83) -- (2.50,1.79) -- (2.55,1.75) -- (2.61,1.71) -- (2.66,1.67) -- (2.71,1.63) -- (2.75,1.58) -- (2.80,1.53) -- (2.84,1.48) -- (2.87,1.43) -- (2.90,1.37) -- (2.93,1.32) -- (2.96,1.26) -- (2.97,1.20);
\draw[] (2.97,1.20) -- (2.98,1.19) -- (2.98,1.18) -- (2.98,1.17) -- (2.98,1.16) -- (2.99,1.15) -- (2.99,1.14) -- (2.99,1.13) -- (2.99,1.12) -- (2.99,1.11) -- (2.99,1.10) -- (2.99,1.09) -- (3.00,1.08) -- (3.00,1.07) -- (3.00,1.06) -- (3.00,1.05) -- (3.00,1.04) -- (3.00,1.03) -- (3.00,1.02) -- (3.00,1.01) -- (3.00,1.00);
\draw[] (3.00,1.00) -- (3.00,0.99) -- (3.00,0.98) -- (3.00,0.97) -- (3.00,0.96) -- (3.00,0.95) -- (3.00,0.94) -- (3.00,0.93) -- (3.00,0.92) -- (2.99,0.91) -- (2.99,0.90) -- (2.99,0.89) -- (2.99,0.88) -- (2.99,0.87) -- (2.99,0.86) -- (2.99,0.85) -- (2.98,0.84) -- (2.98,0.83) -- (2.98,0.82) -- (2.98,0.81) -- (2.97,0.80);
\draw[] (2.97,0.80) -- (2.96,0.74) -- (2.93,0.68) -- (2.90,0.63) -- (2.87,0.57) -- (2.84,0.52) -- (2.80,0.47) -- (2.75,0.42) -- (2.71,0.37) -- (2.66,0.33) -- (2.61,0.29) -- (2.55,0.25) -- (2.50,0.21) -- (2.44,0.17) -- (2.38,0.14) -- (2.32,0.11) -- (2.25,0.08) -- (2.19,0.06) -- (2.13,0.04) -- (2.06,0.02) -- (2.00,0.00);
\draw[] (2.00,0.00) -- (1.95,-0.01) -- (1.90,-0.02) -- (1.85,-0.03) -- (1.80,-0.03) -- (1.75,-0.04) -- (1.70,-0.04) -- (1.65,-0.05) -- (1.60,-0.05) -- (1.55,-0.05) -- (1.50,-0.05) -- (1.45,-0.04) -- (1.40,-0.04) -- (1.35,-0.04) -- (1.30,-0.03) -- (1.25,-0.03) -- (1.20,-0.02) -- (1.15,-0.02) -- (1.10,-0.01) -- (1.05,-0.01) -- (1.00,0.00);
\draw (1.40,0.06) -- (1.70,-0.05);
\draw (1.39,-0.14) -- (1.70,-0.05);
\draw (1.35,1.94) -- (1.65,2.05);
\draw (1.35,2.14) -- (1.65,2.05);
\draw[] (3.00,1.00) -- (3.05,1.00) -- (3.10,1.00) -- (3.15,1.00) -- (3.20,1.00) -- (3.25,1.00) -- (3.30,1.00) -- (3.35,1.00) -- (3.40,1.00) -- (3.45,1.00) -- (3.50,1.00) -- (3.55,1.00) -- (3.60,1.00) -- (3.65,1.00) -- (3.70,1.00) -- (3.75,1.00) -- (3.80,1.00) -- (3.85,1.00) -- (3.90,1.00) -- (3.95,1.00) -- (4.00,1.00);
\draw[] (4.00,1.00) -- (4.05,1.00) -- (4.10,1.00) -- (4.15,1.00) -- (4.20,1.00) -- (4.25,1.00) -- (4.30,1.00) -- (4.35,1.00) -- (4.40,1.00) -- (4.45,1.00) -- (4.50,1.00) -- (4.55,1.00) -- (4.60,1.00) -- (4.65,1.00) -- (4.70,1.00) -- (4.75,1.00) -- (4.80,1.00) -- (4.85,1.00) -- (4.90,1.00) -- (4.95,1.00) -- (5.00,1.00);
\draw (4.35,1.10) -- (4.65,1.00);
\draw (4.35,0.90) -- (4.65,1.00);
\draw (0.50,2.00) node{$P$};
\draw (5.50,1.00) node{$P$};
\draw (0.50,0.00) node{$P$};
\end{tikzpicture}
\end{mymath}
is the strategy on the game $(O\before{}O)\otimes P$ in which both Opponent move
of the left-hand game justify the Proponent move of the right-hand game. When a
move does not justify (or is not justified by) any other move, we draw a line
ended by a small circle. For example, the strategy \hbox{$\varepsilon^P:P\to
  I$}, drawn as
\begin{mymath}
\begin{tikzpicture}[xscale=0.40,yscale=0.40]
\useasboundingbox (-0.5,-0.5) rectangle (6.5,0.5);
\draw[,] (3.00,0.00) -- (1.00,0.00);
\draw (1.85,0.10) -- (2.15,0.00);
\draw (1.85,-0.10) -- (2.15,0.00);
\filldraw[fill=white] (3.00,0.00) ellipse (0.14cm and 0.14cm);
\draw (0.50,0.00) node{$P$};
\end{tikzpicture}

\end{mymath}
is the unique strategy from $P$ to the terminal object $I$. With these
conventions, we introduce notations for some morphisms which are depicted in
Figure~\ref{fig:inno-gen}.

\begin{figure}[!t]
  \vspace{-0.7ex}
  \fbox{
    \vbox{
      \begin{mymath}
        \hspace{-1.5ex}
        \begin{array}{c}
          \begin{array}{r@{\qcolon}l@{\quad}r@{\qcolon}l}
            \mu^O&O\before{} O\to O&\mu^P&P\before{} P\to P\\
            \eta^O&I\to O&\eta^P&I\to P\\
            \delta^O&O\to O\before{} O&\delta^P&P\to P\before{} P\\
            \varepsilon^O&O\to I&\varepsilon^P&P\to I\\
            \gamma^O&O\before{} O\to O\before{} O&\gamma^P&P\before{} P\to P\before{} P\\
            \eta^{OP}&I\to O\before{} P&\varepsilon^{OP}&P\before{} O\to I\\
          \end{array}
          \\
          \begin{array}{rcl}
            \gamma^{OP}&\colon&P\before{} O\to O\before{} P
          \end{array}
        \end{array}
      \end{mymath}
      respectively drawn as
      \begin{mymath}
        \begin{array}{c}
          \begin{array}{c@{\qquad}c}
            \begin{tikzpicture}[xscale=0.40,yscale=0.40]
\useasboundingbox (-0.5,-0.5) rectangle (6.5,2.5);
\draw[] (1.00,2.00) -- (1.05,2.01) -- (1.10,2.01) -- (1.15,2.02) -- (1.20,2.02) -- (1.25,2.03) -- (1.30,2.03) -- (1.35,2.04) -- (1.40,2.04) -- (1.45,2.04) -- (1.50,2.05) -- (1.55,2.05) -- (1.60,2.05) -- (1.65,2.05) -- (1.70,2.04) -- (1.75,2.04) -- (1.80,2.03) -- (1.85,2.03) -- (1.90,2.02) -- (1.95,2.01) -- (2.00,2.00);
\draw[] (2.00,2.00) -- (2.06,1.98) -- (2.13,1.96) -- (2.19,1.94) -- (2.25,1.92) -- (2.32,1.89) -- (2.38,1.86) -- (2.44,1.83) -- (2.50,1.79) -- (2.55,1.75) -- (2.61,1.71) -- (2.66,1.67) -- (2.71,1.63) -- (2.75,1.58) -- (2.80,1.53) -- (2.84,1.48) -- (2.87,1.43) -- (2.90,1.37) -- (2.93,1.32) -- (2.96,1.26) -- (2.97,1.20);
\draw[] (2.97,1.20) -- (2.98,1.19) -- (2.98,1.18) -- (2.98,1.17) -- (2.98,1.16) -- (2.99,1.15) -- (2.99,1.14) -- (2.99,1.13) -- (2.99,1.12) -- (2.99,1.11) -- (2.99,1.10) -- (2.99,1.09) -- (3.00,1.08) -- (3.00,1.07) -- (3.00,1.06) -- (3.00,1.05) -- (3.00,1.04) -- (3.00,1.03) -- (3.00,1.02) -- (3.00,1.01) -- (3.00,1.00);
\draw[] (3.00,1.00) -- (3.00,0.99) -- (3.00,0.98) -- (3.00,0.97) -- (3.00,0.96) -- (3.00,0.95) -- (3.00,0.94) -- (3.00,0.93) -- (3.00,0.92) -- (2.99,0.91) -- (2.99,0.90) -- (2.99,0.89) -- (2.99,0.88) -- (2.99,0.87) -- (2.99,0.86) -- (2.99,0.85) -- (2.98,0.84) -- (2.98,0.83) -- (2.98,0.82) -- (2.98,0.81) -- (2.97,0.80);
\draw[] (2.97,0.80) -- (2.96,0.74) -- (2.93,0.68) -- (2.90,0.63) -- (2.87,0.57) -- (2.84,0.52) -- (2.80,0.47) -- (2.75,0.42) -- (2.71,0.37) -- (2.66,0.33) -- (2.61,0.29) -- (2.55,0.25) -- (2.50,0.21) -- (2.44,0.17) -- (2.38,0.14) -- (2.32,0.11) -- (2.25,0.08) -- (2.19,0.06) -- (2.13,0.04) -- (2.06,0.02) -- (2.00,0.00);
\draw[] (2.00,0.00) -- (1.95,-0.01) -- (1.90,-0.02) -- (1.85,-0.03) -- (1.80,-0.03) -- (1.75,-0.04) -- (1.70,-0.04) -- (1.65,-0.05) -- (1.60,-0.05) -- (1.55,-0.05) -- (1.50,-0.05) -- (1.45,-0.04) -- (1.40,-0.04) -- (1.35,-0.04) -- (1.30,-0.03) -- (1.25,-0.03) -- (1.20,-0.02) -- (1.15,-0.02) -- (1.10,-0.01) -- (1.05,-0.01) -- (1.00,0.00);
\draw (1.70,0.05) -- (1.40,-0.04);
\draw (1.69,-0.15) -- (1.40,-0.04);
\draw (1.65,1.95) -- (1.35,2.04);
\draw (1.65,2.15) -- (1.35,2.04);
\draw[] (3.00,1.00) -- (3.05,1.00) -- (3.10,1.00) -- (3.15,1.00) -- (3.20,1.00) -- (3.25,1.00) -- (3.30,1.00) -- (3.35,1.00) -- (3.40,1.00) -- (3.45,1.00) -- (3.50,1.00) -- (3.55,1.00) -- (3.60,1.00) -- (3.65,1.00) -- (3.70,1.00) -- (3.75,1.00) -- (3.80,1.00) -- (3.85,1.00) -- (3.90,1.00) -- (3.95,1.00) -- (4.00,1.00);
\draw[] (4.00,1.00) -- (4.05,1.00) -- (4.10,1.00) -- (4.15,1.00) -- (4.20,1.00) -- (4.25,1.00) -- (4.30,1.00) -- (4.35,1.00) -- (4.40,1.00) -- (4.45,1.00) -- (4.50,1.00) -- (4.55,1.00) -- (4.60,1.00) -- (4.65,1.00) -- (4.70,1.00) -- (4.75,1.00) -- (4.80,1.00) -- (4.85,1.00) -- (4.90,1.00) -- (4.95,1.00) -- (5.00,1.00);
\draw (4.65,1.10) -- (4.35,1.00);
\draw (4.65,0.90) -- (4.35,1.00);
\draw (0.50,2.00) node{$O$};
\draw (5.50,1.00) node{$O$};
\draw (0.50,0.00) node{$O$};
\end{tikzpicture}&\begin{tikzpicture}[xscale=0.40,yscale=0.40]
\useasboundingbox (-0.5,-0.5) rectangle (6.5,2.5);
\draw[] (1.00,2.00) -- (1.05,2.01) -- (1.10,2.01) -- (1.15,2.02) -- (1.20,2.02) -- (1.25,2.03) -- (1.30,2.03) -- (1.35,2.04) -- (1.40,2.04) -- (1.45,2.04) -- (1.50,2.05) -- (1.55,2.05) -- (1.60,2.05) -- (1.65,2.05) -- (1.70,2.04) -- (1.75,2.04) -- (1.80,2.03) -- (1.85,2.03) -- (1.90,2.02) -- (1.95,2.01) -- (2.00,2.00);
\draw[] (2.00,2.00) -- (2.06,1.98) -- (2.13,1.96) -- (2.19,1.94) -- (2.25,1.92) -- (2.32,1.89) -- (2.38,1.86) -- (2.44,1.83) -- (2.50,1.79) -- (2.55,1.75) -- (2.61,1.71) -- (2.66,1.67) -- (2.71,1.63) -- (2.75,1.58) -- (2.80,1.53) -- (2.84,1.48) -- (2.87,1.43) -- (2.90,1.37) -- (2.93,1.32) -- (2.96,1.26) -- (2.97,1.20);
\draw[] (2.97,1.20) -- (2.98,1.19) -- (2.98,1.18) -- (2.98,1.17) -- (2.98,1.16) -- (2.99,1.15) -- (2.99,1.14) -- (2.99,1.13) -- (2.99,1.12) -- (2.99,1.11) -- (2.99,1.10) -- (2.99,1.09) -- (3.00,1.08) -- (3.00,1.07) -- (3.00,1.06) -- (3.00,1.05) -- (3.00,1.04) -- (3.00,1.03) -- (3.00,1.02) -- (3.00,1.01) -- (3.00,1.00);
\draw[] (3.00,1.00) -- (3.00,0.99) -- (3.00,0.98) -- (3.00,0.97) -- (3.00,0.96) -- (3.00,0.95) -- (3.00,0.94) -- (3.00,0.93) -- (3.00,0.92) -- (2.99,0.91) -- (2.99,0.90) -- (2.99,0.89) -- (2.99,0.88) -- (2.99,0.87) -- (2.99,0.86) -- (2.99,0.85) -- (2.98,0.84) -- (2.98,0.83) -- (2.98,0.82) -- (2.98,0.81) -- (2.97,0.80);
\draw[] (2.97,0.80) -- (2.96,0.74) -- (2.93,0.68) -- (2.90,0.63) -- (2.87,0.57) -- (2.84,0.52) -- (2.80,0.47) -- (2.75,0.42) -- (2.71,0.37) -- (2.66,0.33) -- (2.61,0.29) -- (2.55,0.25) -- (2.50,0.21) -- (2.44,0.17) -- (2.38,0.14) -- (2.32,0.11) -- (2.25,0.08) -- (2.19,0.06) -- (2.13,0.04) -- (2.06,0.02) -- (2.00,0.00);
\draw[] (2.00,0.00) -- (1.95,-0.01) -- (1.90,-0.02) -- (1.85,-0.03) -- (1.80,-0.03) -- (1.75,-0.04) -- (1.70,-0.04) -- (1.65,-0.05) -- (1.60,-0.05) -- (1.55,-0.05) -- (1.50,-0.05) -- (1.45,-0.04) -- (1.40,-0.04) -- (1.35,-0.04) -- (1.30,-0.03) -- (1.25,-0.03) -- (1.20,-0.02) -- (1.15,-0.02) -- (1.10,-0.01) -- (1.05,-0.01) -- (1.00,0.00);
\draw (1.40,0.06) -- (1.70,-0.05);
\draw (1.39,-0.14) -- (1.70,-0.05);
\draw (1.35,1.94) -- (1.65,2.05);
\draw (1.35,2.14) -- (1.65,2.05);
\draw[] (3.00,1.00) -- (3.05,1.00) -- (3.10,1.00) -- (3.15,1.00) -- (3.20,1.00) -- (3.25,1.00) -- (3.30,1.00) -- (3.35,1.00) -- (3.40,1.00) -- (3.45,1.00) -- (3.50,1.00) -- (3.55,1.00) -- (3.60,1.00) -- (3.65,1.00) -- (3.70,1.00) -- (3.75,1.00) -- (3.80,1.00) -- (3.85,1.00) -- (3.90,1.00) -- (3.95,1.00) -- (4.00,1.00);
\draw[] (4.00,1.00) -- (4.05,1.00) -- (4.10,1.00) -- (4.15,1.00) -- (4.20,1.00) -- (4.25,1.00) -- (4.30,1.00) -- (4.35,1.00) -- (4.40,1.00) -- (4.45,1.00) -- (4.50,1.00) -- (4.55,1.00) -- (4.60,1.00) -- (4.65,1.00) -- (4.70,1.00) -- (4.75,1.00) -- (4.80,1.00) -- (4.85,1.00) -- (4.90,1.00) -- (4.95,1.00) -- (5.00,1.00);
\draw (4.35,1.10) -- (4.65,1.00);
\draw (4.35,0.90) -- (4.65,1.00);
\draw (0.50,2.00) node{$P$};
\draw (5.50,1.00) node{$P$};
\draw (0.50,0.00) node{$P$};
\end{tikzpicture}\\
            \begin{tikzpicture}[xscale=0.40,yscale=0.40]
\useasboundingbox (-0.5,-0.5) rectangle (6.5,2.5);
\draw[,] (3.00,1.00) -- (5.00,1.00);
\draw (4.15,1.10) -- (3.85,1.00);
\draw (4.15,0.90) -- (3.85,1.00);
\filldraw[fill=white] (3.00,1.00) ellipse (0.14cm and 0.14cm);
\draw (5.50,1.00) node{$O$};
\end{tikzpicture}&\begin{tikzpicture}[xscale=0.40,yscale=0.40]
\useasboundingbox (-0.5,-0.5) rectangle (6.5,2.5);
\draw[,] (3.00,1.00) -- (5.00,1.00);
\draw (3.85,1.10) -- (4.15,1.00);
\draw (3.85,0.90) -- (4.15,1.00);
\filldraw[fill=white] (3.00,1.00) ellipse (0.14cm and 0.14cm);
\draw (5.50,1.00) node{$P$};
\end{tikzpicture}\\
            \begin{tikzpicture}[xscale=0.40,yscale=0.40]
\useasboundingbox (-0.5,-0.5) rectangle (6.5,2.5);
\draw[] (5.00,0.00) -- (4.95,-0.01) -- (4.90,-0.01) -- (4.85,-0.02) -- (4.80,-0.02) -- (4.75,-0.03) -- (4.70,-0.03) -- (4.65,-0.04) -- (4.60,-0.04) -- (4.55,-0.04) -- (4.50,-0.05) -- (4.45,-0.05) -- (4.40,-0.05) -- (4.35,-0.05) -- (4.30,-0.04) -- (4.25,-0.04) -- (4.20,-0.03) -- (4.15,-0.03) -- (4.10,-0.02) -- (4.05,-0.01) -- (4.00,0.00);
\draw[] (4.00,0.00) -- (3.94,0.02) -- (3.87,0.04) -- (3.81,0.06) -- (3.75,0.08) -- (3.68,0.11) -- (3.62,0.14) -- (3.56,0.17) -- (3.50,0.21) -- (3.45,0.25) -- (3.39,0.29) -- (3.34,0.33) -- (3.29,0.37) -- (3.25,0.42) -- (3.20,0.47) -- (3.16,0.52) -- (3.13,0.57) -- (3.10,0.63) -- (3.07,0.68) -- (3.04,0.74) -- (3.03,0.80);
\draw[] (3.03,0.80) -- (3.02,0.81) -- (3.02,0.82) -- (3.02,0.83) -- (3.02,0.84) -- (3.01,0.85) -- (3.01,0.86) -- (3.01,0.87) -- (3.01,0.88) -- (3.01,0.89) -- (3.01,0.90) -- (3.01,0.91) -- (3.00,0.92) -- (3.00,0.93) -- (3.00,0.94) -- (3.00,0.95) -- (3.00,0.96) -- (3.00,0.97) -- (3.00,0.98) -- (3.00,0.99) -- (3.00,1.00);
\draw[] (3.00,1.00) -- (3.00,1.01) -- (3.00,1.02) -- (3.00,1.03) -- (3.00,1.04) -- (3.00,1.05) -- (3.00,1.06) -- (3.00,1.07) -- (3.00,1.08) -- (3.01,1.09) -- (3.01,1.10) -- (3.01,1.11) -- (3.01,1.12) -- (3.01,1.13) -- (3.01,1.14) -- (3.01,1.15) -- (3.02,1.16) -- (3.02,1.17) -- (3.02,1.18) -- (3.02,1.19) -- (3.03,1.20);
\draw[] (3.03,1.20) -- (3.04,1.26) -- (3.07,1.32) -- (3.10,1.37) -- (3.13,1.43) -- (3.16,1.48) -- (3.20,1.53) -- (3.25,1.58) -- (3.29,1.63) -- (3.34,1.67) -- (3.39,1.71) -- (3.45,1.75) -- (3.50,1.79) -- (3.56,1.83) -- (3.62,1.86) -- (3.68,1.89) -- (3.75,1.92) -- (3.81,1.94) -- (3.87,1.96) -- (3.94,1.98) -- (4.00,2.00);
\draw[] (4.00,2.00) -- (4.05,2.01) -- (4.10,2.02) -- (4.15,2.03) -- (4.20,2.03) -- (4.25,2.04) -- (4.30,2.04) -- (4.35,2.05) -- (4.40,2.05) -- (4.45,2.05) -- (4.50,2.05) -- (4.55,2.04) -- (4.60,2.04) -- (4.65,2.04) -- (4.70,2.03) -- (4.75,2.03) -- (4.80,2.02) -- (4.85,2.02) -- (4.90,2.01) -- (4.95,2.01) -- (5.00,2.00);
\draw (4.61,2.14) -- (4.30,2.05);
\draw (4.60,1.94) -- (4.30,2.05);
\draw (4.65,-0.14) -- (4.35,-0.05);
\draw (4.65,0.06) -- (4.35,-0.05);
\draw[] (3.00,1.00) -- (2.95,1.00) -- (2.90,1.00) -- (2.85,1.00) -- (2.80,1.00) -- (2.75,1.00) -- (2.70,1.00) -- (2.65,1.00) -- (2.60,1.00) -- (2.55,1.00) -- (2.50,1.00) -- (2.45,1.00) -- (2.40,1.00) -- (2.35,1.00) -- (2.30,1.00) -- (2.25,1.00) -- (2.20,1.00) -- (2.15,1.00) -- (2.10,1.00) -- (2.05,1.00) -- (2.00,1.00);
\draw[] (2.00,1.00) -- (1.95,1.00) -- (1.90,1.00) -- (1.85,1.00) -- (1.80,1.00) -- (1.75,1.00) -- (1.70,1.00) -- (1.65,1.00) -- (1.60,1.00) -- (1.55,1.00) -- (1.50,1.00) -- (1.45,1.00) -- (1.40,1.00) -- (1.35,1.00) -- (1.30,1.00) -- (1.25,1.00) -- (1.20,1.00) -- (1.15,1.00) -- (1.10,1.00) -- (1.05,1.00) -- (1.00,1.00);
\draw (1.70,1.10) -- (1.40,1.00);
\draw (1.70,0.90) -- (1.40,1.00);
\draw (5.50,2.00) node{$O$};
\draw (0.50,1.00) node{$O$};
\draw (5.50,0.00) node{$O$};
\end{tikzpicture}&\begin{tikzpicture}[xscale=0.40,yscale=0.40]
\useasboundingbox (-0.5,-0.5) rectangle (6.5,2.5);
\draw[] (5.00,0.00) -- (4.95,-0.01) -- (4.90,-0.01) -- (4.85,-0.02) -- (4.80,-0.02) -- (4.75,-0.03) -- (4.70,-0.03) -- (4.65,-0.04) -- (4.60,-0.04) -- (4.55,-0.04) -- (4.50,-0.05) -- (4.45,-0.05) -- (4.40,-0.05) -- (4.35,-0.05) -- (4.30,-0.04) -- (4.25,-0.04) -- (4.20,-0.03) -- (4.15,-0.03) -- (4.10,-0.02) -- (4.05,-0.01) -- (4.00,0.00);
\draw[] (4.00,0.00) -- (3.94,0.02) -- (3.87,0.04) -- (3.81,0.06) -- (3.75,0.08) -- (3.68,0.11) -- (3.62,0.14) -- (3.56,0.17) -- (3.50,0.21) -- (3.45,0.25) -- (3.39,0.29) -- (3.34,0.33) -- (3.29,0.37) -- (3.25,0.42) -- (3.20,0.47) -- (3.16,0.52) -- (3.13,0.57) -- (3.10,0.63) -- (3.07,0.68) -- (3.04,0.74) -- (3.03,0.80);
\draw[] (3.03,0.80) -- (3.02,0.81) -- (3.02,0.82) -- (3.02,0.83) -- (3.02,0.84) -- (3.01,0.85) -- (3.01,0.86) -- (3.01,0.87) -- (3.01,0.88) -- (3.01,0.89) -- (3.01,0.90) -- (3.01,0.91) -- (3.00,0.92) -- (3.00,0.93) -- (3.00,0.94) -- (3.00,0.95) -- (3.00,0.96) -- (3.00,0.97) -- (3.00,0.98) -- (3.00,0.99) -- (3.00,1.00);
\draw[] (3.00,1.00) -- (3.00,1.01) -- (3.00,1.02) -- (3.00,1.03) -- (3.00,1.04) -- (3.00,1.05) -- (3.00,1.06) -- (3.00,1.07) -- (3.00,1.08) -- (3.01,1.09) -- (3.01,1.10) -- (3.01,1.11) -- (3.01,1.12) -- (3.01,1.13) -- (3.01,1.14) -- (3.01,1.15) -- (3.02,1.16) -- (3.02,1.17) -- (3.02,1.18) -- (3.02,1.19) -- (3.03,1.20);
\draw[] (3.03,1.20) -- (3.04,1.26) -- (3.07,1.32) -- (3.10,1.37) -- (3.13,1.43) -- (3.16,1.48) -- (3.20,1.53) -- (3.25,1.58) -- (3.29,1.63) -- (3.34,1.67) -- (3.39,1.71) -- (3.45,1.75) -- (3.50,1.79) -- (3.56,1.83) -- (3.62,1.86) -- (3.68,1.89) -- (3.75,1.92) -- (3.81,1.94) -- (3.87,1.96) -- (3.94,1.98) -- (4.00,2.00);
\draw[] (4.00,2.00) -- (4.05,2.01) -- (4.10,2.02) -- (4.15,2.03) -- (4.20,2.03) -- (4.25,2.04) -- (4.30,2.04) -- (4.35,2.05) -- (4.40,2.05) -- (4.45,2.05) -- (4.50,2.05) -- (4.55,2.04) -- (4.60,2.04) -- (4.65,2.04) -- (4.70,2.03) -- (4.75,2.03) -- (4.80,2.02) -- (4.85,2.02) -- (4.90,2.01) -- (4.95,2.01) -- (5.00,2.00);
\draw (4.31,2.15) -- (4.60,2.04);
\draw (4.30,1.95) -- (4.60,2.04);
\draw (4.35,-0.15) -- (4.65,-0.04);
\draw (4.35,0.05) -- (4.65,-0.04);
\draw[] (1.00,1.00) -- (1.05,1.00) -- (1.10,1.00) -- (1.15,1.00) -- (1.20,1.00) -- (1.25,1.00) -- (1.30,1.00) -- (1.35,1.00) -- (1.40,1.00) -- (1.45,1.00) -- (1.50,1.00) -- (1.55,1.00) -- (1.60,1.00) -- (1.65,1.00) -- (1.70,1.00) -- (1.75,1.00) -- (1.80,1.00) -- (1.85,1.00) -- (1.90,1.00) -- (1.95,1.00) -- (2.00,1.00);
\draw[] (2.00,1.00) -- (2.05,1.00) -- (2.10,1.00) -- (2.15,1.00) -- (2.20,1.00) -- (2.25,1.00) -- (2.30,1.00) -- (2.35,1.00) -- (2.40,1.00) -- (2.45,1.00) -- (2.50,1.00) -- (2.55,1.00) -- (2.60,1.00) -- (2.65,1.00) -- (2.70,1.00) -- (2.75,1.00) -- (2.80,1.00) -- (2.85,1.00) -- (2.90,1.00) -- (2.95,1.00) -- (3.00,1.00);
\draw (1.35,1.10) -- (1.65,1.00);
\draw (1.35,0.90) -- (1.65,1.00);
\draw (5.50,2.00) node{$P$};
\draw (0.50,1.00) node{$P$};
\draw (5.50,0.00) node{$P$};
\end{tikzpicture}\\
            \begin{tikzpicture}[xscale=0.40,yscale=0.40]
\useasboundingbox (-0.5,-0.5) rectangle (6.5,2.5);
\draw[,] (3.00,1.00) -- (1.00,1.00);
\draw (2.15,1.10) -- (1.85,1.00);
\draw (2.15,0.90) -- (1.85,1.00);
\filldraw[fill=white] (3.00,1.00) ellipse (0.14cm and 0.14cm);
\draw (0.50,1.00) node{$O$};
\end{tikzpicture}&\begin{tikzpicture}[xscale=0.40,yscale=0.40]
\useasboundingbox (-0.5,-0.5) rectangle (6.5,2.5);
\draw[,] (3.00,1.00) -- (1.00,1.00);
\draw (1.85,1.10) -- (2.15,1.00);
\draw (1.85,0.90) -- (2.15,1.00);
\filldraw[fill=white] (3.00,1.00) ellipse (0.14cm and 0.14cm);
\draw (0.50,1.00) node{$P$};
\end{tikzpicture}\\
            \begin{tikzpicture}[xscale=0.40,yscale=0.40]
\useasboundingbox (-0.5,-0.5) rectangle (6.5,2.5);
\draw[] (1.00,2.00) -- (1.05,2.01) -- (1.11,2.01) -- (1.16,2.02) -- (1.21,2.03) -- (1.26,2.03) -- (1.32,2.04) -- (1.37,2.04) -- (1.42,2.05) -- (1.47,2.05) -- (1.52,2.05) -- (1.57,2.06) -- (1.62,2.06) -- (1.67,2.05) -- (1.72,2.05) -- (1.77,2.05) -- (1.82,2.04) -- (1.86,2.03) -- (1.91,2.03) -- (1.96,2.01) -- (2.00,2.00);
\draw[] (2.00,2.00) -- (2.06,1.98) -- (2.12,1.95) -- (2.18,1.92) -- (2.23,1.88) -- (2.29,1.85) -- (2.34,1.80) -- (2.40,1.76) -- (2.45,1.71) -- (2.50,1.66) -- (2.55,1.61) -- (2.59,1.56) -- (2.64,1.50) -- (2.69,1.44) -- (2.73,1.38) -- (2.78,1.32) -- (2.82,1.26) -- (2.87,1.19) -- (2.91,1.13) -- (2.96,1.06) -- (3.00,1.00);
\draw[] (3.00,1.00) -- (3.04,0.94) -- (3.09,0.87) -- (3.13,0.81) -- (3.18,0.74) -- (3.22,0.68) -- (3.27,0.62) -- (3.31,0.56) -- (3.36,0.50) -- (3.41,0.44) -- (3.45,0.39) -- (3.50,0.34) -- (3.55,0.29) -- (3.60,0.24) -- (3.66,0.20) -- (3.71,0.15) -- (3.77,0.12) -- (3.82,0.08) -- (3.88,0.05) -- (3.94,0.02) -- (4.00,0.00);
\draw[] (4.00,0.00) -- (4.04,-0.01) -- (4.09,-0.03) -- (4.14,-0.03) -- (4.18,-0.04) -- (4.23,-0.05) -- (4.28,-0.05) -- (4.33,-0.05) -- (4.38,-0.06) -- (4.43,-0.06) -- (4.48,-0.05) -- (4.53,-0.05) -- (4.58,-0.05) -- (4.63,-0.04) -- (4.68,-0.04) -- (4.74,-0.03) -- (4.79,-0.03) -- (4.84,-0.02) -- (4.89,-0.01) -- (4.95,-0.01) -- (5.00,0.00);
\draw (4.63,-0.15) -- (4.33,-0.06);
\draw (4.62,0.05) -- (4.33,-0.06);
\draw (1.63,1.96) -- (1.32,2.05);
\draw (1.62,2.16) -- (1.32,2.05);
\draw[] (5.00,2.00) -- (4.95,2.01) -- (4.89,2.01) -- (4.84,2.02) -- (4.79,2.03) -- (4.74,2.03) -- (4.68,2.04) -- (4.63,2.04) -- (4.58,2.05) -- (4.53,2.05) -- (4.48,2.05) -- (4.43,2.06) -- (4.38,2.06) -- (4.33,2.05) -- (4.28,2.05) -- (4.23,2.05) -- (4.18,2.04) -- (4.14,2.03) -- (4.09,2.03) -- (4.04,2.01) -- (4.00,2.00);
\draw[] (4.00,2.00) -- (3.94,1.98) -- (3.88,1.95) -- (3.82,1.92) -- (3.77,1.88) -- (3.71,1.85) -- (3.66,1.80) -- (3.60,1.76) -- (3.55,1.71) -- (3.50,1.66) -- (3.45,1.61) -- (3.41,1.56) -- (3.36,1.50) -- (3.31,1.44) -- (3.27,1.38) -- (3.22,1.32) -- (3.18,1.26) -- (3.13,1.19) -- (3.09,1.13) -- (3.04,1.06) -- (3.00,1.00);
\draw[] (3.00,1.00) -- (2.96,0.94) -- (2.91,0.87) -- (2.87,0.81) -- (2.82,0.74) -- (2.78,0.68) -- (2.73,0.62) -- (2.69,0.56) -- (2.64,0.50) -- (2.59,0.44) -- (2.55,0.39) -- (2.50,0.34) -- (2.45,0.29) -- (2.40,0.24) -- (2.34,0.20) -- (2.29,0.15) -- (2.23,0.12) -- (2.18,0.08) -- (2.12,0.05) -- (2.06,0.02) -- (2.00,0.00);
\draw[] (2.00,0.00) -- (1.96,-0.01) -- (1.91,-0.03) -- (1.86,-0.03) -- (1.82,-0.04) -- (1.77,-0.05) -- (1.72,-0.05) -- (1.67,-0.05) -- (1.62,-0.06) -- (1.57,-0.06) -- (1.52,-0.05) -- (1.47,-0.05) -- (1.42,-0.05) -- (1.37,-0.04) -- (1.32,-0.04) -- (1.26,-0.03) -- (1.21,-0.03) -- (1.16,-0.02) -- (1.11,-0.01) -- (1.05,-0.01) -- (1.00,0.00);
\draw (1.68,0.04) -- (1.37,-0.05);
\draw (1.67,-0.16) -- (1.37,-0.05);
\draw (4.68,2.15) -- (4.38,2.06);
\draw (4.67,1.95) -- (4.38,2.06);
\draw (0.50,2.00) node{$O$};
\draw (5.50,2.00) node{$O$};
\draw (0.50,0.00) node{$O$};
\draw (5.50,0.00) node{$O$};
\end{tikzpicture}&\begin{tikzpicture}[xscale=0.40,yscale=0.40]
\useasboundingbox (-0.5,-0.5) rectangle (6.5,2.5);
\draw[] (1.00,2.00) -- (1.05,2.01) -- (1.11,2.01) -- (1.16,2.02) -- (1.21,2.03) -- (1.26,2.03) -- (1.32,2.04) -- (1.37,2.04) -- (1.42,2.05) -- (1.47,2.05) -- (1.52,2.05) -- (1.57,2.06) -- (1.62,2.06) -- (1.67,2.05) -- (1.72,2.05) -- (1.77,2.05) -- (1.82,2.04) -- (1.86,2.03) -- (1.91,2.03) -- (1.96,2.01) -- (2.00,2.00);
\draw[] (2.00,2.00) -- (2.06,1.98) -- (2.12,1.95) -- (2.18,1.92) -- (2.23,1.88) -- (2.29,1.85) -- (2.34,1.80) -- (2.40,1.76) -- (2.45,1.71) -- (2.50,1.66) -- (2.55,1.61) -- (2.59,1.56) -- (2.64,1.50) -- (2.69,1.44) -- (2.73,1.38) -- (2.78,1.32) -- (2.82,1.26) -- (2.87,1.19) -- (2.91,1.13) -- (2.96,1.06) -- (3.00,1.00);
\draw[] (3.00,1.00) -- (3.04,0.94) -- (3.09,0.87) -- (3.13,0.81) -- (3.18,0.74) -- (3.22,0.68) -- (3.27,0.62) -- (3.31,0.56) -- (3.36,0.50) -- (3.41,0.44) -- (3.45,0.39) -- (3.50,0.34) -- (3.55,0.29) -- (3.60,0.24) -- (3.66,0.20) -- (3.71,0.15) -- (3.77,0.12) -- (3.82,0.08) -- (3.88,0.05) -- (3.94,0.02) -- (4.00,0.00);
\draw[] (4.00,0.00) -- (4.04,-0.01) -- (4.09,-0.03) -- (4.14,-0.03) -- (4.18,-0.04) -- (4.23,-0.05) -- (4.28,-0.05) -- (4.33,-0.05) -- (4.38,-0.06) -- (4.43,-0.06) -- (4.48,-0.05) -- (4.53,-0.05) -- (4.58,-0.05) -- (4.63,-0.04) -- (4.68,-0.04) -- (4.74,-0.03) -- (4.79,-0.03) -- (4.84,-0.02) -- (4.89,-0.01) -- (4.95,-0.01) -- (5.00,0.00);
\draw (4.33,-0.16) -- (4.63,-0.05);
\draw (4.32,0.04) -- (4.63,-0.05);
\draw (1.33,1.95) -- (1.62,2.06);
\draw (1.32,2.15) -- (1.62,2.06);
\draw[] (5.00,2.00) -- (4.95,2.01) -- (4.89,2.01) -- (4.84,2.02) -- (4.79,2.03) -- (4.74,2.03) -- (4.68,2.04) -- (4.63,2.04) -- (4.58,2.05) -- (4.53,2.05) -- (4.48,2.05) -- (4.43,2.06) -- (4.38,2.06) -- (4.33,2.05) -- (4.28,2.05) -- (4.23,2.05) -- (4.18,2.04) -- (4.14,2.03) -- (4.09,2.03) -- (4.04,2.01) -- (4.00,2.00);
\draw[] (4.00,2.00) -- (3.94,1.98) -- (3.88,1.95) -- (3.82,1.92) -- (3.77,1.88) -- (3.71,1.85) -- (3.66,1.80) -- (3.60,1.76) -- (3.55,1.71) -- (3.50,1.66) -- (3.45,1.61) -- (3.41,1.56) -- (3.36,1.50) -- (3.31,1.44) -- (3.27,1.38) -- (3.22,1.32) -- (3.18,1.26) -- (3.13,1.19) -- (3.09,1.13) -- (3.04,1.06) -- (3.00,1.00);
\draw[] (3.00,1.00) -- (2.96,0.94) -- (2.91,0.87) -- (2.87,0.81) -- (2.82,0.74) -- (2.78,0.68) -- (2.73,0.62) -- (2.69,0.56) -- (2.64,0.50) -- (2.59,0.44) -- (2.55,0.39) -- (2.50,0.34) -- (2.45,0.29) -- (2.40,0.24) -- (2.34,0.20) -- (2.29,0.15) -- (2.23,0.12) -- (2.18,0.08) -- (2.12,0.05) -- (2.06,0.02) -- (2.00,0.00);
\draw[] (2.00,0.00) -- (1.96,-0.01) -- (1.91,-0.03) -- (1.86,-0.03) -- (1.82,-0.04) -- (1.77,-0.05) -- (1.72,-0.05) -- (1.67,-0.05) -- (1.62,-0.06) -- (1.57,-0.06) -- (1.52,-0.05) -- (1.47,-0.05) -- (1.42,-0.05) -- (1.37,-0.04) -- (1.32,-0.04) -- (1.26,-0.03) -- (1.21,-0.03) -- (1.16,-0.02) -- (1.11,-0.01) -- (1.05,-0.01) -- (1.00,0.00);
\draw (1.38,0.05) -- (1.67,-0.06);
\draw (1.37,-0.15) -- (1.67,-0.06);
\draw (4.38,2.16) -- (4.68,2.05);
\draw (4.37,1.96) -- (4.68,2.05);
\draw (0.50,2.00) node{$P$};
\draw (5.50,2.00) node{$P$};
\draw (0.50,0.00) node{$P$};
\draw (5.50,0.00) node{$P$};
\end{tikzpicture}\\
            \begin{tikzpicture}[xscale=0.40,yscale=0.40]
\useasboundingbox (-0.5,-0.5) rectangle (6.5,2.5);
\draw[] (5.00,2.00) -- (4.93,1.95) -- (4.85,1.90) -- (4.78,1.85) -- (4.70,1.80) -- (4.63,1.75) -- (4.56,1.70) -- (4.50,1.65) -- (4.43,1.60) -- (4.37,1.55) -- (4.31,1.50) -- (4.26,1.45) -- (4.21,1.40) -- (4.16,1.35) -- (4.12,1.30) -- (4.09,1.25) -- (4.06,1.20) -- (4.03,1.15) -- (4.01,1.10) -- (4.00,1.05) -- (4.00,1.00);
\draw[] (4.00,1.00) -- (4.00,0.95) -- (4.01,0.90) -- (4.03,0.85) -- (4.06,0.80) -- (4.09,0.75) -- (4.12,0.70) -- (4.16,0.65) -- (4.21,0.60) -- (4.26,0.55) -- (4.31,0.50) -- (4.37,0.45) -- (4.43,0.40) -- (4.50,0.35) -- (4.56,0.30) -- (4.63,0.25) -- (4.70,0.20) -- (4.78,0.15) -- (4.85,0.10) -- (4.93,0.05) -- (5.00,0.00);
\draw (4.09,1.16) -- (4.01,0.85);
\draw (3.89,1.14) -- (4.01,0.85);
\draw (5.50,2.00) node{$O$};
\draw (5.50,0.00) node{$P$};
\end{tikzpicture}&\begin{tikzpicture}[xscale=0.40,yscale=0.40]
\useasboundingbox (-0.5,-0.5) rectangle (6.5,2.5);
\draw[] (1.00,2.00) -- (1.07,1.95) -- (1.15,1.90) -- (1.22,1.85) -- (1.30,1.80) -- (1.37,1.75) -- (1.44,1.70) -- (1.50,1.65) -- (1.57,1.60) -- (1.63,1.55) -- (1.69,1.50) -- (1.74,1.45) -- (1.79,1.40) -- (1.84,1.35) -- (1.88,1.30) -- (1.91,1.25) -- (1.94,1.20) -- (1.97,1.15) -- (1.99,1.10) -- (2.00,1.05) -- (2.00,1.00);
\draw[] (2.00,1.00) -- (2.00,0.95) -- (1.99,0.90) -- (1.97,0.85) -- (1.94,0.80) -- (1.91,0.75) -- (1.88,0.70) -- (1.84,0.65) -- (1.79,0.60) -- (1.74,0.55) -- (1.69,0.50) -- (1.63,0.45) -- (1.57,0.40) -- (1.50,0.35) -- (1.44,0.30) -- (1.37,0.25) -- (1.30,0.20) -- (1.22,0.15) -- (1.15,0.10) -- (1.07,0.05) -- (1.00,0.00);
\draw (2.11,1.14) -- (1.99,0.85);
\draw (1.91,1.16) -- (1.99,0.85);
\draw (0.50,2.00) node{$P$};
\draw (0.50,0.00) node{$O$};
\end{tikzpicture}
          \end{array}
          \\
          \begin{array}{c}
            \begin{tikzpicture}[xscale=0.40,yscale=0.40]
\useasboundingbox (-0.5,-0.5) rectangle (6.5,2.5);
\draw[] (1.00,2.00) -- (1.05,2.01) -- (1.11,2.01) -- (1.16,2.02) -- (1.21,2.03) -- (1.26,2.03) -- (1.32,2.04) -- (1.37,2.04) -- (1.42,2.05) -- (1.47,2.05) -- (1.52,2.05) -- (1.57,2.06) -- (1.62,2.06) -- (1.67,2.05) -- (1.72,2.05) -- (1.77,2.05) -- (1.82,2.04) -- (1.86,2.03) -- (1.91,2.03) -- (1.96,2.01) -- (2.00,2.00);
\draw[] (2.00,2.00) -- (2.06,1.98) -- (2.12,1.95) -- (2.18,1.92) -- (2.23,1.88) -- (2.29,1.85) -- (2.34,1.80) -- (2.40,1.76) -- (2.45,1.71) -- (2.50,1.66) -- (2.55,1.61) -- (2.59,1.56) -- (2.64,1.50) -- (2.69,1.44) -- (2.73,1.38) -- (2.78,1.32) -- (2.82,1.26) -- (2.87,1.19) -- (2.91,1.13) -- (2.96,1.06) -- (3.00,1.00);
\draw[] (3.00,1.00) -- (3.04,0.94) -- (3.09,0.87) -- (3.13,0.81) -- (3.18,0.74) -- (3.22,0.68) -- (3.27,0.62) -- (3.31,0.56) -- (3.36,0.50) -- (3.41,0.44) -- (3.45,0.39) -- (3.50,0.34) -- (3.55,0.29) -- (3.60,0.24) -- (3.66,0.20) -- (3.71,0.15) -- (3.77,0.12) -- (3.82,0.08) -- (3.88,0.05) -- (3.94,0.02) -- (4.00,0.00);
\draw[] (4.00,0.00) -- (4.04,-0.01) -- (4.09,-0.03) -- (4.14,-0.03) -- (4.18,-0.04) -- (4.23,-0.05) -- (4.28,-0.05) -- (4.33,-0.05) -- (4.38,-0.06) -- (4.43,-0.06) -- (4.48,-0.05) -- (4.53,-0.05) -- (4.58,-0.05) -- (4.63,-0.04) -- (4.68,-0.04) -- (4.74,-0.03) -- (4.79,-0.03) -- (4.84,-0.02) -- (4.89,-0.01) -- (4.95,-0.01) -- (5.00,0.00);
\draw (4.33,-0.16) -- (4.63,-0.05);
\draw (4.32,0.04) -- (4.63,-0.05);
\draw (1.33,1.95) -- (1.62,2.06);
\draw (1.32,2.15) -- (1.62,2.06);
\draw[] (5.00,2.00) -- (4.95,2.01) -- (4.89,2.01) -- (4.84,2.02) -- (4.79,2.03) -- (4.74,2.03) -- (4.68,2.04) -- (4.63,2.04) -- (4.58,2.05) -- (4.53,2.05) -- (4.48,2.05) -- (4.43,2.06) -- (4.38,2.06) -- (4.33,2.05) -- (4.28,2.05) -- (4.23,2.05) -- (4.18,2.04) -- (4.14,2.03) -- (4.09,2.03) -- (4.04,2.01) -- (4.00,2.00);
\draw[] (4.00,2.00) -- (3.94,1.98) -- (3.88,1.95) -- (3.82,1.92) -- (3.77,1.88) -- (3.71,1.85) -- (3.66,1.80) -- (3.60,1.76) -- (3.55,1.71) -- (3.50,1.66) -- (3.45,1.61) -- (3.41,1.56) -- (3.36,1.50) -- (3.31,1.44) -- (3.27,1.38) -- (3.22,1.32) -- (3.18,1.26) -- (3.13,1.19) -- (3.09,1.13) -- (3.04,1.06) -- (3.00,1.00);
\draw[] (3.00,1.00) -- (2.96,0.94) -- (2.91,0.87) -- (2.87,0.81) -- (2.82,0.74) -- (2.78,0.68) -- (2.73,0.62) -- (2.69,0.56) -- (2.64,0.50) -- (2.59,0.44) -- (2.55,0.39) -- (2.50,0.34) -- (2.45,0.29) -- (2.40,0.24) -- (2.34,0.20) -- (2.29,0.15) -- (2.23,0.12) -- (2.18,0.08) -- (2.12,0.05) -- (2.06,0.02) -- (2.00,0.00);
\draw[] (2.00,0.00) -- (1.96,-0.01) -- (1.91,-0.03) -- (1.86,-0.03) -- (1.82,-0.04) -- (1.77,-0.05) -- (1.72,-0.05) -- (1.67,-0.05) -- (1.62,-0.06) -- (1.57,-0.06) -- (1.52,-0.05) -- (1.47,-0.05) -- (1.42,-0.05) -- (1.37,-0.04) -- (1.32,-0.04) -- (1.26,-0.03) -- (1.21,-0.03) -- (1.16,-0.02) -- (1.11,-0.01) -- (1.05,-0.01) -- (1.00,0.00);
\draw (1.68,0.04) -- (1.37,-0.05);
\draw (1.67,-0.16) -- (1.37,-0.05);
\draw (4.68,2.15) -- (4.38,2.06);
\draw (4.67,1.95) -- (4.38,2.06);
\draw (0.50,2.00) node{$P$};
\draw (5.50,2.00) node{$O$};
\draw (0.50,0.00) node{$O$};
\draw (5.50,0.00) node{$P$};
\end{tikzpicture}
          \end{array}
        \end{array}
      \end{mymath}
    }
  }
  \vspace{-3ex}
  \caption{Generators of the strategies.}
  \vspace{-2ex}
  \label{fig:inno-gen}
\end{figure}
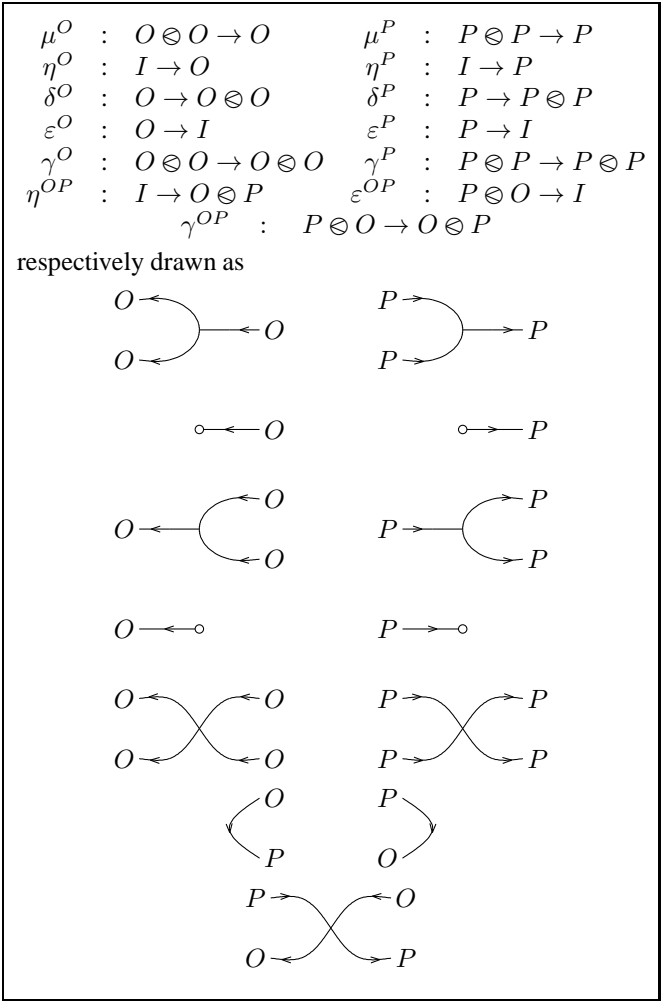

\pagebreak
\paragraph{A game semantics.}
A formula $A$ is interpreted as a filiform game~$\intp{A}$ by
\vspace\reduce
\[
\intp{P}=I
\qquad
\intp{\qforall x A}=O\before{}\intp{A}
\qquad
\intp{\qexists x A}=P\before{}\intp{A}
\vspace\reduce
\]
A cut-free proof $\pi:A\vdash B$ is interpreted as a strategy
$\sigma:\intp{A}\llimp\intp{B}$ whose causality partial order $\leq_\sigma$ is
defined as follows. For every Proponent move $P$ interpreting a quantifier
introduced by a rule which is either
\begin{mymath}
\inferrule{A[t/x]\vdash B}{\qforall x A\vdash B}{\lrule{$\forall$-L}}
\qqtor
\inferrule{A\vdash B[t/x]}{A\vdash \qexists x B}{\lrule{$\exists$-R}}
\end{mymath}
every Opponent move $O$ interpreting an universal quantification $\forall x$ on
the right-hand side of a sequent, or an existential quantification $\exists x$
on the left-hand side of a sequent, is such that $O\leq_\sigma P$ whenever the
variable $x$ is free in the term $t$.
For example, a proof
\begin{mymath}
\inferrule{
\inferrule{
\inferrule{
\inferrule{\null}
{P\vdash Q[t/z]}{\lrule{Ax}}
}
{P\vdash\qexists z Q}{\lrule{$\exists$-R}}
}
{\qexists y P\vdash\qexists z Q}{\lrule{$\exists$-L}}
}
{\qexists x{\qexists y P}\vdash\qexists z Q}{\lrule{$\exists$-L}}
\end{mymath}
is interpreted respectively by the strategies
\begin{myequation}
  \label{eq:ex-intp}
  \hspace{-2ex}
  \sstrid{strat_ex_xy}
  \ \ \ 
  \sstrid{strat_ex_x}
  \ \ \ 
  \sstrid{strat_ex_}
\end{myequation}
when the free variables of $t$ are $\{x,y\}$, $\{x\}$ or $\emptyset$.


\paragraph{An equational theory of strategies.}
\label{subsection:walking-inno}
We can now introduce the equational theory which will be shown to present the
category~$\Games$.

\begin{definition}
  \label{definition:innocent-strategies}
  The \emph{equational theory of strategies} is the equational theory $\eqth{G}$
  with two atomic types $O$ and $P$ and thirteen generators depicted in
  Figure~\ref{fig:inno-gen} such that
  \begin{itemize}
  \item the Opponent structure
    \begin{myequation}
      \label{eq:O-struct}
      (O,\mu^O,\eta^O,\delta^O,\varepsilon^O,\gamma^O)
    \end{myequation}
    is a bicommutative qualitative bialgebra,
  \item the Proponent structure
    $(P,\mu^P,\eta^P,\delta^P,\varepsilon^P,\gamma^P)$, as well as the
    morphism~$\gamma^{OP}$, are deduced from the Opponent structure
    \eqref{eq:O-struct} by composition with the duality morphisms~$\eta^{OP}$
    and~$\varepsilon^{OP}$ (for example
    \hbox{$\eta^P=(\varepsilon^O\before\id_P)\circ\eta^{OP}$}).
  \end{itemize}
\end{definition}
\noindent
We write $\catquot{\mathcal{G}}$ for the monoidal category generated by
$\eqth{G}$. It can be noticed that the generators $\mu^P$, $\eta^P$, $\delta^P$,
$\varepsilon^P$, $\gamma^P$ and~$\gamma^{OP}$ are superfluous in this
presentation (since they can be deduced from the Opponent structure and
duality). However, removing them would seriously complicate the proofs.

We can now proceed as in Section~\ref{section:presentation-rel} to show that the
theory~$\eqth{G}$ introduced in Definition~\ref{definition:innocent-strategies}
presents the category~$\Games$.
First, in the category $\Games$ with the monoidal structure induced by
$\before{}$, the objects $O$ and $P$ can be canonically equipped with thirteen
morphisms as shown in Figure~\ref{fig:inno-gen} in order to form a model of the
theory~$\eqth{G}$.

Conversely, we need to introduce a notion of canonical form for the morphisms
of~$\mathcal{G}$. \emph{Stairs} are defined similarly as before, but are now
constructed from the three kinds of polarized crossings $\gamma^O$, $\gamma^P$
and $\gamma^{OP}$ instead of simply~$\gamma$ in~\eqref{eq:ctx-S}.
The notion of \emph{precanonical form}~$\phi$ is now defined inductively as
shown in Figure~\ref{fig:precan-strat},
\begin{figure}[!t]
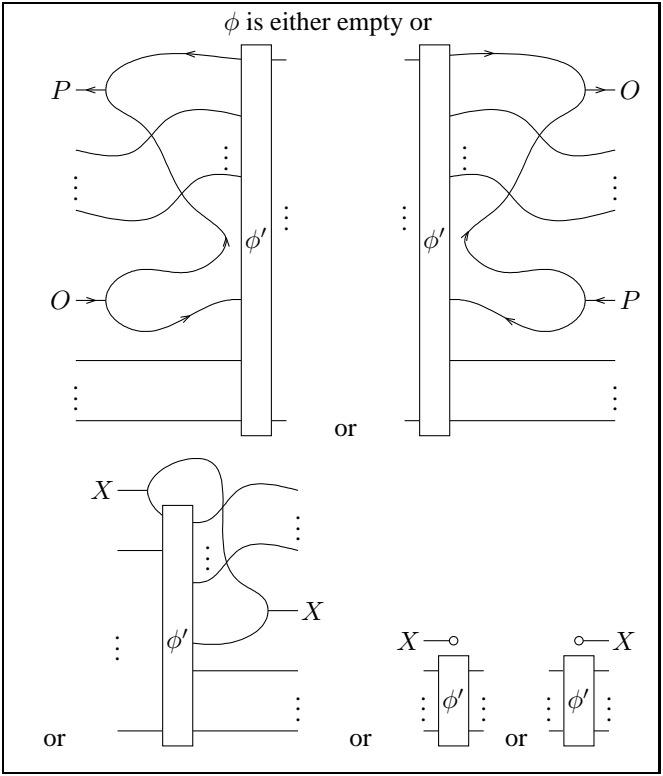

  \centering
  \fbox{
    \vbox{
      $\phi$ is either empty or
      \begin{mymath}
        \begin{array}{c}
          \sstrid{nf_adj}
          \qtor
          \sstrid{nf_coadj}
          \\
          \tor
          \sstrid{nf_mu}
          \tor
          \sstrid{nf_eps}
          \tor
          \sstrid{nf_eta}
        \end{array}
      \end{mymath}
    }
  }
  \caption{Precanonical forms for strategies.}
  \vspace{-4ex}
  \label{fig:precan-strat}
\end{figure}
where the object~$X$ is either~$O$ or~$P$ and $\phi'$ is a precanonical
form. These cases correspond respectively to the productions of the following
grammar
\begin{mymath}
\phi
\qgramdef
Z
\sgramor
A_i\phi
\sgramor
B_i\phi
\sgramor
W^X_i\phi
\sgramor
E^X\phi
\sgramor
H^X\phi
\end{mymath}
By induction on the size of morphisms, it can be shown that every morphism
of~$\mathcal{G}$ is equivalent to a precanonical form and a notion of canonical
form can be defined by adapting the rewriting system~\eqref{eq:mrel-cf-rs} into
a normalizing rewriting system for precanonical forms. Finally, a reasoning
similar to Lemma~\ref{lemma:mrel-bij} shows that canonical forms are in
bijection with morphisms of the category~$\Games$.

\begin{theorem}
  The monoidal category $\Games{}$ (with the~$\before{}$ tensor product) is
  presented by the equational theory $\eqth{G}$.
\end{theorem}
\vspace{-2.3ex}

\pagebreak
\ 
\vspace{1ex}

\noindent
As a direct consequence of this Theorem, we deduce that
\begin{enumerate}
\item the composite of two strategies, in the sense of
  Definition~\ref{def:strategy}, is itself a strategy (in particular, the
  acyclicity property is preserved by composition),
\item the strategies of~$\Games$ are definable (when the set~$\axioms$ of axioms
  is reasonable enough): it is enough to check that generators are definable --
  for example, the first case of~\eqref{eq:ex-intp} shows that~$\mu^P$ is
  definable.
\end{enumerate}

\section{Conclusion}
\vspace{-1ex}
We have constructed a game semantics for first-order propositional logic and
given a presentation of the category~$\Games$ of games and definable
strategies. This has revealed the essential structure of causality induced by
quantifiers as well as provided technical tools to show definability and
composition of strategies.

We consider this work much more as a starting point to bridge semantics and
algebra than as a final result. The methodology presented here seems to be very
general and many tracks remain to be explored.
First, we would like to extend the presentation to a game semantics for richer
logic systems, containing connectives (such as conjunction or
disjunction). Whilst we do not expect essential technical complications, this
case is much more difficult to grasp and manipulate, since a presentation of
such a semantics would have generators up to dimension 3 (one dimension is added
since games would be trees instead of lines) and corresponding diagrams would
now live in a 3-dimensional space.
It would also be interesting to know whether it is possible to orient the
equalities in the presentations in order to obtain strongly normalizing
rewriting systems for the algebraic structures described in the paper. Such
rewriting systems are given in~\cite{lafont:boolean-circuits} -- for monoids and
commutative monoids for example -- but finding a strongly normalizing rewriting
system presenting the theory of bialgebras is a difficult
problem~\cite{mimram:phd}.
Finally, some of the proofs (such as the proof of Lemma~\ref{lemma:mrel-precan})
are very repetitive and we believe that they could be mechanically checked or
automated. However, finding a good representation of morphisms, in order for a
program to be able to manipulate them, is a difficult task that we should
address in subsequent works.

\paragraph{Acknowledgements.}
I would like to thank my PhD supervisor Paul-André Melliès as well as John Baez,
Albert Burroni, Yves Guiraud, Martin Hyland, Yves Lafont and François Métayer
for the lively discussion we had, during which I learned so much.

\vspace{-1ex}
\bibliographystyle{latex8}
\bibliography{these}

\end{document}